\documentclass[11pt]{article}
\usepackage{tikz}
\usepackage{subcaption}
\usepackage{eurosym}
\usepackage{amsthm}
\newtheorem{assumption}{Assumption}
\newtheorem{proposition}{Proposition}
\newtheorem{remark}{Remark}
\usepackage{amsmath}
\usepackage{amssymb}
\usepackage{mathtools}
\usepackage{graphicx}
\usepackage{setspace}
\usepackage{sectsty}
\usepackage[authoryear]{natbib}
\usepackage[margin=1in]{geometry}
\usepackage{color}
\IfFileExists{bbm.sty}{\usepackage{bbm}}{}
\usepackage{booktabs}
\usepackage{pdflscape}
\usepackage{dcolumn}
\usepackage{threeparttable}
\IfFileExists{siunitx.sty}{\usepackage{siunitx}}{}%
\usepackage{array}
\usepackage{appendix}
\usepackage{enumitem}
\usepackage{placeins}
\usepackage{fancyhdr}
\usepackage{xr-hyper}
\usepackage[breaklinks=true,hidelinks]{hyperref}
\hypersetup{pdftitle={Tolerable Inflation, Intolerable Uncertainty},pdfauthor={Eric Vansteenberghe}}
\usepackage{cases}
\usepackage{etoolbox}
\usepackage{titlesec}
\usepackage{ragged2e}
\usepackage{lmodern,textcomp}
\usepackage[utf8]{inputenc}
\usepackage{longtable} %
\usepackage{import}

\theoremstyle{plain}                     %
  \newtheorem{theorem}{Theorem}          %
  \newtheorem{lemma}[theorem]{Lemma}
  \newtheorem{corollary}[theorem]{Corollary}

\theoremstyle{definition}                %

\newcommand{\Var}{\operatorname{Var}}
\newcommand{\E}{\mathbb{E}}
\newcommand{\Cov}{\operatorname{Cov}}
\newcommand{\Iset}{\mathcal{I}}

\newcommand{\lstar}{\lambda^{*}}
\newcommand{\lcb}{\hat\lambda^{CB}}
\newcommand{\lstare}{\lambda^{*}_{e}}   %
\newcommand{\lcbe}{\lambda^{CB}_{e}}    %
\newcommand{\lB}{\lambda_{B}}
\newcommand{\leff}{\lambda_{\mathrm{eff}}}
\newcommand{\de}{d_t^{e}}              %
\newcommand{\dpl}{(\de)_{+}}           %
\newcommand{\sw}{\sigma^{2}_{\omega}}
\newcommand{\ep}{\epsilon_{p}}         %

\newcommand{\rsub}{\rho_{\mathrm{sub}}}
\newcommand{\robj}{\rho_{\mathrm{obj}}}

\begin{document}

\hypersetup{pageanchor=false}
\begin{titlepage}
\renewcommand{\baselinestretch}{1}
	\title{
		\sc{
			\Large{
            Tolerable Inflation, Intolerable Uncertainty
			}%
		}%
	\thanks{I am grateful to Marco Airaudo, Jean Barth\'{e}lemy, Agn\`{e}s B\'{e}nassy-Qu\'{e}r\'{e}, Antoine Camous, Efrem Castelnuovo, Laurent Clerc, Olena Havrylchyk, Magali Marx, Adrian Penalver, Guido Traficante and Elias Wolf for helpful comments and discussions. I also thank participants at Banque de France seminars and at the JIR, NEUR and PIEC conferences and workshops. The views expressed in this paper are solely those of the author and do not necessarily reflect the views of his past, present, or future employers. Replication code for every figure and table, and for the measures themselves, is available at \url{https://github.com/skimeur/normalized-uncertainty}.}
		\\
	}
\author{
\textbf{Eric Vansteenberghe}%
\thanks{
Banque de France and Universit\'{e} Paris~1 Panth\'{e}on-Sorbonne. E-mail: \href{mailto:eric.vansteenberghe@banque-france.fr}{\nolinkurl{eric.vansteenberghe@banque-france.fr}}
}
}

	\date{November 2026}

	\maketitle

\abstract{\noindent
How much of an inflation deviation a central bank should tolerate is a well-defined optimum in the New Keynesian model, but it cannot be identified: optimal accommodation of supply shocks and an unwarranted drift of the effective inflation target are observationally equivalent. This impossibility has a cost. A public that cannot verify the bank's intention prices each expected point of overshoot, so uncertainty about inflation is flat below the announced target and rises linearly above it: the observational-equivalence effect. Professional forecasts, swap and option markets and realized inflation all show this pattern, with the kink at the announced number. A second and larger effect produces the same shape: inflation is more volatile when it is high, and in realized inflation that level effect is mostly energy. Both effects load on the distance from the target, so the data admit the observational-equivalence effect without isolating it, and its cost cannot be estimated with the data at hand. We therefore cannot say how optimal monetary policy should change on its account, but three policy conclusions survive: in the textbook model, not observing the slope of the Phillips curve and the composition of shocks never calls for tolerating less; the announced target remains worth having, because it compresses beliefs at the target; and what the central bank can do is monitor the point-years inflation spends above target, on which the effect is charged, and read inflation uncertainty net of the distance from the target. Normalized Uncertainty provides that reading by purging the distance component of measured uncertainty. In French loan-level data, raw inflation uncertainty is unrelated to corporate loan rates after controls, whereas a one-standard-deviation increase in Normalized Uncertainty is associated with rates 89 basis points higher. The correction also reveals a negative association with growth in cross-country data and alters euro-area time-series results.
\label{abstract}
}



\thispagestyle{empty}
\end{titlepage}
\hypersetup{pageanchor=true}
    \section{Introduction}
\label{sec:intro}

The Federal Reserve has ``no tolerance for persistently elevated inflation'' \citep{warsh2026testimony}: however right and credible its Chairman, it cannot demonstrate it has no such tolerance---no central bank can---and this impossibility is paid for in uncertainty. A central bank does not, and should not, react to every inflation deviation in full. When a deviation is judged transitory and supply-driven, looking through part of it is optimal; when it signals a persistent, demand-driven drift, leaning against it is. How much of a given deviation to tolerate therefore depends on its nature---how persistent it is---and on the nature of the shock behind it---how much of it is supply. It is a deliberate, state-contingent choice, not a loss of resolve. We write that choice into the bank's target: it responds to the deviation of inflation from an \emph{effective} target---the announced number plus the part of the current impulse it chooses to accommodate---rather than from the announced number itself. To an observer of the policy rule this reads as a response to only a fraction of the inflation gap, and two distinct objects meet in that product. The \emph{rule}: how aggressively this bank responds to inflation as a standing matter. The \emph{dial}: how much of \emph{this} deviation the bank's target absorbs, an episode-specific, state-contingent setting. Only the product is ever observed---a hawk tolerating half the gap responds exactly like a dove tolerating nothing---yet the two configurations mean opposite things: deliberate, reversible look-through in the first case, plain softness in the second. %

What is \emph{not} optimal, and cannot be made so, is knowing how much to tolerate. The optimal accommodation is a well-defined object---the supply share of the shock times a structural look-through intensity---and it is unidentifiable, not for want of data but by construction. The intensity rests on the Phillips slope, which seventy years of research have not pinned; the dial is revealed only through a policy-rule coefficient whose own identification is in dispute; and the supply share is latent, in real time and ex post, and carries most of the optimum's irreducible variance. The consequence is a statement in the oldest register of the identification literature: a tolerated gap that optimal policy prescribes and a gap that unwarranted drift would produce are \emph{observationally equivalent} in the sense of \citet{sargent1976observational}---two readings of one history that no accumulation of equilibrium data separates.\footnote{The phrase is \citeauthor{koopmans1949identification}'s: ``We shall call two structures $S$ and $S'$ (observationally) equivalent (or indistinguishable) if the two conditional distributions of endogenous variables generated by $S$ and $S'$ are identical for all possible values of the exogenous variables'' \citep[p.~133]{koopmans1949identification}. It entered macroeconomics under the title of \citet{sargent1976observational}, applied there not to structural parameters but to the invariance assumption itself.} The verdict on a central bank's stance can be argued; it cannot be settled.

The contract behind these episodes is asymmetric. The announced target is a point in the \emph{level}---two per cent---but an interval in \emph{time}: the return is promised ``over the medium term,'' a window the ECB calls ``flexible because the appropriate monetary policy response to a deviation of inflation from the target depends on the origin, magnitude and persistence of the deviation'' \citep{ecb2021orientation}. The vagueness leaves the bank free to tolerate what it is optimal to tolerate, and it transfers a decision to the public: an agent who is never told \emph{when} the deviation will close must decide for herself how large an overshoot is worth acting on. Trusting the anchor is a bet its holder can never verify---the deviation may be tolerance at its optimum or drift wearing its clothes---so above a private threshold an agent withdraws the benefit of the doubt and re-bases her plans. The bank's vagueness about timing purchases its flexibility; the public's benefit of the doubt prices it.

That impossibility is not free. In the model, as expected inflation rises above the target, the unverifiable alternative---that the anchor itself is drifting---gains adherents one private threshold at a time, and the variance of inflation forecasts rises with the overshoot: flat below the target, linear above it, kinked at the announced number.

The data show a relationship of this form in four sources: professional density forecasts in the euro area and the United States, the daily inflation-swap market, and inflation options net of their risk premium. In the ECB Survey of Professional Forecasters (Figure~\ref{fig:lawintro}) the distance of expected inflation from the $2\%$ target explains $70\%$ of the variation in forecast variance, with the below-target side flat. Two effects produce it, and the data do not separate them. The first is a \emph{level effect}: inflation is more volatile when it is high. It is the larger of the two and the one the literature has long documented \citep{okun1971mirage,friedman1977nobel}, more strongly at horizons of several years than at short ones \citep{ball1990inflation}. In realized euro-area inflation it is carried by energy prices, which account for between half and two-thirds of the rise in variance above the target. The second is the \emph{observational-equivalence effect} (OEE), which this paper derives: because the public cannot tell an optimal tolerance from an unwarranted drift, it prices each expected point of overshoot at a distrust that no history revises, and uncertainty rises with the overshoot for that reason alone. In beliefs, the rise in forecast variance above the target survives controls for the volatility of energy prices and for the recent record of realized variance. Both effects load on the same observable, the distance from the target, and we call their sum the \emph{distance component} of measured uncertainty.

Two consequences follow. Because the OEE cannot be isolated from the level effect, its size is not known, and we cannot say by how much optimal monetary policy should depart from the textbook prescription: the theory says that the cost is zero at the target and rises with the expected overshoot, and no more. And because the distance component sits inside every measure of inflation uncertainty built from forecast densities or market prices, it has to be purged before such measures are put to academic or policy use: what cannot be separated can still be removed, and the correction changes the answers those measures have been giving. Nothing here argues against announcing a number. An announced objective coordinates expectations, compresses their dispersion and holds them through shocks \citep{levin2004macroeconomic,gurkaynak2010does}; what the paper adds is an entry on the other side of the ledger.

\paragraph{Contributions.} The paper makes three contributions.

\emph{First, the observational-equivalence effect.} Optimal tolerance and unwarranted drift are observationally equivalent, and we derive what a public that cannot tell them apart does. Agents hold dispersed, private thresholds on the overshoot they will tolerate before withdrawing the benefit of the doubt; their re-based decisions make inflation harder to forecast for everyone; and forecast variance is flat below the announced target and rises linearly with the expected overshoot above it, with the kink at the \emph{announced} number, since a deviation must be a deviation \emph{from} something. The effect is a proposition of theory. It shares its regressor with the level effect, so no regression isolates it, and the United States, where a number was announced inside the sample, shows only that it could be there (Appendix~\ref{app:us}).

\emph{Second, the purge.} Normalized Uncertainty (NU) removes the distance component of measured uncertainty, the level effect and the OEE together, without needing their shares. Under one stated assumption---that no charge the distance carries falls with the overshoot---the fitted component bounds the OEE from above. In the pricing of bank credit, on $71{,}069$ French overdraft facilities, raw inflation uncertainty carries no information once the policy stance and bank effects are held, while the purged measure is associated with loan rates nearly nine-tenths of a percentage point higher per standard deviation: what lenders price is the uncertainty the anchor does not explain. In the cross-country growth regressions of \citet{barro1995inflation} the correction answers his own suspicion: what realized inflation variability appears to carry belongs mostly to the \emph{level} of inflation---beside the level, raw variability adds nothing---yet once that level component is purged, residual inflation uncertainty weighs on growth in its own right, strongly enough to turn his null into a result. And in a vector autoregression in the tradition of \citet{baker2016measuring}, the industrial-production decline that follows a raw uncertainty shock before 2020 shrinks by about a quarter under the correction and all but disappears once the distance of inflation from the anchor is held, while the purged measure keeps three-fifths of its own response: the level and residual uncertainty are two channels, not one.

\emph{Third, the observational-equivalence cost, and why it cannot be estimated.} The OEE is a statement about beliefs. What it costs the economy is the part of the variance of realized inflation, and of any effect on its level, that the doubters' re-based decisions add: the \emph{observational-equivalence cost}. Realized euro-area inflation is more volatile above the target than below, but that is the level effect, mostly energy. Measuring the cost requires the reverse link, from inflation uncertainty to the inflation gap, and Section~\ref{sec:reverse} shows that it cannot be estimated with the data now available. The tool is the level-dependent GARCH of the term-structure literature, applied to realized inflation, the market's expected inflation and the expected policy rate. It replicates the law in all three conditional variances, with the market's slope matching the survey's, about $0.9$ against $0.86$ per point of expected overshoot, and it delivers a \emph{predictive} link: the market's expected gap predicts next month's headline inflation about one for one, partly through energy, and core inflation with a smaller coefficient. Its uncertainty terms are fragile, and a GARCH measure of uncertainty could not have identified them in any case: its innovation is a function of past news, so an effect of uncertainty and the propagation of large shocks are observationally equivalent, the reverse link's own version of the paper's first result. The external variation identification would need is not there either. Oil-specific uncertainty, purged of a global factor, is exogenous to the euro area but moves none of eight euro-area outcomes, and the global factor that does co-move with euro-area uncertainty also moves rates, inflation compensation and oil prices, so it fails the exclusion restriction. The OEE cannot be isolated but can be purged; its cost can be neither, and it has no current estimate. That is what makes ``living with it'' a result, and the policy section ends on a laboratory design that could measure the reverse link: a spread of the inflation distribution that leaves its mean in place.

\subsection{Related literature}
\label{subsec:related_literature}

\paragraph{Observational equivalence; policy rules, tolerance, and learning.}
The impossibility we establish is of the kind \citet{sargent1976observational} named: two structures that differ only in what would survive a change of regime fit one equilibrium history identically. For the tolerance parameter the equivalence rests on two coefficients, the Phillips slope and the rule's inflation response, and each has been the object of a seventy-year identification program in which every generation re-attributed what its predecessor had claimed to identify; the companion review \citep{vansteenberghe2026seventy} traces both. On the curve, a bank that stabilizes inflation erases the covariation the slope regression requires, so the estimable curve flattens exactly when policy is good \citep{mcleay2020optimal}, who also set out the routes---supply-shock controls, instrumental variables, regional variation---by which the slope can be recovered, each at the price of the restrictions it maintains. On the rule, the inflation coefficient may not be identified from equilibrium data even in principle, because the Taylor principle works through off-equilibrium threats that equilibrium histories never display \citep{cochrane2011determinacy}. The model adds no ingredient of its own to what the policy-rule tradition agrees on; it is a summary of it. The discretionary optimum offsets demand shocks and accommodates a fraction of a cost-push shock \citep{clarida1999science,woodford2003interest}; the rule reveals only the composite response, whose strength governs the persistence of the gap \citep{benati2008investigating,conrad2012explaining}; the tolerated share of a deviation is the classic delegation margin \citep{rogoff1985degree,lohmann1992optimal,orphanides2002opportunistic,hofmann2026targeted}; and the public that must read a persistent overshoot is the learner of the imperfect-credibility tradition \citep{erceg2003imperfect}, with uncertainty about the bank's type priced in forecast variance as in \citet{ball1992does} and chosen opacity about the objective as in \citet{cukierman1986ambiguity}. Together they leave the public unable to tell an intended overshoot from a drift (Lemma~\ref{prop:oe}), a standard non-identification once the pieces are put together; what the paper adds is its price in forecast uncertainty (Proposition~\ref{prop:law}). In the imperfect-credibility family, households who cannot split the policy residual into a shifted objective and a transitory shock read accommodation as drift and learn its level; here the split itself is what the inflation history cannot deliver, so on that margin nothing is learned and the benefit of the doubt is extended, finitely.

\paragraph{Attention thresholds, and the law.}
Closest to the mechanism is the attention literature, in which households begin to track inflation only once it is high enough to be worth tracking: \citet{pfauti2026attention} estimates the threshold near $4\%$ on realized inflation, finds that the updating coefficient on forecast errors roughly doubles above it, and shows that identified supply shocks pass through about twice as strongly when attention is high; \citet{weber2025tell} confirm across countries and randomized information treatments that attention rises with inflation; and \citet{korenok2023inflation} document the cutoff form directly. That literature works on the first moment and estimates its threshold freely; this paper works on the second and indexes the thresholds to the announced number. Its contribution is what dispersed, target-indexed thresholds imply for the variance of beliefs: a population whose private thresholds bind only while expected inflation exceeds the announced target, and whose members widen their forecasts by the re-basing they expect from others, generates a forecast variance that is flat below the target, linear above it, and kinked at the number---the law of Proposition~\ref{prop:law}. The beauty-contest weight on others' plans \citep{morris2002social,angeletos2007efficient} is what makes the unidentified attribution consequential without making it learnable; the baseline the law sits on is the risk agents have learned from experience \citep{evans2001learning,malmendier2016learning}, and the scales at which they price that baseline and the contagion they expect from others are the documented overconfidence of density forecasts \citep{giordani2003inflation,clements2014forecast}. The mechanism itself is assembled from standard parts. The doubting share is the threshold model of collective behaviour \citep{granovetter1978threshold}, read as an adjustment hazard \citep{caballero1993microeconomic}; its multiplier is the social multiplier of discrete choice with interactions \citep{brock2001discrete,glaeser2003social}, and its coordination motive the strategic complementarity of price and wage setting \citep{cooper1988coordinating,ball1990real}. The cost that keeps most plans on the announced number is the cost of observing and processing \citep{reis2006inattentive,alvarez2011optimal}. What an agent stands to lose is the loss people name when asked why they dislike inflation, the erosion of what they are owed \citep{shiller1997people,stantcheva2024why}, and what she can recover by re-basing is what separates a relative-price loss someone must bear from a drift nobody must \citep{blanchard2007real,lorenzoni2023inflation}; that nominal terms are rarely revised down \citep{akerlof1996macroeconomics} is why the law has one arm.

\paragraph{Measuring inflation uncertainty, and what it does to the economy.}
Inflation uncertainty is not observed, and the literature proxies it: news-based indices \citep{baker2016measuring}, the expected volatility of the unforecastable component of many series \citep{jurado2015measuring}, conditional variances from ARCH models of inflation \citep{engle1982autoregressive} and from trend--cycle decompositions with time-varying volatility \citep{stock2007has}, and---closest to this paper---survey histograms, which since \citet{zarnowitz1987consensus} separate average individual uncertainty from disagreement \citep{lahiri2010measuring}. Every one of these measures co-moves with the level of inflation, a confound noted since \citet{gordon1971steady}, \citet{okun1971mirage} and \citet{cukierman1979differential}; the one survey-level test of the link on euro-area data found none \citep{abel2016measurement}, on a sample that, we show, almost never left the flat side of the target (Appendix~\ref{app:abel_replication}). On the consequences, the tradition that runs from \citet{friedman1977nobel} associates inflation uncertainty with lower activity \citep{haa2024economic}, with lower growth forecasts at the level of the individual forecaster \citep{paloviita2014inflation}, and, in the cross-country regressions of \citet{barro1995inflation}, with nothing at all once the level of inflation is held---a null Barro himself attributed to realized variability failing to measure uncertainty. The paper's contribution to both literatures is a single claim: the co-movement of measured uncertainty with the inflation level is not a nuisance to be controlled away, but a component with a structure of its own: the level effect the literature knows and, in our theory, the OEE. A regression using a raw uncertainty measure therefore conflates two distinct objects---what the distance from the announced target predicts and genuine imprecision about the outlook. Purging the first (Section~\ref{sec:nu}) moves the literature's answers in both directions. Where the level was doing the work, the raw coefficient falls: the activity response to an uncertainty shock in a vector autoregression in the tradition of \citet{baker2016measuring}, estimated before 2020, shrinks by about a quarter, and all but vanishes once the distance from the anchor is held. Where genuine imprecision matters, a coefficient appears that the raw measure hid: Barro's null becomes a growth effect significant at five percent, and loan pricing, indifferent to raw uncertainty once the policy stance is held, prices the purged measure at nearly nine-tenths of a percentage point per standard deviation.

\paragraph{Uncertainty in the mean, and identifying uncertainty shocks.}
The reverse link is the old question of GARCH-in-mean models of inflation, which test whether conditional inflation uncertainty raises inflation or lowers output \citep{engle1987estimating,grier1998inflation,grier2000effects}, the direction \citet{cukierman1986ambiguity} predict, opposite to the one from inflation to its uncertainty that runs from \citet{friedman1977nobel} to \citet{ball1992does}; level-dependent volatility enters them additively for inflation \citep{brunner1993higher} and multiplicatively for the short rate \citep{brenner1996another,koedijk1997dynamics}. Our use of these models is to show what they cannot deliver: a GARCH uncertainty has no innovation of its own, so its effects are predictive. The causal literature on uncertainty shocks seeks external variation instead---disasters \citep{baker2024using}, the price of gold \citep{piffer2018identifying}---and stresses that much measured uncertainty is an endogenous response to the cycle \citep{ludvigson2021uncertainty} and that recursive identification is invalid \citep{kilian2025macroeconomic}. Oil supply news identifies first-moment oil shocks \citep{kanzig2021macroeconomic}; we ask whether oil \emph{uncertainty}, purged of global uncertainty, identifies a second-moment shock for the euro area, and find that it does not.
    \section{The model and the purge}
\label{sec:results}

Two effects load on the distance of expected inflation from the announced target. The level effect, by which inflation is more volatile when it is high, needs no model. This section gives the theory of the other, the observational-equivalence effect (OEE; Section~\ref{sec:model}), and builds the purge that removes both from measured uncertainty (Section~\ref{sec:purge}). The theory is stated for beliefs, and Section~\ref{sec:evidence} does not test it against the level effect: the two cannot be separated in field data. What the OEE costs in realized inflation is the question of Section~\ref{sec:reverse}.

\subsection{The model}
\label{sec:model}

\subsubsection{Optimal tolerance and observational equivalence}\label{sec:model-oe}

\paragraph{The model.}
The real block is the textbook New-Keynesian economy (\citealp[ch.~3]{Gali2015}; \citealp[ch.~4]{woodford2003interest}), closed by a policy rule that responds to the deviation of inflation from the bank's \emph{effective} target:
\begin{align}
  d_t &= \beta\,\E_t d_{t+1} + \kappa\,y_t + u_t,
        && \text{(forward-looking Phillips curve)} \label{eq:nkpc}\\
  y_t &= \E_t y_{t+1} - \sigma\bigl(\hat\imath_t - \E_t d_{t+1}\bigr) + v_t,
        && \text{(dynamic IS)} \label{eq:is}\\
  \hat\imath_t &= \varphi_\pi\,(\pi_t-\pi^{*}_t) + \varphi_y\,y_t,
        && \text{(policy rule, } \varphi_\pi>1\text{)} \label{eq:taylor}\\
  \pi^{*}_t &= \bar{\pi} + m_t + \xi_t,
        && \text{(effective target)} \label{eq:target}
\end{align}
with $d_t=\pi_t-\bar{\pi}$ the realized inflation gap to the \emph{announced} target $\bar{\pi}$, $y_t$ the output gap, $\hat\imath_t\equiv i_t-r^\star-\bar{\pi}$ the policy-rate gap, $\kappa>0$ the Phillips slope, $\sigma>0$ the intertemporal elasticity, $u_t$ a cost-push shock, $v_t$ a demand disturbance, and $\ep>1$ the demand elasticity that sets the welfare weight $\vartheta=\kappa/\ep$ in the micro-founded loss; $\E_t[\cdot]$ and $\Var_t(\cdot)$ are conditional on the public information set $\Iset_t$. The bank's effective target \eqref{eq:target}---the inflation it aims to deliver in period $t$---is the announced number, plus the component $m_t$ of the current cost-push impulse it chooses to accommodate (its \emph{intention}), plus a drift $\xi_t$ it does not announce nor control; the rule \eqref{eq:taylor} responds to the deviation from the effective target. The bank does not observe the cost-push itself: it forms an estimate $\hat u_t\equiv\E[u_t\mid\Iset^{CB}_t]$ from indicators $\Iset^{CB}_t$ that carry no information on the drift $\xi_t$, and accommodates that estimate, $m_t\equiv\lstar\hat u_t$. By certainty equivalence the static optimum applies to the estimate \citep{svensson2003indicator}: choosing the intended gap $\hat d$ to minimize $\E[\hat d^{2}+\vartheta((\hat d-u_t)/\kappa)^2\mid\Iset^{CB}_t]$ gives $\hat d=\lstar\hat u_t$. This is the static benchmark: with persistent shocks and a public that filters, the forward term and the public's beliefs become state variables, and that policy is not solved here. Lemma~\ref{prop:oe} uses only the laws of $m_t$ and $\xi_t$, not the optimality of $m_t$. $\lstar$ is what this paper calls the bank's \emph{optimal tolerance} (derivation and calibration: Appendix~\ref{app:optlambda}):
\begin{equation}\label{eq:lstar-model}
   \lstar=\frac{1}{1+\ep\kappa}=\frac{\vartheta}{\vartheta+\kappa^2}\in(0,1)
\end{equation}
(with a persistent impulse the share is $\iota(\rho_u)$ of \eqref{eq:iota_rho}, which falls with persistence).

\begin{assumption}[Impulse and composition]\label{ass:comp}
Each period an inflation impulse $\Pi_t$ arrives, measured as the inflation gap
it would produce on impact under full look-through, and divides into a supply (cost-push)
component $u_t\equiv s_t\Pi_t$, which shifts the Phillips curve, and a demand
component $(1-s_t)\Pi_t\equiv\kappa'v_t$, the inflation produced by the shift
$v_t/\sigma$ in the natural rate, with $\kappa'\equiv\kappa/(1+\sigma\varphi_y)$
(Appendix~\ref{app:optlambda}). The pair
$(s_t,\Pi_t)$ follows a stationary process, with $\sigma_u^2\equiv\Var(u_t)$ and
autocorrelation $r_k\equiv\Cov(u_t,u_{t-k})/\sigma_u^2$; the benchmark is an
AR(1), $r_k=\rho_u^{|k|}$ (the cost-push shock's own persistence). The public observes realized inflation but not the
composition $(s_t,\Pi_t)$, and in particular does not observe $u_t$.
\end{assumption}

\begin{assumption}[Preferences and conduct]\label{ass:loss}
The bank has the micro-founded quadratic loss
$\tfrac12 d^2+\tfrac12\vartheta y^2$, with $d$ the inflation gap, $y$ the output
gap, and weight $\vartheta=\kappa/\ep$, where $\kappa$ is the Phillips slope and
$\ep$ the elasticity of substitution across varieties. It acts under discretion,
re-optimizing each period with no commitment; in the static benchmark there is no
endogenous state.
\end{assumption}

\begin{assumption}[The error, the drift, and what the public can tell apart]\label{ass:innov}
Measured inflation is the effective target up to an error,
$\pi_t=\pi^{*}_t+\varepsilon_t$, with $\varepsilon_t$ mean zero and
unpredictable---the part of inflation nothing forecasts, measurement noise in the index
together with Ball's control error \citep{ball1992does,cukierman1986ambiguity}.
The relation is the high-gain limit of the rule \eqref{eq:taylor}, which removes disturbances to the rate, not this
error. The bank's \emph{effective} target
\eqref{eq:target} carries, beside the announced number and the intention, a drift $\xi_t$ the bank does not announce and may or may not incur;
$\mu_t\equiv\pi^{*}_t-\bar{\pi}=m_t+\xi_t$ is its unannounced part. \emph{(a)} The public observes
the gap $d_t=\mu_t+\varepsilon_t$ and none of its components. \emph{(b)} No a priori asymmetry: the
public reads the unannounced part either as all intention ($\xi\equiv0$, \emph{tolerance}) or as all
drift ($m\equiv0$), with weight $q\in(0,1)$ on the first, and which reading holds is independent of
the path of $(\mu_t,\varepsilon_t)$: whatever path an intention could take, a drift could take too.
\end{assumption}

\noindent The error $\varepsilon_t$ is the part of measured inflation that says
nothing about the target: noise in the index---sampling and quality adjustment,
seasonal factors, one-off tax and administered-price changes, later
revisions---and the bank's imperfect control of prices within the period. The
bank reads the same index and, knowing its intention but not the drift, cannot
tell the drift from the error in a print either. The error is what makes the
target latent, since no one reads $\pi^{*}_t$ off a single number; it plays no
part in the identification result, which holds with $\varepsilon_t\equiv0$
(Lemma~\ref{prop:oe}).

Timing within a period is as in Ball's protocol, with the composition in the
place of his political lottery: expectations are set first and are rational;
the bank then sets its target on its estimate---in the static benchmark the map
\eqref{eq:lstar-model} applies to the estimate whatever the law of the
shock---and the rule enforces it against realized inflation, so the part of the
cost-push not built into the target, $u_t-\lstar\hat u_t$, moves output; measured inflation is observed by all, and the composition by no one: not by the public, and not by
the bank (Wall~III, Appendix~\ref{app:walls}). The unannounced part $\mu_t$ is the latent persistent component
of the gap: the public filters it from the history of
the gap, but not its split between accommodation and drift. In the true structure the two components share the mean and
may differ in persistence: with precise estimates the tolerated component
inherits the supply shock's own $\rho_u$, while a filtered estimate also carries
the filter's own dynamics; the drift carries $\rho_\xi$. What the public does not know is which persistence is whose.

\begin{lemma}[Observational equivalence and the unidentified intention]\label{prop:oe}
Let Assumption~\ref{ass:innov} hold, and let the public information $\Iset_t$ be the
history $d^{t}=\{d_s\}_{s\le t}$, joined by any signal whose law depends on $(m,\xi)$
only through their sum---the published survey, or the re-basing decisions of others
(Section~\ref{sec:model-law}). At every date $t$:
\begin{enumerate}[label=\emph{(\roman*)},leftmargin=2.2em,itemsep=2pt]
\item \emph{Observational equivalence.} The two readings give $\Iset_t$ and the future
path of the gap the same law. A payoff that depends on the path of inflation alone has
the same expectation under both: in particular, the predictive variance of $d_{t+1}$
contains no term for the label.
\item \emph{No learning.} The weight on tolerance stays at $q$, whatever the history.
\item \emph{The expected drift.} If $\mu_{t+1}$ is integrable, the public expects the
fraction $1-q$ of the overshoot to be drift,
\begin{equation}\label{eq:drift-moments}
\E_t\bigl[\xi_{t+1}\bigr]=(1-q)\,\de ,\qquad \de\equiv\E_t\,d_{t+1} .
\end{equation}
\end{enumerate}
\end{lemma}

\noindent The proof is short (Appendix~\ref{app:oe}): by Assumption~\ref{ass:innov}(b) the label is independent of the overshoot and of everything the public observes, so no history moves the weight on it, and the expected drift is the prior share of the expected overshoot. The tools are standard: (i) is observational equivalence in the sense of \citet{sargent1976observational}, and (ii) is the non-updating of beliefs about a parameter the likelihood does not contain \citep{poirier1998revising,moon2012bayesian}; a public that also entertained structures in which the bank both tolerates and drifts would meet the label switching of mixture models, with the same conclusion (Appendix~\ref{app:oe}). The data are informative---the size and persistence of the overshoot are learned in the usual way---but they never reveal \emph{whose} it is. Where \citet{cukierman1986ambiguity} let the public separate persistent shifts in the policymaker's emphasis from transitory control error by signal extraction, here both candidates are persistent, and only an a priori asymmetry---knowing that drifts outlast supply shocks, or that their innovations carry a distinctive signature---would label them; the assumption imposes none. What is new is what the unidentified label costs. The weight $q$, the public's \emph{confidence} that the bank is following its intention, is a constant of the regime, and the next step prices it; an even-handed public, $q=\tfrac12$, books half of each expected point of overshoot as drift.

\subsubsection{From observational equivalence to the uncertainty law}\label{sec:model-law}

Lemma~\ref{prop:oe} leaves the public unable to learn whether an overshoot was intended. This step turns that
margin into a law of forecast uncertainty, in four links. Optimal tolerance puts expected inflation above
the announced number. The public cannot tell that overshoot from drift, and no history changes the weight
it gives each reading. The label cannot widen a forecast directly, because inflation has the same law
under both readings; it can only change a payoff that depends on the drift itself, and through that payoff
an action, the re-basing of nominal plans. The re-based plans scatter, and their scatter reaches inflation
on the innovation side. No agent infers from a larger overshoot that drift has become more likely: the
regime fixes that probability. What rises with the overshoot is the amount at stake if it is drift. Agents
share the public information set $\Iset_t$ and the expected gap $\de\equiv\E_td_{t+1}$, and differ only in
their private stakes and costs.

\paragraph{Plans, the stake, and the threshold.} Agent $i$ holds a nominal claim and sets its terms: the
wage she asks, the price she posts, the loan rate she quotes. She keeps the terms on the announced number
($a_{i,t}=0$), or she \emph{re-bases} them ($a_{i,t}=1$): she indexes, shortens or reprices the claim on her
own reading of the effective target, at a flow cost $k_i\ge0$ of attention and processing that inattentive
agents avoid \citep{reis2006inattentive,alvarez2011optimal}. The share who re-base, the \emph{doubters}, is
$S_t\equiv\int_0^1a_{i,t}\,di$.

What does the anchored plan cost her? Inflation above the number the claim was written on erodes its real
value, point for point. That erosion is the loss people name first when asked why they dislike inflation
\citep{shiller1997people,stantcheva2024why}, but by Lemma~\ref{prop:oe}(i) it is the same under both
readings and cannot be what confidence prices. What differs across the readings is whether re-basing
recovers it, and the block of Section~\ref{sec:model-oe} says why. A tolerated point of overshoot is the part of a
cost-push shock the bank lets into prices, and it comes with an output gap of $-\ep$ per point
($y_t=-\ep m_t$, the discretionary optimum). Behind it is a real loss that someone bears whatever the plan,
so a re-based claim only passes it on: the second round the bank leans against
\citep{blanchard2007real,lorenzoni2023inflation}. A drifted point has no real loss behind it and comes with
an output gap of $+1/\kappa$ per point ($d_t=\kappa y_t$), so a re-based claim recovers it in full. Neither
the agent nor the bank observes the output gap in real time \citep{orphanides2002unreliability}, and what
she later observes of her own demand mixes it with the demand disturbance, so the stake is priced at the
prior and its realization does not reveal the label (Appendix~\ref{app:oe}, Scope). Per period, the anchored plan
costs her, relative to the re-based one,
\begin{equation}\label{eq:exposure}
c_i\,z_t,\qquad z_t\;\equiv\;(1-q)\,\de\;+\;\zeta\,S_t ,\qquad \zeta\ge0 ,
\end{equation}
with $c_i>0$ the size of her claim. The first term is the drift she expects, \eqref{eq:drift-moments}. It
is signed: an expected undershoot is a windfall to a nominal claim. The second is the coordination motive.
An anchored plan falls out of line with the re-based plans of others (relative wages, relative prices,
competitors' loan rates) at a cost that rises with their share: the strategic complementarity of price and
wage setting \citep{ball1990real,cooper1988coordinating}, the weight on others of a beauty contest
\citep{morris2002social}. It acts through actions, not the mean: others'
re-basing leaves $\E_td_{t+1}$ where it was. She re-bases when the stake covers the cost,
\begin{equation}\label{eq:rebase-choice}
a_{i,t}=\mathbf 1\{z_t\ge\tau_i\},\qquad \tau_i\equiv\frac{k_i}{c_i}\ \ge0 .
\end{equation}

\begin{assumption}[Claims and thresholds]\label{ass:rebase}
The pairs $(c_i,k_i)$ are private, with $c_i>0$ and $k_i\ge0$. The thresholds $\tau_i$ are drawn
independently of the aggregate path from a distribution $G$ on $[0,\infty)$ with $G(0)=0$ and a continuous
density $g$ with $g(0)>0$. They are dispersed relative to the coordination motive:
$\zeta\,\sup_{\tau}g(\tau)<1$.
\end{assumption}

\noindent The stake is first order in the drift because the real value of a claim is monotone in the price
level. Its holder is not at an interior optimum with respect to it, so no envelope argument makes the
loss second order, as it would for a price a firm has just chosen \citep{akerlof1985near}. A positive
density at zero says that some agents have so much at stake, relative to their cost of re-basing, that
they act on the first tenth of a point; a support that began at some $\underline\tau>0$ would put a dead
band at the target instead of a kink. The thresholds are private and dispersed because the announcement
fixes a number without fixing a date: ``over the medium term'' offers no common date on which to
coordinate.\footnote{The formulation is the institution's own: the Treaty ``does not give a precise
definition of what is meant by price stability,'' and the Governing Council ``considers that price
stability is best maintained by aiming for 2\% inflation over the medium term''
\citep{ecb2024pricestab}---precise in the level, open in time, in one sentence.} Dispersion is the
weak-interaction condition of discrete choice with social interactions \citep{brock2001discrete}, and the
condition under which heterogeneity rules out multiplicity \citep{herrendorf2000ruling}.

\paragraph{The doubting share.} Integrating \eqref{eq:rebase-choice} over the thresholds, and using that
$G$ vanishes on negative exposures,
\begin{equation}\label{eq:share-fixed-point}
S_t=G\bigl((1-q)\,\de+\zeta\,S_t\bigr)=G\bigl((1-q)\,\dpl+\zeta\,S_t\bigr).
\end{equation}
The positive part is a result: a windfall crosses no threshold. Equation~\eqref{eq:share-fixed-point} is the
threshold model of collective behaviour \citep{granovetter1978threshold} with an exogenous push, the
expected drift; read as a function of the gap, the share is an adjustment hazard
\citep{caballero1993microeconomic}. Write
\begin{equation}\label{eq:RU}
R_U\equiv\zeta\,g(0)<1
\end{equation}
for the \emph{reproduction number} of doubt at the target: the further mass of agents that one marginal
mass of doubters brings across their thresholds. $1/(1-R_U)$ is the social multiplier
\citep{glaeser2003social}.

\paragraph{From doubt to uncertainty.} Doubters act. Having left the common anchor, each sets wages, prices
and loan rates on her own reading of an effective target nobody announces, and those readings scatter.

\begin{assumption}[Contagion and the forecasters' scale]\label{ass:contagion}
The doubters' re-based decisions add to next period's inflation an innovation $Q_{t+1}$ with $\E_tQ_{t+1}=0$
and $\Var_t(Q_{t+1})=\upsilon S_t$, $\upsilon>0$. It is uncorrelated with the fundamental innovation
$\varepsilon^0_{t+1}\equiv d_{t+1}-\de-Q_{t+1}$, the news about the target and the rest of the error, whose conditional variance is a constant $\sigma_0^2$. Forecaster $i$ states the predictive variance
\[
\sigma^2_{i,t}=a_i+\gamma_0\,\hat\sigma^2_{0,t}+\gamma\,\upsilon\,S_t ,\qquad \gamma_0,\gamma>0 ,
\]
with $a_i\ge0$ of cross-sectional mean $a$, and
$\hat\sigma^2_{0,t}=\tfrac1t\sum_{s\le t}(d_s-\E_{s-1}d_s)^2$ the fundamental variance learned from the raw
surprises with a decreasing gain.
\end{assumption}

\noindent A period is the survey's horizon, one year. The doubters reach measured
inflation through the error $\varepsilon_{t+1}$ of Assumption~\ref{ass:innov}, so $d_{t+1}=\de+\varepsilon^0_{t+1}+Q_{t+1}$ and
\begin{equation}\label{eq:realized-variance-share}
\Var_t(d_{t+1})=\sigma_0^2+\upsilon\,S_t .
\end{equation}
The variance is linear in the share because independent re-basing errors add variance in proportion to the
number of agents who make them, the shape granular price and wage setting gives
\citep{gabaix2011granular}; a common error would add it in proportion to the square of the share and leave
no kink (Appendix~\ref{app:clocks}). Granularity fixes the shape, not the size: $\upsilon$ is a reduced-form loading,
the variance one doubter's decisions add to inflation by whatever route. The two scales are the
overconfidence of density forecasts \citep{giordani2003inflation,clements2014forecast}, applied separately
to the risk forecasters have learned and to the contagion they expect; the learned baseline is
least-squares learning applied to a second moment \citep{evans2001learning}, which lets lived volatility
shape uncertainty \citep{malmendier2016learning}. The doubters' decisions depend on the two components of
the target only through their sum and through the prior, never through the label, so Lemma~\ref{prop:oe}
survives them. An equilibrium is the rule around the effective target, rational filtering of $\mu_t$,
confidence held at its prior, the decisions \eqref{eq:rebase-choice} and their share
\eqref{eq:share-fixed-point}, and the variances of Assumption~\ref{ass:contagion}, consistent each period.

\paragraph{The law.}
\begin{proposition}[The uncertainty law]\label{prop:law}
Let Assumptions~\ref{ass:innov}, \ref{ass:rebase} and~\ref{ass:contagion} hold.
\begin{enumerate}[label=\emph{(\roman*)},leftmargin=2.2em]
\item \emph{One equilibrium.} For every expected gap there is a unique share $S(\de)\in[0,1]$ solving
\eqref{eq:share-fixed-point}. It is continuous and nondecreasing in $\de$.
\item \emph{No doubt at or below the target.} $S(\de)=0$ for every $\de\le0$.
\item \emph{The local share.} Above the target,
\begin{equation}\label{eq:share-local}
S(\de)=\frac{(1-q)\,g(0)}{1-R_U}\,\de+o(\de)\qquad(\de\downarrow0).
\end{equation}
\item \emph{Each forecaster.} Forecaster $i$ states
$\sigma^2_{i,t}=a_i+\gamma_0\hat\sigma^2_{0,t}+\gamma\upsilon\,S(\de)$. Her density widens with the re-basing
of others; her own decision appears nowhere in it.
\item \emph{The law.} The average individual variance $W_t\equiv\overline{\sigma^2_{i,t}}=V(\de)$ is
$a+\gamma_0\hat\sigma^2_{0,t}+\gamma\upsilon\,S(\de)$: constant below the announced target, kinked at it, and
to first order above it
\begin{equation}\label{eq:law}
V(\de)=a+\gamma_0\,\hat\sigma^2_{0,t}+b_+\dpl,\qquad
b_+=\gamma\,\upsilon\,(1-q)\,\frac{g(0)}{1-R_U}>0 .
\end{equation}
\end{enumerate}
\end{proposition}

\noindent The slope is a product,
\begin{equation}\label{eq:slope-decomposition}
b_+=\underbrace{\upsilon}_{\text{contagion}}\times
\underbrace{\gamma}_{\text{subjective scale}}\times
\underbrace{(1-q)}_{\text{distrust}}\times
\underbrace{g(0)}_{\text{marginal agents}}\times
\underbrace{\frac{1}{1-R_U}}_{\text{social multiplier}} :
\end{equation}
the variance each doubter adds to inflation, the scale at which forecasters price it, the share of each
point the public books as drift, the density of agents at their threshold, and their effect on one
another. Distrust prices every point of overshoot, and since no history revises it, the slope is a
property of the regime: the same from one episode to the next, and moved only by what moves the regime.
The law depends on \emph{how far} above the target the public expects inflation, not on how long it has
been there. When the expected overshoot closes, the same thresholds stop binding and agents return to the
anchored plan along the law; no separate law of re-entry is needed. The multiplier works through actions,
not inference: a larger doubting share does not make drift more probable, and the actions it generates do
not reveal which component produced the gap. The proposition is proved in Appendix~\ref{app:clocks}.

\paragraph{What pins the kink.} Four conditions put a kink at the announced number, and each is a property
the evidence of Section~\ref{sec:evidence} can reject.
\begin{enumerate}[label=\textbf{K\arabic*},leftmargin=2.2em]
\item \emph{A first-order stake indexed to the announced number.} A stake on the squared miss alone is
second order near the target and gives a dead band there, the near-rationality of
\citet{akerlof1985near}. So does the drift's variance on its own, which is the lottery of
\citet{ball1992does} (Remark~\ref{rem:ball}, Appendix~\ref{app:clocks}).
\item \emph{A positive density of thresholds at zero.} Thresholds bounded away from zero would move the
break above the announced number, where attention-threshold accounts estimate theirs
\citep{korenok2023inflation,pfauti2026attention}.
\item \emph{Contagion linear in the share.} A common re-basing error adds variance in the square of the
share, which has no kink.
\item \emph{A signed stake.} The loss of a nominal claim is monotone in the drift, so the law has one arm.
A symmetric stake would give two. The steep arm lies on the side where the plan loses: above the number
for a nominal claim, and below the benchmark for a claim on real growth, where the placebo of Section~\ref{sec:law}
finds it. Without an announced number the reference points of nominal plans are dispersed and the kinks
average out.
\end{enumerate}
The one-sidedness is thus the stake's. It is not an implication of Lemma~\ref{prop:oe}, which treats an
overshoot and an undershoot alike, nor of the bank's objective: the ECB considers ``negative and positive
deviations from our 2\% inflation target to be equally undesirable'' \citep{ecb2024pricestab}. The other
side of each contract loses on an undershoot. It carries no arm in the model because the claimant quotes
the terms, and nominal terms are rarely revised down \citep{akerlof1996macroeconomics}, as negotiated wages show (Section~\ref{sec:wages}); the lower arm of
Section~\ref{sec:evidence} measures how far that holds. Nor could a mean-side mechanism have produced \eqref{eq:law}: any
policy that moves only the conditional mean leaves the conditional variance unchanged (Lemmas~\ref{lem:ce}
and~\ref{lem:invariance}, Appendix~\ref{app:meanside}), so a gap-dependent variance must reach the innovation side, and contagion does. We
keep the first-order form \eqref{eq:law}. The exact law takes the shape of $G$ away from the target, and
the survey's open upper interval cannot tell a bend in $G$ from its own truncation, and we leave that
shape open. Uniqueness rests on heterogeneity, not on the noisy private
signals of global games \citep{carlsson1993global,morris1998unique}, so it holds whatever the public
knows and no public signal can select another equilibrium \citep{angeletos2006crises,hellwig2006self}.

\paragraph{What tolerance costs in uncertainty.} The law prices the expected gap whatever its source, and
tolerance is one of its two sources.
\begin{corollary}[The uncertainty bill of tolerance]\label{cor:bill}
Let the assumptions of Proposition~\ref{prop:law} hold, and write $\hat u^e_t\equiv\E[\hat u_{t+1}\mid\Iset_t]$ and
$\xi^e_t\equiv\E[\xi_{t+1}\mid\Iset_t]$ for the expected components under the true structure, which the
public cannot tell apart. Then $\de=\lstar\hat u^e_t+\xi^e_t$ and, to first order at the announced target
and at given expected components,
\begin{equation}\label{eq:bill}
\frac{\partial W_t}{\partial\lstar}=b_+\,\hat u^e_t\,\mathbf 1\{\de>0\}.
\end{equation}
For a bank that never drifts, $\de=\lstar\hat u^e_t$ and $W_t=a+\gamma_0\hat\sigma^2_{0,t}+b_+\,\lstar(\hat u^e_t)_+$ exactly as
\eqref{eq:law} holds: at a given expected impulse the bill is linear in the tolerance, and it is the bill a drifting bank would pay
for the same expected gap.
\end{corollary}

\noindent A bank that tolerates optimally and never drifts pays $b_+$ per expected point it lets through,
because the public cannot tell that point from drift and prices it at its unrevisable distrust. The corollary is a statement about forecasters' variance; what it costs in realized inflation is the question of Section~\ref{sec:reverse}.

\begin{remark}[Disagreement]\label{rem:disagree}
The model's agents share the expected gap. Forecasters' point forecasts differ, because they filter the
level and persistence of $\mu_t$ with different models and priors
\citep{mankiw2003disagreement,patton2010why,lahiri2010measuring}. That disagreement lies on the identified
margin and is label-free (Lemma~\ref{prop:oe}(i)): it is about where inflation is going, not about whose
overshoot it is, and nothing in it depends on the side of the target. What the model adds to it is
one-sided, the scatter of the doubters' own readings, in proportion to $S_t$. Disagreement can therefore
rise on both sides of the target while each forecaster's variance rises only above it, and the gap
between the two arms of disagreement is the model's part of it.
\end{remark}

\paragraph{The OEE and the level effect.} The law of Proposition~\ref{prop:law} is a prediction about beliefs, and we call it the \emph{observational-equivalence effect}, because without the equivalence of Lemma~\ref{prop:oe} it would not exist. Were the intention observable, a tolerated overshoot would carry no stake of trust, and a drifting one would be re-based on by everyone at the same known number, which adds no scatter: the slope $b_+$ of \eqref{eq:law} would be zero, through the stake in the first case and through $\upsilon$ in the second, whatever the thresholds. No communication removes it, since it is charged by the equivalence and not by the miss; communication may raise the thresholds, but it does not remove the agents for whom the first tenth of a point is worth acting on. And it is charged only against an announced number, which a deviation is a deviation \emph{from}. The distance from the target predicts a second thing, a \emph{level effect}: inflation is more volatile when it is high, whatever the public believes, and energy, which drives most overshoots, is more volatile than the rest of the basket. The literature has long documented it \citep{okun1971mirage,friedman1977nobel}. \citet{ball1990inflation} find that, within a country over time, the level of inflation raises the variance of permanent shocks much more than that of temporary ones, so that the relationship is stronger at horizons of several years than at the one-year horizon used here; across countries it holds at short horizons as well. They read it as high inflation making policy less stable, a reading about beliefs on policy that is closer to the OEE than to the level effect. Section~\ref{sec:evidence} measures a relationship of the form \eqref{eq:law} in four sources: the ECB Survey of Professional Forecasters' density forecasts, the daily euro-area inflation-linked swap market, the US Survey of Professional Forecasters, and euro-area inflation options once the variance risk premium is removed. Both effects are functions of the same observed distance, and no variation in the data moves one without the other, so the OEE has no coefficient of its own: the slope a regression returns is their sum, with whatever risk premia and learning about the rule load on the same regressor. We call that sum the \emph{distance component} of measured uncertainty. The data admit the OEE and can neither isolate nor reject it. What can be done is to remove the sum.

\subsection{The purge: Normalized Uncertainty}
\label{sec:purge}

\paragraph{The purge.} Neither effect, the level effect nor the OEE, can be removed from the economy. The distance
component can, and should, be removed from our measurements---otherwise every
empirical use of inflation uncertainty confounds the imprecision of beliefs
with priced vigilance. Let $\Var_t$ denote \emph{measured uncertainty} at round
$t$: the average individual conditional variance of the one-year-ahead
inflation gap---the density variance in the surveys, the implied variance in
the market---the object the model prices as $V(\de)$ in \eqref{eq:law}. We keep one multiplicative model for the whole distance component as the variance-stabilizing choice: genuine uncertainty scales the
structural component rather than sitting beside it, and the measurement model
is multiplicative,
\begin{equation}\label{eq:var_decomp}
  \Var_t\;=\;V_{\mathrm{struct}}(\de)\;\times\;U_t,
  \qquad
  V_{\mathrm{struct}}(\de)\;=\;a^{s}+b^{s}_-\,(-\de)_+ + b^{s}_+\,(\de)_+ ,
\end{equation}
where $V_{\mathrm{struct}}$ is what the distance from the anchor predicts and
$U_t$---equal to one when beliefs are exactly as imprecise as the distance
predicts---is what it does not. Its coefficients
$a^{s},b^{s}_+,b^{s}_-$ are the \emph{fitted} ones---what a regression of measured
uncertainty on the distance recovers. The contagion variance the mechanism creates is unlearnable by Lemma~\ref{prop:oe}, hence a floor
that no data, no attention and no institutional design can erode; and under
one further assumption---that
every other charge the distance carries loads on the arms with the same sign,
so that nothing the distance predicts \emph{sharpens} beliefs on net---the
fitted coefficients bound the law's from above,
\begin{equation}\label{eq:struct-bounds}
  a^{s}\ge a,\qquad b^{s}_+\ge b_+,\qquad b^{s}_-\ge 0 ,
\end{equation}
an inequality between population projections. \emph{Normalized Uncertainty} is
motivated by the \emph{variance-stabilizing transformation}, the classical
repair for dispersion that moves with an observable state
\citep{anscombe1948transformation,box1964analysis,nelder1972generalized}: for a
variance linear in the state---the signature of Proposition~\ref{prop:law}---the stabilizer is the square root. Dividing the
measured variance by the envelope at the observed distance removes the
dependence exactly under the multiplicative model \eqref{eq:var_decomp}, and
that form is the measure,
\begin{equation}\label{eq:NU_definition}
  \mathrm{NU}_t\;\equiv\; \sqrt{ \frac{\Var_t}{V_{\mathrm{struct}}(\de)}},
\end{equation}
the measured standard deviation in units of the standard deviation the
fitted envelope predicts at the observed distance. $\mathrm{NU}$ is distance-normalized dispersion. Under the multiplicative model it reads $\mathrm{NU}_t=\sqrt{U_t}$: one when beliefs are exactly as imprecise as the distance predicts, with departures from one measuring imprecision that the level of expectations does
not explain---in a word, \emph{residual} uncertainty, the part of measured
uncertainty the anchor does not explain. The derivation and the checks are in \citet{vansteenberghe2026uncertain}, which also compares $\mathrm{NU}$ with a residual-based alternative, the log-ratio of measured variance to a multiplicative envelope estimated by Poisson pseudo-maximum likelihood. We keep $\mathrm{NU}$ for two reasons. First, it needs no estimation beyond the envelope \eqref{eq:var_decomp} itself. Second, it stays defined and well behaved when a reported standard deviation is zero or close to it, where the log-ratio is not. Section~\ref{sec:nu} takes the measure to the data.

\paragraph{Individual uncertainty and disagreements.} Where forecaster-level densities exist---the ECB survey---measured uncertainty is $W_t$, the average of the individual predictive variances, and the law is estimated on it (Section~\ref{sec:fitlaw}). Where the data are aggregate---the US survey, the market rungs, the cross-country panel, the activity regressions---the object is the total mixture variance $T_t=W_t+D_t$, the individual variance plus the between-forecaster dispersion of means, or its realized counterpart, because that is what aggregate distributions and outcomes deliver. The law carries across (Table~\ref{tab:armsbysource}): the individual component is most of the total at the target---$0.383$ of its $0.420$ intercept---and lends it its shape, so a purge built on the aggregate law removes the same component. 

\paragraph{An upper bound.} Normalized Uncertainty removes
\emph{at least} the OEE. Under the same-sign assumption the envelope the
denominator removes contains the contagion law's projection with room to
spare---removing it in full still leaves $a^{s}+b^{s}_-(-\de)_+\ge0$---an
inequality between population projections that holds wherever
\eqref{eq:struct-bounds} does and owes nothing to the sample. Equality would
require the denominator to contain nothing but that effect; it contains the level effect and
further charges, each with its own seat in the fitted law. The intercept
$a^{s}$ collects what widens forecasts at every distance---the dispersion present when no one doubts, the fundamental risk forecasters have learned, past contagion included (Assumption~\ref{ass:contagion}, Remark~\ref{rem:scar}), the filtering floor on an effective
target the public never observes (Assumption~\ref{ass:innov}), and the
ARCH-type persistence by which measured variance is high because it was high
\citep{engle1982autoregressive}---and is estimated freely rather than read off the model. Whatever loads below the
target leaves through $b^{s}_-(-\de)_+$, no part of an effect charged only
above; the measured $b^{s}_-$ is indistinguishable from zero. On the upper
arm the slope $b_+$ of \eqref{eq:law} shares its regressor with the level effect and with two further companions of
the same sign: credibility-erosion learning, whose forecast-error variance
tracks the squared surprises that above-target episodes inflate
\citep{orphanides2005inflation}, and state-dependent caution about the
rule's effective response as the gap opens. Erosion has little room here: it
lives on a perceived anchor that moves, and every location estimate of
Section~\ref{sec:evidence} puts the kink at the announced number through the
surge itself, so its seat in $b^{s}_+$ is plausibly near empty for a modern,
target-anchored central bank. Each companion loads with a nonnegative weight
under the assumption, and none can touch the contagion the public cannot learn away (Lemma~\ref{prop:oe}), so $b_+\le b^{s}_+$: the
correction is an upper bound on the OEE twice over---across the components it removes, and inside the arm that effect shares. The assumption excludes a charge that \emph{falls} with the overshoot. The natural candidate is attention---forecasters sharpening their distributions as the stakes rise---and the envelope's fit cannot detect it: a falling charge inside $U_t$ is absorbed into the fitted slope and leaves the binned means of the ratio at one.
    \section{The uncertainty law, measured}
\label{sec:evidence}

Proposition~\ref{prop:law} predicts the variance of inflation beliefs flat below the announced target, rising linearly above it, kinked at the announced number itself---and \eqref{eq:struct-bounds} states exactly how much of it a regression can recover: the fitted coefficients of the envelope \eqref{eq:var_decomp} are \emph{upper bounds} on the law's under the same-sign assumption of Section~\ref{sec:purge}: the contagion variance the mechanism prices is unlearnable by Lemma~\ref{prop:oe}, and every other charge the distance carries is assumed to load on the same regressor with the same sign. This section is the measurement built on that reading. It establishes the law's \emph{form} and the location of its kink on the highest-frequency data available, locates the kink at the announced number, then fits the law source by source---survey, realized inflation, options---and takes the correction the law defines to three literatures.

What this section measures is the \emph{distance component} of Section~\ref{sec:model-law}: the sum of the level effect and of the OEE. No estimate below separates the two, and none is read as evidence for the second; the form, the location of the kink and the sizes are facts about the sum.

Two distances recur, and each estimate is stated on the object on which it is estimated. The \emph{realized} gap $d_t=\pi_t-\bar\pi$ is the deviation of realized inflation from the announced target; the \emph{forecast} gap $\de$ is the consensus one-year-ahead forecast net of the same target, the survey counterpart of the expected gap $\E_t d_{t+1}$ defined in Section~\ref{sec:results}. The ECB Survey of Professional Forecasters is the primary survey sample, used at the quarterly round level from 1999Q1 to the round conducted in 2026Q2, and the envelope is estimated throughout in the arms form of \eqref{eq:var_decomp}---the two sides of the target allowed to differ---for the reason this section begins by establishing: it is the form the data select, before it is the shape the law implies. One warning governs every number below. \emph{Slopes are not comparable across samples whose regressors have different ranges.}

\subsection{The law in the data: form, anchor, sources}
\label{sec:law}

\paragraph{The form, from the daily market.}
The section opens on the source with the power to settle the law's form: the daily euro-area inflation-swap market---public, high-frequency, and priced continuously through every episode the survey sees four times a year. Uncertainty is read as the three-month (sixty-three-day) realized variance of the daily two-year swap rate and the distance as the rate's own gap to the target: $5{,}421$ trading days, of which roughly $1{,}600$ sit above it. Three properties of the law are decided on this sample. \emph{Linear, not quadratic}: within the overshoot arm a linear specification beats a quadratic one in levels, in logs and by maximum likelihood, and the state-dependence itself is not volatility clustering in disguise, surviving inside a GARCH specification. \emph{One-sided, not symmetric}: fitted in the arms form, the above-target arm carries $b_+=+1.606$ while the below-target arm is $-0.078$---twenty times smaller and of the wrong sign for a symmetric reading.

\paragraph{The anchor, estimated twice.}
\emph{Kinked at the announced number}: estimated on the market's monthly non-overlapping realized variances with the kink left \emph{free}, the two sides meet at $\hat m = 1.96$, with a bootstrap median of $1.85$. We call this point the \emph{anchor}: the number at which the market's pricing of inflation uncertainty changes regime, estimated from the data rather than imposed on them. The survey gives the same answer at the quarterly frequency of the headline fit. Profiling the breakpoint $c$ over the central ninety percent of the consensus forecast's range, $1.10$ to $2.57$, the residual sum of squares is minimised at $\hat c=1.90$ (Figure~\ref{fig:kinkloc}), with a $95\%$ set of $[1.76,\,2.00]$; freeing the breakpoint buys $3.3\%$ of the residual sum of squares. The location is not an artifact of how the survey's open bins are closed: with the bins entering the likelihood as censoring intervals the free kink is $1.87$ ($[1.76,\,1.99]$), and on an interquartile measure that never touches the tails it is $1.83$ ($[1.68,\,1.95]$). All estimates sit at or slightly below two, where the announced formulation itself sat for most of the sample: ``below $2\%$'' from 1998, ``below, but close to, $2\%$'' from 2003, a symmetric $2\%$ only since the July 2021 review \citep{ECB2021strategy,ecb2024pricestab}. The break sits at an institutional number. That is what the OEE predicts, and a level effect that sets in near the same number would put it there too: the estimates locate the kink and do not tell the two apart.

\paragraph{Asymmetry alone is not the finding.} On the survey's GDP growth densities, with the distance measured from potential growth, the arms are asymmetric in the opposite direction ($b_-=2.64$, $t=6.1$; $b_+=1.04$, $t=3.2$), as ordinary downside risk predicts. Only inflation, which has an announced target, shows a side of the benchmark on which uncertainty does not respond to distance at all. The contrast points to the target; it does not separate the two effects.

\paragraph{The United States: the announcement effect.}
The United States is the one sample in which a numerical objective is announced inside the data, in January 2012. It shows that the OEE could be there; Appendix~\ref{app:us} reports it.

\paragraph{The survey law.}\label{sec:fitlaw}
On the ECB-SPF survey panel, regressing the average individual predictive variance $W_t$ on the forecast gap in the arms form gives
\begin{equation}\label{eq:lawsurvey}
\Var_t \;=\; \underset{(0.063)}{0.383} \;+\; \underset{(0.097)}{0.004}\,(-\de)_+ \;+\; \underset{(0.070)}{0.855}\,(\de)_+ ,
\qquad R^2=0.70,\quad n=109 ,
\end{equation}
with HAC standard errors in parentheses. The fitted right-hand side is the estimated envelope $\widehat V_{\mathrm{struct}}$ of \eqref{eq:var_decomp}.\footnote{Within the model the intercept is the dispersion present when no one doubts and the fundamental risk forecasters have learned, $a+\gamma_0\hat\sigma^2_{0,t}$ (Proposition~\ref{prop:law}); the measured intercept also collects every other source of baseline dispersion, which is why it is estimated freely.} The overshoot arm is estimated at twelve standard errors from zero; the below-target arm is four thousandths, with a $t$ of $0.05$. The symmetric restriction $b^{s}_-=b^{s}_+$ is rejected at $p<10^{-4}$.

\paragraph{Individual uncertainty and disagreements: measured with surveys.}
The survey reports each forecaster's own density, so the variance of the consensus distribution---the total mixture variance $T_t$---separates exactly into the average of the individual variances, $W_t$, and the between-forecaster dispersion of the density means, $D_t$ \citep{sastry2026disagreement}. The law is a statement about $W_t$: the model prices the population's benefit of the doubt, and as it is withdrawn the doubters' decisions widen every forecaster's predictive density, whatever her own reading (Proposition~\ref{prop:law}(iv)). Estimating the arms on each component separates what the model attributes to the mechanism from its population echo:
\begin{center}
\begin{tabular}{lccccc}
\toprule
 & $b_-$ & $t$ & $b_+$ & $t$ & $R^2$ \\
\midrule
$W_t$, average individual variance (the law) & $+0.004$ & $0.05$ & $+0.855$ & $12.22$ & $0.705$ \\
$D_t$, disagreement                          & $+0.113$ & $2.36$ & $+1.028$ & $8.93$  & $0.719$ \\
$T_t$, total mixture variance                & $+0.117$ & $1.01$ & $+1.882$ & $13.69$ & $0.789$ \\
\bottomrule
\end{tabular}
\end{center}
Below the target the average individual variance does not move---four thousandths per point, $t=0.05$---while disagreement rises measurably, $+0.113$ per point with $t=2.4$. That is consistent with the threshold decision and harder to square with a common-shock account of the kink: were uncertainty to rise with the distance from the target in both directions because inflation is harder to forecast away from it, individual densities would widen on both sides; instead, below the target forecasters drift apart about where inflation is going while each remains as confident as before, because below the target both readings prescribe the anchored plan and no one re-bases (Remark~\ref{rem:disagree}). Above the target both components rise, and the anchor is what is common to them. The total inherits the individual shape---its lower arm is $+0.117$ with $t=1.0$, its upper arm $1.882$ with $t=13.7$---because the individual component is most of it at the target---$0.383$ of the $0.420$ intercept. That inheritance licenses the rest of the paper: wherever only aggregate distributions or realized outcomes are available, the kink that drives the individual component also drives the observed total, and the purge can be applied without forecaster-level data.

\paragraph{Wages, and the flat lower arm.}\label{sec:wages}
Downward wage rigidity accounts for part of the one-sidedness, and not for the law. The variance of surprises in euro-area negotiated wage growth is lower when inflation is below the target than when it is above, because it falls with the undershoot; it does not rise with the overshoot. Regressing the squared surprises on the two arms of the gap gives
\begin{center}\small
\begin{tabular}{lcccc}
\toprule
Arms of the & upper arm & lower arm & upper $-$ lower & $n$ \\
\midrule
realized gap, previous quarter & $0.011$ $(0.5)$ & $-0.098$ $(-1.5)$ & $0.110$ $(2.2)$ & $109$ \\
expected gap (survey), previous round & $-0.041$ $(-0.5)$ & $-0.304$ $(-1.7)$ & $0.263$ $(2.1)$ & $105$ \\
\bottomrule
\end{tabular}
\end{center}
The individual arms are imprecise; their difference is not.\footnote{ECB indicator of negotiated wages, euro area, annual growth, quarterly, 1999--2025. Surprises are the residuals of an AR(2) in wage growth with the two arms of the gap; their squares are regressed on the same arms. HAC standard errors with four lags, $t$-statistics in parentheses.} Wage setting thus contributes a flat lower arm and no slope above the target, for which the level effect and the mechanism of Section~\ref{sec:model-law} remain necessary. It is where the re-basing the model describes meets a constraint of practice: nominal terms are re-based upward readily and downward hardly at all \citep{akerlof1996macroeconomics}.

\paragraph{The realized shadow.}\label{sec:realized}
The same shape appears in realized inflation, entirely outside the survey. Fitting an AR(1) to the euro-area gap ($\hat\rho=0.93$) and projecting the squared innovation on the two arms of the lagged gap gives
\begin{equation}\label{eq:lawrealized}
\widehat{\Var}(d_{t+1}\mid\Iset_t)\;=\;\hat\sigma_{0,\mathrm{real}}^2 \;+\; \underset{(0.065)}{0.074}\,(-d_t)_+ \;+\; \underset{(0.080)}{0.377}\,(d_t)_+ ,
\qquad R^2=0.40,\quad n=109 ,
\end{equation}
robust to using absolute rather than squared innovations ($\hat g=0.202$, $t=8.1$, on the symmetric specification).\footnote{The squared innovation proxies the unobservable conditional variance without bias in the slope, at a cost in precision \citep{patton2011volatility}. The coefficient is not comparable to the survey's $0.855$ as it stands: it prices one quarter's innovation, the survey the variance of the gap a year ahead.} Euro-area inflation is not merely volatile; it is volatile when it is \emph{above} target. This is the level effect, and it is mostly energy: between half and two-thirds of the realized upper arm is energy, and core inflation carries the arm only in the 2021--23 episode.\footnote{With monthly surprises of headline, energy and ex-energy HICP (autoregressions with month dummies, 1999--2026), the squared headline surprise splits into an energy part, an ex-energy part and their covariance; regressed on the arms of the lagged gap, the energy part carries $66$ per cent of the headline upper arm and the ex-energy part $8$ per cent. The GARCH estimates of Section~\ref{sec:reverse:law}, with oil in the mean and in the variance, give about half.} The realized arm is a fact about inflation, not a measurement of the observational-equivalence cost (Section~\ref{sec:reverse}). Nor does it outlast the gap. We estimate a joint ARMA--GARCH model of realized euro-area inflation \citep{engle1982autoregressive,bollerslev1986generalized} whose variance equation also carries the two arms of the lagged gap. Without the arms, realized variance is persistent, with a half-life of about two quarters. With them, the persistence falls to $0.10$ a month while the overshoot arm stays significant. Realized variance is therefore persistent because the gap is, not on its own. Forecasters' variance shows the same: conditional on the current gap, past overshoots leave no trace in it.\footnote{Monthly HICP, 1997--2026 ($n=357$), with the gap measured year on year. The mean is an AR(2) with a twelve-month moving-average term, selected by BIC, estimated jointly with the variance by Student-$t$ maximum likelihood. The variance is $h_t=\omega_h+\alpha\varepsilon_{t-1}^2+\beta_h h_{t-1}+\delta_+(d_{t-1})_++\delta_-(-d_{t-1})_+$. With the arms, $\hat\alpha+\hat\beta_h=0.10$ a month; $\hat\delta_+=0.059$ ($0.016$) and $\hat\delta_-=0.015$; constant variance around the arms is not rejected ($p=0.25$). Without the arms, $\hat\alpha+\hat\beta_h=0.72$ a month and $0.73$ in quarterly data. Gaussian errors, GJR asymmetry, the sample to 2019 and month-on-month inflation all give a quarterly persistence of at most $0.13$. On the ECB-SPF, the average individual variance regressed on its own lag and on the current and lagged arms loads $+0.62$ on the current overshoot and $-0.63$ on the lagged one, the sign of habituation rather than memory.}

\paragraph{Inflation options, once the premium is removed.}
The fourth source is the option-implied distribution, and it carries the paper's clearest illustration of why raw second moments cannot be used untreated. On raw option-implied variance the arms come out backwards: the below-target arm is strongly \emph{negative} ($-1.173$ on the days of the adjusted fit; $-1.109$ with $t=-7.8$ on the full raw sample) and the above-target arm is near zero. Removing the variance risk premium, which peaks within a quarter-point of $2\%$ and is near zero far from it, reverses it to
\[
b_-=-0.089\ (t=-0.78), \qquad b_+=+0.680\ (t=+7.85), \qquad R^2=0.36, \qquad n=3{,}941 ,
\]
flat below and sharply positive above. The premium is identified without ever touching the law it is then shown to recover. Implied variance is physical variance plus a variance risk premium, $V^{Q}_t=V^{P}_t+\mathrm{VRP}_t$. The physical variance is unobservable in the moment, but its \emph{realized} counterpart is not: the integrated variance actually delivered over each option's life, computed after the fact from the daily swap series. The premium estimate is a one-year centred rolling median of $V^{Q}_t$ minus that realized variance---a median so that single crisis days do not drag it, rolling so that the premium may drift across eras---and subtracting it from the implied series is the whole adjustment. The estimator never uses the fitted physical law $V^{P}(\de)$; it does use the swap rate's realized variance---the object the swap market's own row of Table~\ref{tab:armsbysource} fits---which carries the law itself, so what the adjusted surface can show is that an option-implied variance, net of a premium measured against that realized variance, keeps the law's form rather than the raw surface's inverse of it. The raw series prices insurance, not beliefs, and insurance is most expensive exactly where physical uncertainty is smallest. The option evidence is therefore read as corroboration of the law's form, not as a source independent of the swap market; the two surveys and the swap market carry the independent weight.

Figure~\ref{fig:lawfour} puts the four sources on one pair of axes, each normalised by its own intercept, so that what is compared is the shape and not the magnitude; Table~\ref{tab:armsbysource} collects each source's fitted arms with its own inference. One form, four sources: a quarterly survey of forecasters, a daily swap market, an option surface net of its premium, and the US survey; realized inflation carries the same shape.

\paragraph{The null that was a flat arm.} \citet{abel2016measurement}, on ECB-SPF data through 2013, find no link between aggregate uncertainty and the aggregate point prediction. We reproduce their null to three decimals ($\hat\beta=+0.038$, standard error $0.066$), and on the sample extended to the round conducted in 2026Q2 the same regression returns $+0.299$ ($0.063$). The two results are one: their null is the flat arm, on a sample in which the consensus forecast exceeded the target in $6$ of $60$ rounds and never by more than $0.39$ points (Appendix~\ref{app:abel_replication}).

\paragraph{The distribution of thresholds is not observed.} The slope of the law prices a population's distribution of private thresholds, which is never observed and can differ across jurisdictions and drift across episodes. The theory therefore delivers the \emph{form} of the law, not its level: each sample identifies its own arm, a ceiling on its own mechanism, as the euro-area and US arms show, $0.855$ against $0.063$ for the same form. The US law also lives in a different component: after the announcement US disagreement rises with distance on both sides, while each forecaster's own density is flat below the target and rises above it, so the aggregate looks symmetric and the individual component, which is what \eqref{eq:law} is about, does not (Appendix~\ref{app:us}).

\subsection{Normalized Uncertainty, applied: three literatures re-read}
\label{sec:nu}

Section~\ref{sec:results} defined the correction---Normalized Uncertainty, the square root of the raw second moment's ratio to the fitted envelope \eqref{eq:NU_definition}, anchor-adjusted uncertainty in a word---and its denominator carries the arms of Section~\ref{sec:law}: the form is the data's selection before it is the theory's. Nothing here restates the construction. What has to be said before using it is the ledger. Measured uncertainty at a date can be high because agents genuinely face a wider distribution of outcomes, or because the distance from target widens the distribution they report, through the level effect and through the contagion of Section~\ref{sec:results}. The first is information; the second is the distance component. What the denominator removes is exactly the part the distance predicts: everything correlated with $\dpl$ through the anchor leaves, everything orthogonal to it stays. Three boundaries follow it everywhere. $\mathrm{NU}$ is a measurement correction, not an identification strategy: it does not establish that the remainder is causal for anything, and no regression in this subsection is read as a causal estimate. The correction removes the predictable footprint of the doubt, not the doubt itself: the walls stand, the set does not close, and the verdict remains unidentified after the purge exactly as before it. And the direction of a corrected coefficient is not fixed in advance: removing a component the distance predicts leaves a residual that is cleaner as well as smaller, and a cleaner regressor can carry a \emph{larger} coefficient---the applications below do exactly that.

\paragraph{The corrected series.} Figure~\ref{fig:nu} shows what the correction does to the euro-area series itself, drawn in its simplest reading: unit coefficients rather than fitted ones, and the above-target arm alone. Panel~(a) plots the raw survey uncertainty against its purged counterpart: the two series are nearly indistinguishable through the anchored decades, because with the forecast gap at zero the denominator is flat and there is nothing to remove; they separate exactly where the distance opens---the 2021--23 surge and the 2026 episode---and the gap between them is the priced footprint the law attributes to the anchor. Panel~(b) sets the purged series against the euro-area economic policy uncertainty index, a proxy built from newspaper text that shares none of the survey's construction: the corrected series continues to register the episodes a text-based measure recognizes---the financial and sovereign crises, the pandemic, the climb after 2023---and, as drawn, agrees with the index more closely than the raw series does, $0.85$ against $0.75$, while no longer tracking the mechanical arithmetic of the distance from target. How much weight the height of the panel~(a) gap can bear is limited by the survey's own instrument: its top bin is open-ended, and over 2022--23 enough probability mass sits inside it that the raw second moment reflects where the tail is closed as much as what forecasters believe (figure note). That the two series separate is a fact about them; the size of the separation in those quarters is not a measurement. The growth densities give the correction an out-of-family check: applied to the survey's GDP growth histograms, the same purge raises the series' correlation with the same text-based index from $0.561$ to $0.716$ ($0.556$ against $0.713$ controlling for the distance from potential growth).

\paragraph{Cross-country growth.} The first setting is the oldest. \citet{barro1995inflation} reports that higher average
inflation lowers growth, and that conditional on the level the standard deviation of
inflation carries a coefficient of $-0.004$ with a standard error of $0.009$---``virtually
zero,'' in his words. He offers a reading of his own null: ``the realized variability of
inflation over each period does not adequately measure the uncertainty of inflation, the
variable that one would have expected to be negatively related to growth. This issue is
worth further investigation.'' That sentence is this subsection's premise, written thirty
years earlier.

We rebuild his panel from the \citeauthor{barro1995inflation}--Lee source and estimate his
equation by ordinary least squares over his own periods---1965--75, 1975--85 and 1985--90.
The rebuild reproduces him, coefficient by coefficient, and reproduces his null on raw
volatility.\footnote{Eleven of his thirteen controls fall within two of his standard errors;
the interaction of initial income with human capital comes out at $-0.434$ (s.e.\ $0.146$)
against his $-0.44$ ($0.17$); the sample is 241 country--periods against his 251; the fit is
$R^2=0.60$ against his $0.63$, $0.60$ and $0.49$; and inflation entered alone carries
$-0.0170$ (s.e.\ $0.0048$) against his $-0.0236$ ($0.0048$). The two controls missed by more
than two of his standard errors are fertility and the investment ratio---the latter the one
coefficient he himself singles out as the place where least squares and his instruments part
company. His rule-of-law index, observed only from the early 1980s and used among his
instruments, is not in the \citeauthor{barro1995inflation}--Lee distribution and is not
matched; everything else in his list is here, including the interaction of initial income
with human capital that governs the rate of convergence in his estimates.}

Conditional on the level, raw volatility carries $-0.008$ with
a $t$ of $-0.8$: virtually zero, as he reports. What it does is mask the level---beside raw
volatility, mean inflation carries $t=-1.3$. Deflate that volatility by the part the inflation
rate predicts and the cleaned measure becomes significant at five percent, $-0.142$ with a $t$
of $-2.2$, while the level's own $t$ improves to $-1.6$. In one-standard-deviation terms the
cleaned measure is worth $0.28$ percentage points of annual growth against the raw measure's
$0.15$. A wild-cluster bootstrap by country---the conservative inference at ninety-four
clusters---keeps the raw null a null ($p=0.43$) and sustains the purged result
($p=0.032$). And the construction matters: with mean inflation already in the regression, \emph{subtracting} a level-fitted
envelope from $\sigma_\pi$ reproduces the raw coefficient identically ($t=-0.8$), so an
additive correction cannot add information the regression does not already hold---only
the ratio can (Section~\ref{sec:purge}). Nothing in the exercise turns on knowing a target, explicit or implicit. His periods had no announced number, so the baseline deflator is built from the inflation rate itself---$\sqrt{1+|\bar\pi|}$ exactly, nothing estimated---and imposing a reference anyway does not change the result: deflating by $\sqrt{1+(\bar\pi-2)_+}$, the one-sided form the law implies around a two-percent number, returns $-0.132$ ($t=-2.3$; wild-cluster $p=0.024$), and the symmetric $\sqrt{1+|\bar\pi-2|}$ the same sign at ten percent ($t=-1.7$; wild $p=0.079$). The correction needs only the inflation level the deflator is built from, not the number a central bank did or did not say---which is what makes it portable to samples where none was ever announced. Table~\ref{tab:crosscountry} reports these two regressions in columns~(4) and~(5), beside each summary of the inflation experience entered on its own in columns~(1)--(3). Barro suspected that realized variability does not measure the uncertainty that should matter for growth; on his own panel, his own equation and his own periods, the correction he called for is what turns his null into a result. The reading is two-sided: what realized variability appeared to carry belongs mostly to the level of inflation, which remains a growth channel in its own right; and residual inflation uncertainty, once measured net of that level, weighs on growth as well---the two effects standing side by side where the raw regression showed one masking the other.

\paragraph{Credit pricing.} The second setting is the transmission of uncertainty to loan pricing, on French loan-level data---AnaCredit France, accessed at the Banque de France: $71{,}069$ overdraft facilities to non-financial corporations, new business, September 2018 to March 2026, drawn by $60$ bank groups over $31$ quarters. The specification regresses the annualised agreed rate on loan $\ell$ at bank $b$ in quarter $q$ on the quarter's uncertainty measure,
\begin{equation}\label{eq:creditspec}
r_{\ell bq} \;=\; \beta\, U_q \;+\; \gamma' x_{\ell bq} \;+\; \alpha_b \;+\; \alpha_{s(\ell)} \;+\; \varepsilon_{\ell bq},
\end{equation}
in four nested columns: the uncertainty measure alone; adding the loan-level debtor default probability and the macroeconomic controls (the deposit facility rate and log industrial production); adding bank fixed effects $\alpha_b$; and adding sector, size and location effects. Regressors are standardised, and because $U_q$ varies only by quarter---the classic setting for overstated precision---displayed $t$-statistics are clustered by bank, and because the measure varies only by quarter the table's note reports the quarter, two-way and wild-cluster bootstrap treatments of exactly that concern. Table~\ref{tab:credit} reports the four columns for two measures of $U_q$---the raw standard deviation $\sqrt{\Var_t}$ and $\mathrm{NU}$---substituted one at a time, never entered jointly, each with the deposit facility rate's own coefficient from the same regression beneath it.

The table reads in one contrast, and the contrast has an economic reading. Once bank fixed effects absorb between-bank pricing differences, the raw standard deviation carries no information about overdraft pricing at all---$t=0.34$ in the full specification---while $\mathrm{NU}$ is significant at the one percent level in the same column; the purge is not shrinking a result here, it is producing one. The reading starts from what the specification already holds: the stance. The deposit facility rate is on the right-hand side---its own coefficient, reported beneath each uncertainty measure, stays large and precisely estimated whichever measure is substituted ($0.0137$ beside the raw measure, $0.0109$ beside $\mathrm{NU}$)---and with the rate comes the level of inflation the stance is set against. Raw dispersion adds little to that pair because much of it \emph{is} the stance: the mechanical widening the distance from target predicts, priced by the law of Section~\ref{sec:law} and collinear with the rate the regression controls for; nor is the square root doing the work, since the raw forecast variance is as dead in the same column ($t=-0.28$). Normalized Uncertainty is the part of measured uncertainty the anchor does not explain---the residual left after the policy rate and the distance have said their piece, which the model reads as genuine imprecision about the inflation outlook. That residual is what a lender cannot read off the stance, and it is what the table shows loan rates moving with. The contrast survives the treatments the common-regressor problem requires---by quarter, $t=4.98$ against $0.31$; two-way, $3.43$ against $0.30$; and, sharpest, a restricted wild-cluster bootstrap on the thirty-one quarters, $p<0.001$ for $\mathrm{NU}$ against $0.783$ for the raw standard deviation.

The size needs the units stated. The regressors are standardised but the dependent variable is not: the agreed rate enters as a decimal fraction, with a mean of $0.0889$ on the estimation sample. The full-specification coefficient is $0.89$ \emph{percentage points} of loan rate per standard deviation of the purged measure, against a mean rate of $8.9\%$. That is a substantial effect.

One discrimination closes the reading. Measured inflation uncertainty could be priced because it proxies for general, economy-wide uncertainty rather than for uncertainty about inflation itself. Appendix~\ref{app:ngu} runs the identical specification on the growth analog of the correction---Normalized Growth Uncertainty, the companion paper's measure---and on the two orthogonalized components. Entered alone, general uncertainty does predict overdraft pricing ($+0.0238$, $t=3.3$): uncertain times are expensive times. With the default probability, the stance, activity and the fixed effects in, it prices nothing ($+0.0016$, $t=1.1$) while the purged inflation measure stands ($+0.0089$, $t=4.0$); and the orthogonalization sharpens the attribution---the inflation-specific component net of general uncertainty survives the full controls ($+0.0058$, $t=4.1$; wild-quarter $p=0.004$), while the general component net of inflation uncertainty does not raise loan rates ($-0.0025$, $t=-2.5$; wild-quarter $p=0.087$). What loan pricing prices is uncertainty about inflation, not uncertainty in general; the diffuse component rides on the controls the specification already holds.

\paragraph{A vector autoregression: the level, and the uncertainty.}
The natural fourth exercise is a macroeconomic vector autoregression in the tradition of
\citet{baker2016measuring}: a quarterly euro-area system in the policy-uncertainty index,
an uncertainty measure, HICP inflation, the deposit facility rate, unemployment and log
industrial production, identified recursively with the uncertainty block ordered
first.\footnote{Lag order by information criterion, ninety-percent asymptotic bands; the
subsample of Figure~\ref{fig:rawnu} ends in 2019Q4. The checks: the uncertainty measure
ordered last, a forced longer lag order, the system without the policy-uncertainty index,
a growth-uncertainty variant, the median- and variance-based raw alternates, which
behave like the mean, and the lag order forced to each of $p=1,2,3$: the raw response is away from zero at every one, and so is the purged series', smaller at every lag order. Every estimate is in the results file of \texttt{fig06\_ip\_response.py}.} On the full sample the design is uninformative about inflation uncertainty: where a survey
measure of it, raw or purged, is followed by an activity response outside the ninety-percent
band, the excursion ends within the first three quarters, and none of the cumulated
industrial-production responses, between $-1.3\%$ and $+2.8\%$, is negative in every
variant---on a hundred quarters and six variables, a recursive design cannot separate them
from noise. The sample that stops in 2019, before the 2021--23 episode dominates every
second moment, separates the measures instead (Figure~\ref{fig:rawnu}). There a
one-standard-deviation shock to the raw dispersion---the numerator of
\eqref{eq:NU_definition}, unpurged---is followed by a fall in industrial production away
from zero at ninety percent through most of the first two years, with a matching rise in
unemployment, and the response survives the reverse ordering, the longer lag order and
the exclusion of the policy-uncertainty index. The same shock to Normalized Uncertainty
is followed by about three-quarters of that path, and what remains is robust: it survives the reverse ordering, every lag order and the exclusion of the policy-uncertainty index. The difference between the two shocks is the denominator, and before 2020 the
denominator moves almost only in 2006--08 and 2011, when the consensus forecast ran above
the anchor on the eve of the two production collapses. The main result makes that reading
a test inside the VAR itself, and Figure~\ref{fig:rawnu} shows it. Ordering the
above-anchor distance $(\pi_t-\bar{\pi})^{+}$ ahead of the uncertainty measure in the
recursion---the recursive way to control for the distance---removes the raw measure's
industrial-production response almost entirely, a cumulative $-6.1\%$ collapsing to
$-0.4\%$ with no horizon away from zero beyond a single quarter, while the corrected measure keeps three-fifths of its response: $-4.6\%$ in the baseline ordering and $-2.8\%$ with the distance ordered first, away from zero at ninety percent through quarters two to five. Ordering the linear \emph{level} first
changes nothing, so it is the kinked distance, not inflation itself, that carries the raw measure's covariance; the same holds at the estimated anchor $1.96$, and
every estimate is in the results file of \texttt{fig06\_ip\_response.py}. Read
together, the two channels stand separately: the \emph{level} of inflation---through its
distance from the anchor---is followed by lower industrial production, and residual
uncertainty about inflation---the part the distance does not explain---is followed by
lower industrial production too. What the raw series carried was the sum; the correction takes out part of the first and keeps the second. As with any recursive design these are conditional covariances, not causal
effects; the exercise is an account of what the purge
removes and of what remains once it is removed, not an estimate of what uncertainty
does.
    \section{Identifying the observational-equivalence cost}
\label{sec:reverse}

The OEE is a statement about beliefs. Expectations are an argument of the Phillips curve \eqref{eq:nkpc}, however, so the doubt the law prices need not stay there: having left the common anchor, doubters set wages, prices and contracts around effective targets nobody announces. We define the \emph{observational-equivalence cost} as the part of the variance of realized inflation, and of any effect on its level, that this re-basing adds: what the economy pays because optimal tolerance and unwarranted drift cannot be told apart. It is charged only while expected inflation sits above the announced target, and on the point-years spent there.\footnote{Point-years above target---size multiplied by duration, so two years at one point count as one year at two---are what the cost is charged on, and peak inflation, the number that dominates public discussion, is the wrong statistic for it.} It is what a central bank would need to weigh against the gain of tolerance, and measuring it requires the \emph{reverse link}, from inflation uncertainty to the inflation gap: its level, its volatility, and the path of the policy rate expected to close it. Realized euro-area inflation is more volatile above the target than below (Section~\ref{sec:realized}), but that one-sided arm is the level effect, between half and two-thirds of it energy, and nothing in it separates the doubters' contribution. This section asks whether the reverse link can be estimated on euro-area data. It cannot, and the reason is structural rather than a matter of sample size: identification requires variation in uncertainty that leaves inflation's fundamentals unmoved, and the field data hold none.

The tool, level-dependent conditional variances, first replicates the one-sided law on the paper's own ground (Section~\ref{sec:reverse:law}). It delivers a predictive link from the market's expected gap to realized inflation, not a causal one (Section~\ref{sec:reverse:level}); its uncertainty terms are fragile and could not be read causally in any case, because GARCH uncertainty has no innovation of its own (Section~\ref{sec:reverse:unc}); and the one external source of uncertainty we could build, oil, is exogenous but moves nothing in the euro area (Section~\ref{sec:reverse:oil}). Details and robustness checks are in Appendix~\ref{app:lgm}.

\subsection{Data and the level-dependent variance}
\label{sec:reverse:data}

Three euro-area series enter. The \emph{realized} gap $d_t$ is month-on-month HICP inflation, seasonally adjusted and annualized, minus the announced target ($2\%$), monthly from 1999. The \emph{expected} gap $d^{e}_t$ is the one-year zero-coupon inflation swap minus the target, daily from October 2009, with the days on which the swap's reference month rolls treated as missing. The \emph{expected policy rate} $i^{e}_t$ is the one-year \texteuro STR overnight-index swap. The expected gap is the object $\de=\E_t d_{t+1}$ of Section~\ref{sec:results}, measured here by the market rather than by the survey consensus of Section~\ref{sec:evidence}; and the realized gap is monthly and annualized rather than the quarterly gap of Section~\ref{sec:realized}, so slopes are compared below only after conversion to common units, and then as orders of magnitude.

The tool is the level-dependent GARCH that the term-structure literature uses for the short rate \citep{brenner1996another,koedijk1997dynamics}, in which the conditional variance is the product of a level-free GARCH component and a function of the lagged level. Two of our series are negative on part of the sample, so the level function is split at a centre, with one power on each side: for a daily series $y_t$ with lagged level $L_{t-1}$ (the expected gap, centred at the target; the expected rate, centred at zero),
\begin{equation}\label{eq:lgm-mult}
\begin{aligned}
\Delta y_t&=c+(\phi-1)\,y_{t-1}+\varpi\,\sigma_t+\varepsilon_t,\qquad \sigma_t^2=h_t\,\Lambda_t,\\
\Lambda_t&=\bigl(1+(L_{t-1})_+\bigr)^{2\eta_+}\bigl(1+(-L_{t-1})_+\bigr)^{2\eta_-},\qquad
h_t=\omega_h+\alpha\,\frac{\varepsilon_{t-1}^2}{\Lambda_{t-1}}+\beta_h\,h_{t-1}.
\end{aligned}
\end{equation}
The powers $\eta_\pm$ carry the level effect on each side of the centre, and $h_t$ is \emph{level-free} uncertainty: the GARCH runs on the shock normalized by the level factor, so $\mathcal U_t\equiv\sqrt{h_t}$ is the part of the conditional standard deviation that the level does not explain---the time-series counterpart of the residual uncertainty that Normalized Uncertainty isolates in the survey. For monthly realized inflation the literature's form is additive \citep{brunner1993higher,grier1998inflation,grier2000effects}, and the data agree (Appendix~\ref{app:lgm}): with $\mu_t$ an AR(6) prediction of $d_t$, known at $t-1$,
\begin{equation}\label{eq:lgm-add}
d_t=\mu_t+\varpi\sqrt{h_t}+\varepsilon_t,\qquad
h_t=\omega_h+\alpha\,\varepsilon_{t-1}^2+\beta_h\,h_{t-1}+\delta_+(\mu_t)_++\delta_-(-\mu_t)_+ .
\end{equation}
The terms $\varpi\sigma_t$ and $\varpi\sqrt{h_t}$ are the GARCH-in-mean channel of \citet{engle1987estimating}. Errors are Student-$t$, every likelihood-ratio test is computed after iterated cross-polishing of the nested models, and one-sided tests use the $\bar\chi^2$ distribution.

\subsection{The law in conditional variances}
\label{sec:reverse:law}

All three conditional variances carry the one-sided level effect (Table~\ref{tab:lgm-law}). In realized inflation, a predicted overshoot raises the conditional variance ($\hat\delta_+=7.97$, likelihood ratio $13.0$, $p=0.0002$) and a predicted undershoot does not ($p=0.35$); the data leave room for only a small ARCH term, and within that room the result holds. In the daily expected gap, the volatility of the market's revisions rises with the expected overshoot ($\hat\eta_+=1.04$, likelihood ratio $36.5$) and not with the undershoot ($p=0.67$): one point of expected overshoot doubles the conditional standard deviation. The effect survives a variance dummy for 2021--23 ($\hat\eta_+=0.75$), and before mid-2021, when the expected gap was positive on $4\%$ of days, it keeps its size at borderline significance. Once the level factor is removed, the 2022 surge in the swap's volatility largely disappears from $\mathcal U$ (Figure~\ref{fig:lgm-levelfree}): what the raw conditional variance showed in 2022 was mostly the level.

Two qualifications define the scope. The realized law is a property of \emph{headline} inflation, the level effect of Section~\ref{sec:model-law}, and about half of it is energy: with oil in the mean and the variance $\hat\delta_+$ halves, to $4.6$, and stays significant and one-sided; ex-energy inflation carries a third of it ($2.8$, $p=0.02$); core inflation carries none outside the 2021--23 episode. And the expected policy rate shows the same pattern ($\hat\eta_+=1.36$, likelihood ratio $19.4$; $\eta_-$, $p=0.31$), robust to the sample, a hiking-phase dummy and the form, but there it is the familiar compression of short-rate volatility near the lower bound \citep{swanson2014measuring}, a property of the policy-rate process rather than an analogue of the law.

In the survey's units---variance of inflation one year ahead, per point of expected overshoot---the market and the survey agree on the slope (Table~\ref{tab:lgm-compare}): about $0.9$ over the first point against the survey's $0.86$. Realized headline inflation responds two to five times more, much of the excess being energy, which forecasters' densities do not scale with the expected overshoot. The three objects differ, so the comparison is one of orders of magnitude, not a test (Appendix~\ref{app:lgm:units}).

\subsection{First moments: a predictive link, not a causal one}
\label{sec:reverse:level}

To ask whether expected inflation and its uncertainty move realized inflation, the realized equation \eqref{eq:lgm-add} is augmented with the market's expected gap and its level-free uncertainty, both of the previous month (2009:11--2026:01, $n=195$),
\begin{equation}\label{eq:lgm-aug}
\mu_t=c+\textstyle\sum_{i=1}^{6}\phi_i d_{t-i}+\theta_e\,d^{e}_{t-1}+\varpi_{\mathcal U}\,\mathcal U^{e}_{t-1},\qquad
h_t=\ldots+\delta_{\mathcal U}\,\bigl(\mathcal U^{e}_{t-1}\bigr)^2 ,
\end{equation}
with $\mathcal U^e$ the level-free uncertainty of the expected gap from \eqref{eq:lgm-mult}, a monthly mean, and the dots the terms of \eqref{eq:lgm-add}. The coefficient $\theta_e$ is a level effect of expectations on realized inflation; $\varpi_{\mathcal U}$ and $\delta_{\mathcal U}$ are the reverse link, uncertainty in the mean and in the variance.

The expected level predicts realized headline inflation about one for one (Table~\ref{tab:lgm-theta}): $\hat\theta_e=1.07$ ($p=0.005$). The swap thus carries information on next month's inflation beyond inflation's own past. The coefficient holds alone ($p=0.038$), under any variance memory the data admit, with dummies for 2021--23 in the mean and the variance ($1.18$), and under a twelve-month block bootstrap (weighted estimate $0.96$, $90\%$ interval $0.41$--$1.71$). Monthly $d^{e}$ is nearly a unit-root process, so the test is a predictive regression on a persistent regressor, which standard inference can over-reject \citep{stambaugh1999predictive,campbell2006efficient}. Inference robust to any degree of persistence \citep[IVX;][]{phillips2009econometric,kostakis2015robust} rejects $\theta_e=0$ in every variant ($p\le0.004$), with larger estimates, $1.4$ to $1.9$: on headline inflation, a coefficient of at least one. Part of it is energy. The one-year swap is priced partly off the energy futures curve, and with last month's Brent price, exchange rate and energy inflation as controls the coefficient falls to $0.73$ ($p=0.076$). The swap also predicts \emph{core} inflation, with a smaller coefficient that oil controls leave intact, $0.24$ to $0.33$ ($p\le0.004$; $0.37$ to $0.41$ under weighted IVX), borderline once 2021--23 is dummied. Out of sample, over 2016--2026, adding the swap improves one-month-ahead forecasts of core inflation by $9\%$ in mean squared error \citep[$p=0.005$;][]{clark2007approximately} and of headline only marginally ($3\%$, $p=0.11$).

The two indices bracket the size of the link, and we report both: the model of Section~\ref{sec:results} is not tied to a price index. $\theta_e$ relates last month's market expectation to this month's annualized inflation, and is a predictive coefficient, not a structural pass-through: the swap prices the same news that moves inflation, so a market that forecasts well and a public whose expectations move prices produce the same $\theta_e$. It is the \emph{predictive link}.

\subsection{Uncertainty terms, and why a GARCH cannot identify them}
\label{sec:reverse:unc}

The reverse link proper is not established. Uncertainty in the mean is insignificant alone and in the full model ($\varpi_{\mathcal U}$, $p=0.10$). Uncertainty in the variance is borderline ($\hat\delta_{\mathcal U}=26.9$, $p=0.048$) and fragile on every margin: it is not separately identified from the GARCH constant, both being slow-moving positive shifters (with a floor on the constant it falls to $12.9$, $p=0.11$); it disappears once variance persistence is forced into the admissible range; and it is energy, falling to zero in every specification with oil controls, on core and on ex-energy inflation. Between the two market series the picture is the same. In the two-way Granger tests on the expected gap and the expected rate (Appendix~\ref{app:lgm}), the levels carry the signs economics predicts---a higher expected rate lowers expected inflation the next day, higher expected inflation raises the expected rate---but lose significance once their joint near-unit-root persistence is handled \citep{toda1995statistical}. The one robust cross-series relation runs between second moments: uncertainty about the expected gap predicts a higher variance of the expected rate (likelihood ratio $25.3$; $26.5$ with a model-free realized volatility), before mid-2022 as well as after, while rate uncertainty does not predict the variance of expected inflation. That is a volatility spillover in predictions.

None of these terms could have been read causally, whatever its $p$-value, and the reason is the measure itself. The level-free uncertainty $\mathcal U=\sqrt{h}$ is the conditional standard deviation of a GARCH process, and from \eqref{eq:lgm-mult} its innovation is
\begin{equation}\label{eq:lgm-innov}
h_{t+1}-\E_{t-1}h_{t+1}\;=\;\alpha\Bigl(\frac{\varepsilon_t^2}{\Lambda_t}-h_t\Bigr),
\end{equation}
an exact function of the first-moment shock. GARCH uncertainty has no innovation of its own: it rises only in the wake of large news. A positive $\delta_{\mathcal U}$, or a variance spillover, is therefore generated alike by two different mechanisms---a causal effect of uncertainty on the dispersion of inflation, and the propagation of large shocks, whose size raises both $h$ and the volatility that follows---and within the model the two leave the same likelihood, because the variation in $\mathcal U$ that would identify the first is produced by the second. The reverse link thus has an observational equivalence of its own, and it mirrors the forward one. In Lemma~\ref{prop:oe}, the public cannot tell tolerance from drift because both generate the same inflation history; here the econometrician cannot tell an effect of uncertainty from the echo of the news that created it, because both generate the same conditional variances. Conditioning does not help, since the inflation level, the policy rate and global risk all co-move with $\mathcal U$ and the common driver is unobserved.\footnote{Under the null that its coefficient is zero, a generated uncertainty regressor leaves the size of the likelihood-ratio test intact \citep{pagan1984econometric}, so the tests above are valid as tests; it is the causal reading of a rejection that fails. Each uncertainty result is also re-run with a model-free realized volatility, with the same conclusions (Appendix~\ref{app:lgm}).} Separating the mechanisms requires variation in uncertainty that is neither a function of euro-area first-moment shocks nor driven by euro-area inflation and policy: an external shock to uncertainty.

\subsection{An external source of uncertainty: oil}
\label{sec:reverse:oil}

Oil-price uncertainty is the natural candidate. The euro area is a price-taker on world oil markets: euro-area monetary announcements move Brent little, and no differently from the exchange rate \citep{soriano2022response}, energy prices do not respond to macroeconomic news at daily or monthly horizons \citep{kilian2011do}, and oil uncertainty can be read from option prices, which carry an innovation distinct from return news. Oil \emph{supply news} has served as an external instrument for the oil price \citep{kanzig2021macroeconomic,gagliardone2023oil}, but it is a level shock, and no instrument of that kind exists for oil uncertainty; the designs that separate first and second moments use natural experiments for uncertainty in general---disasters \citep{baker2024using}, the price of gold \citep{piffer2018identifying}---and recursive orderings are invalid for the purpose \citep{kilian2025macroeconomic}. The difficulty is that oil uncertainty loads on global uncertainty, which raises oil-price volatility \citep{vanrobays2016macroeconomic}, is tied to the oil level through precautionary demand \citep{alquist2010what,gao2022oil}, and is a common factor across asset classes and macroeconomic aggregates \citep{bloom2009impact,jurado2015measuring,baker2016measuring,caldara2022measuring}, largely an endogenous response to the cycle on its macroeconomic side \citep{ludvigson2021uncertainty}. A GARCH on oil returns would not help: its uncertainty, like \eqref{eq:lgm-innov}, would have no innovation of its own \citep{elder2010oil,jo2014effects}.

We therefore build the shock from option prices. The daily oil-specific uncertainty news $o_t$ is the change in log implied crude-oil volatility (the CBOE OVX) net of current and lagged oil returns, which remove first-moment news and the leverage effect, and net of a global block of implied volatilities and policy uncertainty, standardized (2008:06--2026:09; construction in Appendix~\ref{app:lgm:oil}). It passes every exogeneity check we ran: it does not load on the ECB's monetary-event-window surprises, and lagged euro-area changes in expected inflation, the expected rate and option-implied inflation uncertainty do not predict it. It is not a pure second-moment shock: it forecasts a fall in Brent of $7.3\%$ after a year, the precautionary-demand pattern, so responses are also reported net of the realized oil path.

The test is a monthly local projection of eight euro-area outcomes---the inflation level and its volatility, expected inflation and its uncertainty, the policy rate and its variance, the expected rate and its uncertainty, twenty-four measures in all---on the shock, with the average response over the first year as statistic and Holm's correction across the measures \citep{holm1979simple}. A global uncertainty factor and the oil supply-news shock of \citet{kanzig2021macroeconomic} run as benchmarks in the same projections (Appendix~\ref{app:lgm:oil}).

The result is a clean negative (Table~\ref{tab:lgm-oil}). Oil-specific uncertainty moves none of the eight outcomes: no Holm-adjusted $p$-value falls below $0.43$, its only traces are a transient move in short inflation compensation that reverts within a month and a tilt in the survey's skew, and it does not move option-implied inflation uncertainty, so it has no first stage ($F\approx0.2$). The test has power where it matters: on the conditional standard deviation of the expected rate, the $95\%$ interval for $o_t$ excludes the global factor's response. The global factor, by contrast, co-moves with euro-area uncertainty---option-implied inflation uncertainty one year ahead rises by $4.1$ basis points ($p=0.021$), the conditional standard deviation of the expected rate by $0.25$ basis points, the frequency of policy moves with them---but it also moves euro-area first moments: the deposit facility rate falls by $9$ basis points, the one-year-in-one-year inflation swap by $4.9$, Brent by $4.8\%$. It may well be external to the euro area; what fails is the exclusion restriction, since it reaches euro-area inflation through demand, oil prices and policy as well as through uncertainty. None of its responses survives the multiple-testing correction either (smallest adjusted $p=0.50$). Two further designs on the same shock, oil uncertainty as a moderator of the response of expected inflation to inflation surprises and instruments built from OPEC announcement days, fail their pre-test and are weak, respectively (Appendix~\ref{app:lgm:oil}).

\subsection{What this means}
\label{sec:reverse:means}

The variation in uncertainty that is exogenous to the euro area does not move euro-area uncertainty, and the variation that moves it also moves first moments. The causal effect of inflation uncertainty on the inflation gap---on its level, its volatility, or the policy path expected to close it---therefore has no current estimate, and neither has the observational-equivalence cost. What the data deliver is the law itself, in three conditional variances, and a predictive link from the market's expected gap to realized inflation, headline about one for one and core about a third of that; neither is a reverse link. The limits are partly those of the instruments: monthly power is low for option-implied uncertainty and the policy rate, whose intervals still include a global-sized effect, and implied volatility on rate options and oil-option skews would sharpen the tests. But the obstacle of \eqref{eq:lgm-innov} remains for any measure built from past news, and the external variation it calls for has not been found. The OEE cannot be isolated but can be purged (Section~\ref{sec:purge}); its cost can be neither isolated nor purged, so it cannot be priced. Defining the cost leaves any episode's tolerance as unidentified as Lemma~\ref{prop:oe} says: it is a statement about the cost of the doubt, not a judgment on the bank. Section~\ref{sec:living} draws the consequence for policy.
    \section{What follows for policy}
\label{sec:living}

The optimum cannot be identified, only bracketed by published readings no data certify; the verdict on any past episode is undecidable; and the observational-equivalence cost of that impossibility has no current estimate (Section~\ref{sec:reverse}). The OEE can be purged from measurement and bounded (Section~\ref{sec:purge}); its cost can be neither, so the weight of the tolerance--uncertainty relation in the bank's problem is not measurable, and a central bank has to live with it. This section draws the consequences in three steps: whether the bank's ignorance of the slope and of the shocks calls by itself for less tolerance, which would remove the OEE's cost as a by-product; what announcing a target buys, and why the cost it carries has no remedy; and what it would take to measure the public's tolerance, the OEE and its cost.

\subsection{Is there a divine coincidence?}
\label{sec:acting}

The OEE and its cost cannot be estimated (Sections~\ref{sec:evidence} and~\ref{sec:reverse}), so the bank cannot weigh them against the gain of tolerance. One way out would be a divine coincidence, in the sense of \citet{blanchard2007real}: the policy that is optimal for the problem the bank can state would also be optimal for the one it cannot. The bank observes neither the slope of the Phillips curve nor the composition of its shocks (Appendix~\ref{app:walls}). If that ignorance alone called for less tolerance, the policy robust to it would also shrink the OEE, whose part due to tolerance is zero at zero tolerance and rises linearly with it (Corollary~\ref{cor:bill}), and the unmeasurable cost would be handled without being measured.

There is no such coincidence. In the linear-quadratic model, ignorance of the composition is certainty-equivalent: the bank acts on its estimate of the share \citep{svensson2003indicator}. Ignorance of the slope calls for caution in the use of the instrument \citep{brainard1967uncertainty}, that is, for more tolerance, not less. The textbook New Keynesian model, at the baseline calibration of \citet{Gali2015} with the slope and the composition unobserved over their published ranges, confirms both (Appendix~\ref{app:dsge}, Table~\ref{tab:dsge}, Figure~\ref{fig:dsge}). The tolerance a bank chooses without observing either, between $0.54$ for a transitory impulse and $0.78$ at a quarterly persistence of one half, is slightly above the average tolerance of a bank that sees the composition of each impulse and above the average optimum of a bank that knows the share and the slope. Robustness to the worst case raises it further, to between $0.58$ and $0.87$. Not observing the slope and the shocks never calls for tolerating less.

Tolerating less to shrink the OEE is therefore a choice, and the model prices it: halving tolerance raises the expected loss by $4$ to $29$ per cent, and tolerating nothing by $16$ to $116$ per cent. Its benefit is the reduction of a cost the paper cannot measure, so the trade-off cannot be settled, and we cannot say how optimal monetary policy should change on the OEE's account. One regularity survives. The price of tolerating less falls with the persistence of the impulse, from $29$ to $4$ per cent for halving it, while the point-years above target on which the OEE is charged rise with that persistence: insuring against the OEE is cheapest, in textbook terms, in the episodes in which the theory says it is largest. Persistence, unlike the composition, is revealed as an episode lasts (Lemma~\ref{prop:oe}), and the point-years are observed in real time. Monitoring them, and measuring uncertainty net of the distance component (Section~\ref{sec:nu}), is what the bank can do.

\subsection{The benefit of announcing a target, and what it does to measured inflation uncertainty}
\label{sec:announce}

This subsection takes the announced number's two effects on measured inflation uncertainty in order---the benefit, which the anchoring literature documents, and the cost, the OEE of Section~\ref{sec:results} and whatever it costs in realized inflation---then states why observational equivalence leaves the cost without a remedy, and what the pair implies for how inflation uncertainty should be measured.

\paragraph{The benefit.} The case for announcing a number is a commitment case, and its best-measured benefit is the one the anchoring literature documents: a numerical objective coordinates beliefs, compresses their dispersion, and gives the public something to hold the bank to \citep{levin2004macroeconomic,gurkaynak2010does}---the benefit the ECB itself names: ``This target provides a clear anchor for inflation expectations'' \citep{ecb2024pricestab}. The record on inflation itself is more guarded---for advanced economies the announcement's marginal effect on inflation is hard to distinguish from zero \citep{ball2004does}---and the costs the literature prices are costs of the \emph{constraint}: the real volatility of defending the number against supply shocks, the lower-bound exposure of a low one, the excess weight a salient public signal can command \citep{morris2002social}.

\paragraph{Why nothing removes the cost.} The OEE, and the observational-equivalence cost it may carry, are costs of the \emph{impossibility of demonstrating optimality}. An announced number gives the doubt its object: it creates the standard a deviation is measured against, and with it an obligation to explain any deviation that observational equivalence guarantees the bank can never discharge. The cost is not that of a target that is missed but of having said the number and being unable to prove why inflation is not sitting at it. Communication cannot remove it, because it is charged by the equivalence itself (Section~\ref{sec:model-law}); moving the choice to another institution cannot, because it attaches to the announced number and its unverifiable warrant; and opacity is no door, since a public that does not know the bank's stance must also price the dispersion of the stances it entertains, a charge that grows with the square of the distance from the target. This does not argue against announcing: the anchoring benefit is real and present at every distance.

\paragraph{Nor does credibility: tolerance is not credibility.} Tolerance and imperfect credibility are converse failures. A tolerant bank knows its intention but cannot verify, even to itself, that the intention is optimal against the true structure (Walls~III and~I, Appendix~\ref{app:walls}); the public cannot verify the intention at all (Lemma~\ref{prop:oe}). An imperfectly credible bank may act optimally and still not be believed. Credibility would cure the second failure, not the first. In the reputation tradition credibility is a belief that conduct builds and deviations erode \citep{barro1983rules,backus1985inflation,cukierman1986ambiguity}, at a speed set by how well observed inflation reveals the bank's goals \citep{faust2001transparency}. Here that speed is zero: a truthful bank and one that books its drift as intention leave histories of identical likelihood, so no record of conduct builds, or erodes, the trust that would separate them. Confidence granted rather than earned does act on the law, across regimes: a public that puts more than even weight $q$ on the tolerance reading books less of each point as drift, and the slope falls with $1-q$ \eqref{eq:slope-decomposition}; it leaves the kink at the announced number for every $q<1$. Forecasters do extend the bank credibility in their point forecasts, pricing a return to target faster than realized inflation delivers ($\rsub=0.752<\robj=0.929$ in the euro area, $0.681<0.866$ in the United States). Credibility is a first-moment asset and the OEE a second-moment charge; the first does not pay the second.

\paragraph{What it means for measuring inflation uncertainty.} A cost that cannot be removed from the economy can still be removed from the measure. Survey densities, option-implied distributions and realized inflation all record the distance from the target as uncertainty, so a regression that uses a raw measure as a regressor inherits that distance as a confound. Normalized Uncertainty removes it (Section~\ref{sec:purge}), and Section~\ref{sec:nu} shows that the correction changes the answers in credit pricing, growth and activity. The benefit of the number is a first-moment fact, and it stands; the cost is a second-moment fact, and it stays; inflation uncertainty measured against an announced number has to be read net of the number's own footprint.

\subsection{Measuring the public's tolerance, and the observational-equivalence effect}
\label{sec:measuring_tolerance}

The public's tolerance is not a preference the bank can consult: it
is a distribution of private thresholds, visible only in its consequences. The observational route this paper has run measures the
distance component---its slope and a kink located at the announced target
(Section~\ref{sec:law})---but it cannot separate the OEE inside it, it holds no shock to uncertainty that leaves the expected level in place (Section~\ref{sec:reverse}), and it contains
no counterfactual, because the data hold one path. Two designs would
reach what observation cannot.

\paragraph{The bank's side: the field's one menu.} Behind the public's thresholds sits the harder object, the bank's dial against its rule---the pair the walls leave on a ridge, separated only by \emph{counterfactual menus}: answers conditioned on states of known magnitude and composition, exactly what field histories never contain (in the language of experimental design, \citealp{healy2026which}). The field has fielded one such menu, the New York Fed's scenario-matrix question to dealers and market participants; Appendix~\ref{app:spd_panel} reads it, and what it shows is a band of perceived rules that never closes: respondents who share the same history keep disagreeing about a coefficient that history does not pin down.

\paragraph{The laboratory route: a spread without a level.} What the field never supplies is a shock to
uncertainty that leaves the expected level where it was. The sources Section~\ref{sec:reverse} could find
move both, and survey treatments that report how far forecasters disagree move respondents' means along
with their uncertainty \citep{kumar2023effect,georgarakos2024causal}. A laboratory can supply the shock by
construction. Professional market participants---forecasters, traders, treasurers---face a stylized
economy with an announced target and are shown the distribution from which next year's inflation, and
their payoff, will be drawn; the treatment widens it and keeps its mean, a mean-preserving spread
\citep{rothschild1970increasing}. Each subject's density, elicited before and after under a proper
scoring rule, verifies the first stage: the variance moves and the mean does not. The outcomes are where
the uncertainty goes: into beliefs about the policy rate, growth and the anchor, and into an
incentivized decision, the price each subject would pay to index a nominal claim rather than keep it on
the announced number. The difference between the arms in that price, at a common expected gap, is the
reverse link in isolation. Laboratory work has varied what subjects are told about uncertainty
\citep{rholes2021should}, but with students, and with forecasts as the only outcome. Randomizing the
share of treated subjects in groups that see one another's forecasts and choices
\citep{baird2018optimal} then measures how far the uncertainty travels to those who were not treated:
the contagion of Assumption~\ref{ass:contagion}, in the laboratory's units. The design gives up the
observational-equivalence effect itself, which only revealing the composition of the overshoot to some
subjects could isolate; it buys the reverse link, which is what the cost needs. Only with that link
measured could the cost be given a size, and a consequence for optimal policy.
    \section{Conclusion}
\label{sec:conclusion}

How much of an inflation deviation to tolerate is a well-defined optimum inside the profession's own workhorse model, and our first result is that it cannot be identified---not for want of data but by construction. The curve, the rule and the shock fail separately, and together they leave an optimally tolerated gap and an unwarranted drift observationally equivalent in the inflation history: the verdict on any past stance can be argued, never settled.

The impossibility is not free. Because the announced target names the number a deviation is measured against, a public that cannot verify the bank's intention prices the unverifiable bet one private threshold at a time, and forecast variance is flat below the target and rises linearly above it. That is the observational-equivalence effect, a proposition of theory. The data show a relationship of the same form in four sources, with the kink at the announced number, but two effects produce it. The larger is a level effect, by which inflation is more volatile when it is high; in realized inflation it is mostly energy. The two share one regressor, the distance from the target, and the data admit the observational-equivalence effect without isolating it.

The forward link, from tolerance to uncertainty, cannot be isolated, but it can be purged. Normalized Uncertainty removes the distance component of measured uncertainty, both effects at once, and the purge changes inference. Re-read through it, uncertainty's price in credit rises from nothing to decisive, \citeauthor{barro1995inflation}'s null becomes a result on his own panel, and the industrial-production decline that follows a raw uncertainty shock before 2020 shrinks by a quarter, and all but disappears once the distance from the anchor is held, while the purged measure keeps three-fifths of its response. The level's distance from the anchor and residual uncertainty about inflation are two channels, and the raw series carried their sum.

What the observational-equivalence effect costs remains unmeasured. The cost is charged on the point-years inflation spends above target, not for missing the target but for being unable to prove why inflation is not sitting at it, and the reverse link, from uncertainty to the inflation gap, is not identified in field data: in conditional variances the law holds again and the market's expected gap predicts realized inflation, but the uncertainty a GARCH measures moves only after the news that created it, and the one external source we could isolate, oil-specific uncertainty, moves nothing in the euro area, while the global uncertainty that does also moves rates and prices. A laboratory that widens the distribution of inflation professionals face, and holds its mean, would give the reverse link its first clean measurement. Until then we cannot say how optimal monetary policy should change on the effect's account. The announced number nonetheless remains worth having: it is what compresses beliefs at the target. A central bank that cannot prove its tolerance optimal can still choose it deliberately, say what it is doing, and know that the saying has a cost of unknown size.

    \newpage
    \bibliographystyle{aea}
    \bibliography{literature}

\begin{figure}[!htbp]\centering
\includegraphics[width=0.92\linewidth]{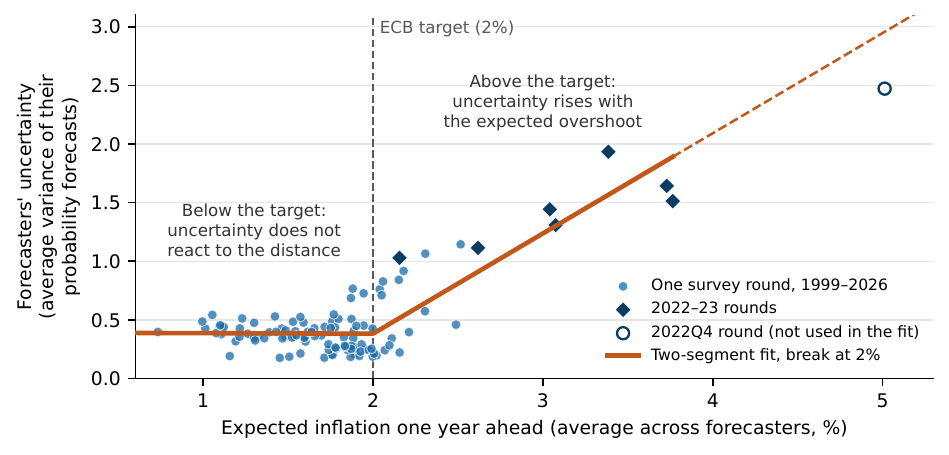}
\caption{\textbf{Uncertainty about inflation is flat below the 2\% target and rises above it.}
The figure is built from the ECB Survey of Professional Forecasters. Every quarter, about forty
professional forecasters each report a probability distribution for euro-area inflation one year
ahead: the probability they attach to inflation falling in each of a set of intervals. Each dot is
one survey round. Its horizontal position is the round's consensus, the average across forecasters
of the inflation each expects; its vertical position is the round's average uncertainty, the
variance of each forecaster's distribution averaged across forecasters ($W_t$, in squared
percentage points). The line is the survey law \eqref{eq:lawsurvey}: two straight segments that
meet at the $2\%$ target, fitted by least squares to the $109$ rounds from 1999Q1 to 2026Q2.
Below the target the segment is flat (slope $0.004$, indistinguishable from zero); above it,
average uncertainty rises by $0.86$ for each percentage point of expected overshoot, and the
distance from the target accounts for $70\%$ of the variation across rounds. Diamonds mark the
2022--2023 rounds. The hollow circle is the 2022Q4 round, whose consensus ($5.0\%$) lies
outside the range $[-1,5]$ used for estimation: it is shown but not fitted, and the dashed line
extends the fit beyond the estimation sample. Script: \texttt{make\_nts\_figure.py --paper},
which checks that the fitted coefficients are those of \eqref{eq:lawsurvey} before drawing.}
\label{fig:lawintro}\end{figure}

\begin{figure}[!htbp]\centering
\includegraphics[width=0.86\linewidth]{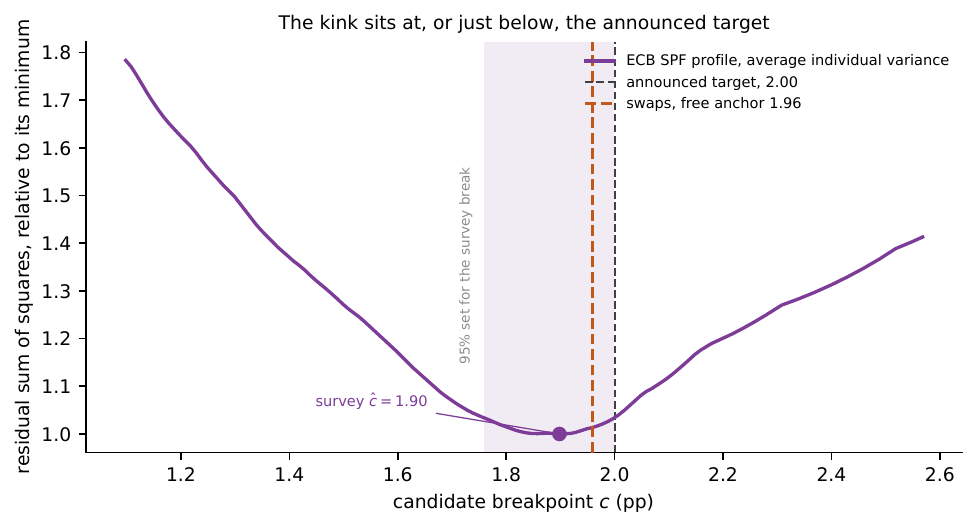}
\caption{\textbf{The kink sits at, or just below, the announced target.} Residual sum of
squares from the ECB SPF---the average individual predictive variance regressed on the two
arms---as the breakpoint $c$ is profiled over the central ninety percent of the consensus
forecast's range, normalised by its minimum. The profiled optimum is $\hat c=1.90$; the
shaded band is the profile-likelihood $95\%$ set, $[1.76,\,2.00]$, with the announced target
at its edge; a nonparametric bootstrap over rounds, $2{,}000$ draws, gives a $90\%$ interval
of $[1.78,\,2.01]$. The independent free-anchor estimate from the daily inflation-swap
market, $1.96$, lies inside the survey's set. Nothing imposes $2\%$ in either estimate.
Script: \texttt{fig01\_kink\_location.py} (the survey profile is estimated there; the
swap estimate is carried as printed).}
\label{fig:kinkloc}\end{figure}

\begin{figure}[!htbp]\centering
\includegraphics[width=0.86\linewidth]{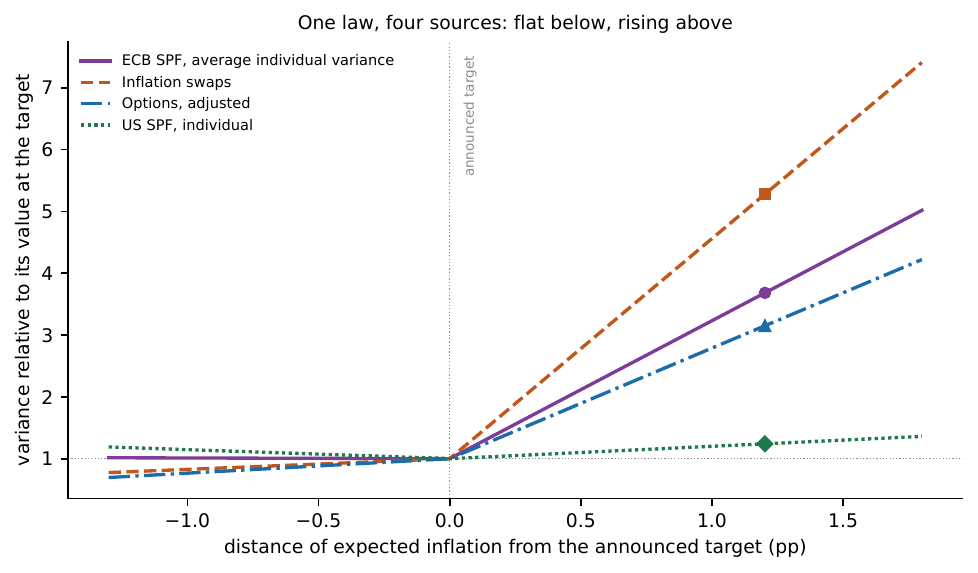}
\caption{\textbf{One law, four sources.} Each source's fitted arms specification
\eqref{eq:law}, divided by its own intercept, so every line equals one at the announced
target. The vertical scale is therefore a \emph{shape}, not a magnitude: slopes cannot be compared across sources whose variances are in
different units. What the figure asks is whether each source is flat to the left of the
target and rising to the right of it. Script: \texttt{fig04\_law\_across\_sources.py},
which estimates the two survey sources and draws the two market sources from the
coefficients printed in Table~\ref{tab:armsbysource}.}
\label{fig:lawfour}\end{figure}

\begin{figure}[!htbp]\centering
\includegraphics[width=0.75\linewidth]{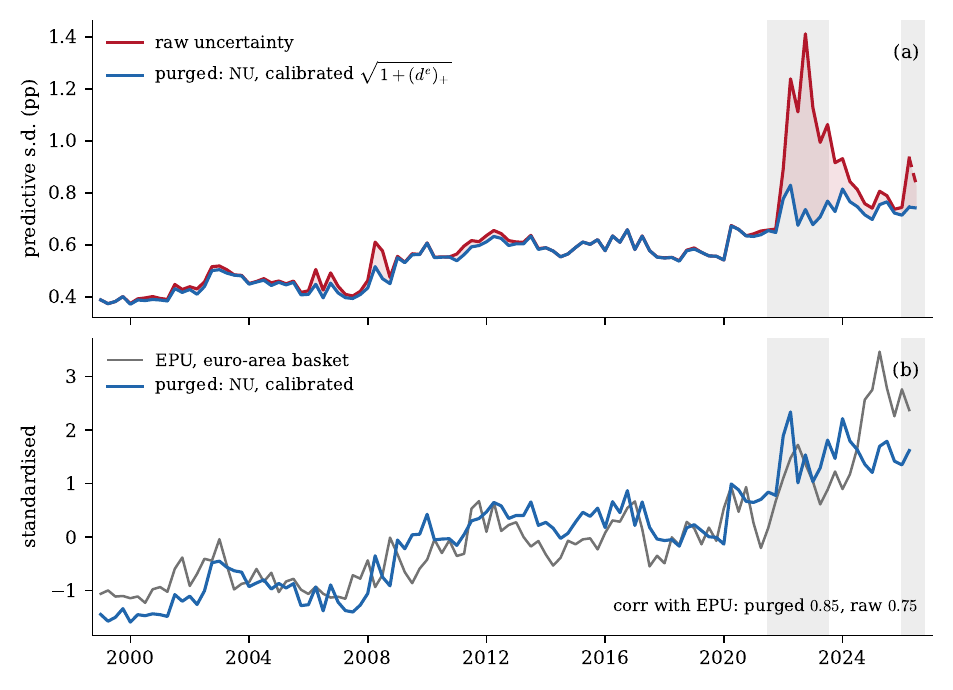}
\caption{\textbf{The purge in the euro-area series.} \emph{Panel (a):} raw
survey uncertainty---the round mean of individual one-year predictive standard
deviations---and its purged counterpart $\mathrm{NU}$, quarterly,
1999Q1--2026Q3. Here the denominator of \eqref{eq:NU_definition} is
\emph{calibrated} at unit coefficients, $\sqrt{1+\dpl}$, rather than estimated,
and it loads on the above-target arm alone. The choice is deliberate. An
estimated envelope is a property of the sample it is fitted on: its arms are
identified by the above-target rounds, which every window contains in a
different number and depth---in this survey the 2021--23 surge supplies
nearly all of the wide gaps---so a series purged with fitted coefficients is a
function of the observation period, rewritten back to its first round whenever
the window is extended. A purge meant to be comparable across vintages, across
sources and between studies has to fix its rule once, independently of the
sample it will be applied to; the unit calibration does that, at the cost of
the size of the correction, which the fitted envelope of
Section~\ref{sec:fitlaw} sets at $b^{s}_+/a^{s}=2.23$ where a regression
coefficient is at stake. The figure is therefore the correction in its
simplest reading, and nothing drawn in it depends on a fitted coefficient. The
two series coincide while the consensus forecast sits at the target and
separate where the distance opens; the gap between them is the component the
law attributes to the anchor. The last round postdates the sample the
paper's fitted results use and is drawn dashed. \emph{Panel (b):} that same purged
series against the euro-area economic policy uncertainty index of
\citet{baker2016measuring} (average over the available euro-area member
indices; both standardised over their overlap, 1999Q1--2026Q2). For the
calibrated series drawn here the correlation with the index is $0.85$, against
$0.75$ for the raw series: removing the arithmetic of the distance from target
costs the measure nothing in its agreement with a proxy built from newspaper
text, which shares none of the survey's construction. Script:
\texttt{fig05\_nu\_purge.py}.}
\label{fig:nu}\end{figure}

\begin{figure}[!htbp]\centering
\includegraphics[width=0.95\linewidth]{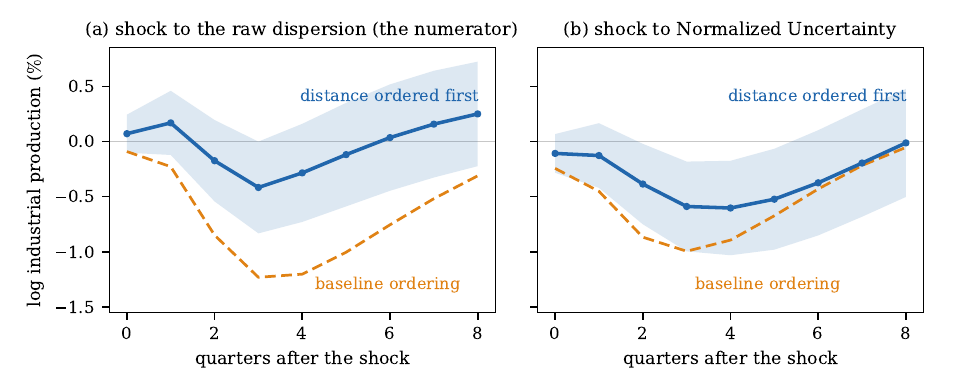}
\caption{\textbf{The level, and the uncertainty: industrial production after an
uncertainty shock, with and without a distance control, 2000--2019.} Orthogonalised
responses of log industrial production (percent) to a one-standard-deviation shock in
the quarterly euro-area system of Section~\ref{sec:nu} (policy-uncertainty index, shock
variable, HICP inflation, deposit facility rate, unemployment, log industrial
production; recursive identification, lag order by information criterion), estimated on
the sample ending in 2019Q4. Dashed: the baseline ordering, shock variable second.
Solid, with ninety-percent asymptotic bands: the above-anchor distance $(\pi_t-2)^{+}$
ordered ahead of the shock variable---the recursive control for the distance.
\emph{Panel (a):} the shock is the raw survey dispersion---the round mean of individual
one-year predictive standard deviations, the numerator of \eqref{eq:NU_definition};
the distance control removes the response almost entirely (cumulative $-6.1\%$ to
$-0.4\%$). \emph{Panel (b):} the shock is Normalized Uncertainty, on panel (a)'s scale; the control reduces the response from $-4.6\%$ to $-2.8\%$, which stays away from zero at ninety percent through quarters two to five. No causal reading is
claimed. Script: \texttt{fig06\_ip\_response.py}.}
\label{fig:rawnu}\end{figure}

\begin{figure}[tp]
\centering
\includegraphics[width=\textwidth]{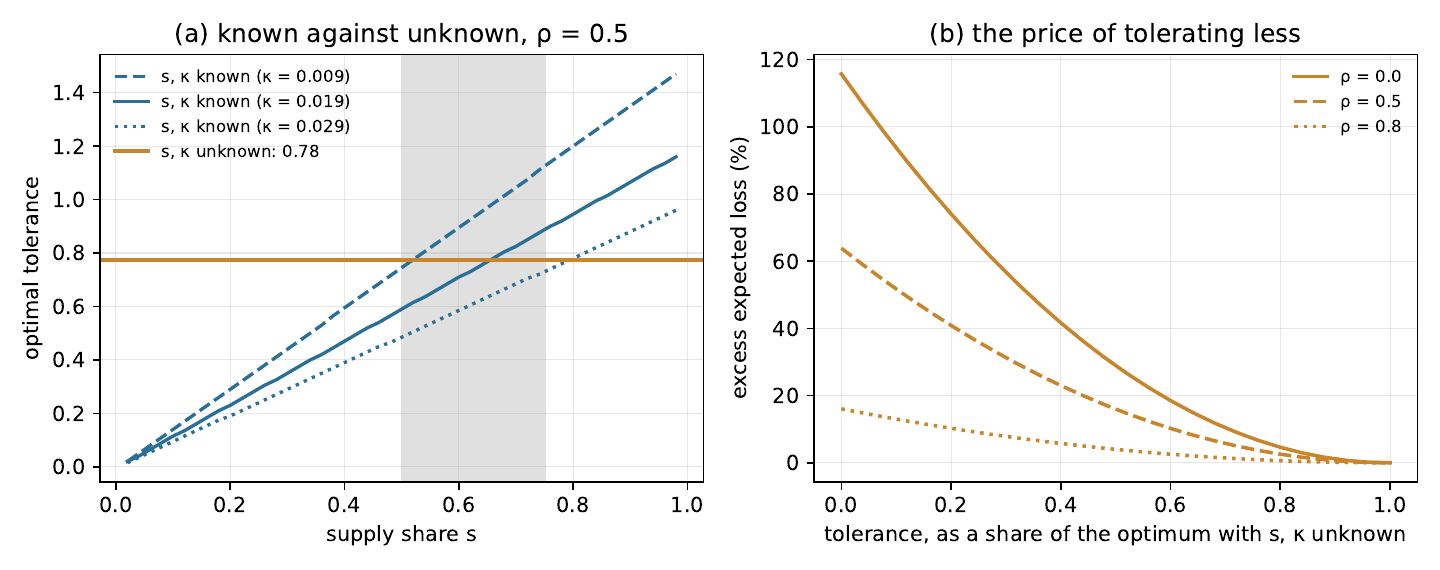}
\caption{No divine coincidence: tolerance with the slope and the shocks unobserved (Appendix~\ref{app:dsge}). \emph{Panel (a):} the optimal tolerance if the supply share and the slope were known, against $s$, at the ends and the midpoint of the published slope range, and the optimal tolerance when neither is known (horizontal line); quarterly persistence $0.5$; the shaded interval is the published range of $s$.  \emph{Panel (b):} the excess expected loss of tolerating less than the optimum with $s$ and $\kappa$ unknown. Galí's calibration; high-gain limit of the rule. Script: \texttt{fig07\_tolerance\_dsge.py}.}
\label{fig:dsge}
\end{figure}

\begin{figure}[!htbp]\centering
\includegraphics[width=0.74\linewidth]{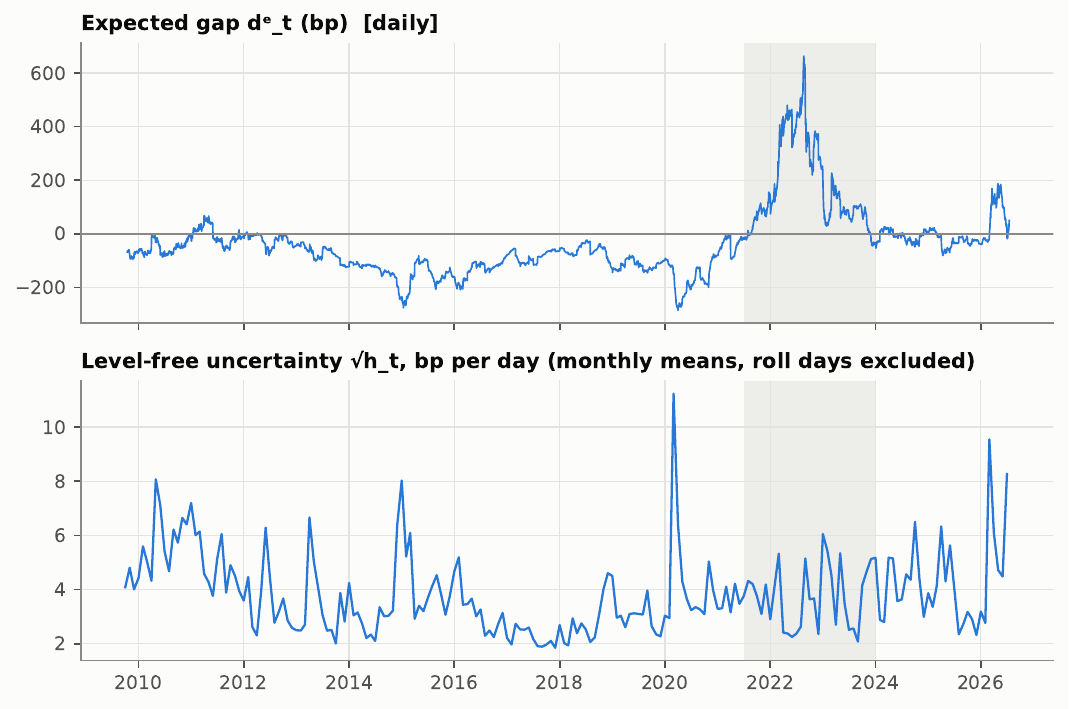}\\[0.3em]
\includegraphics[width=0.74\linewidth]{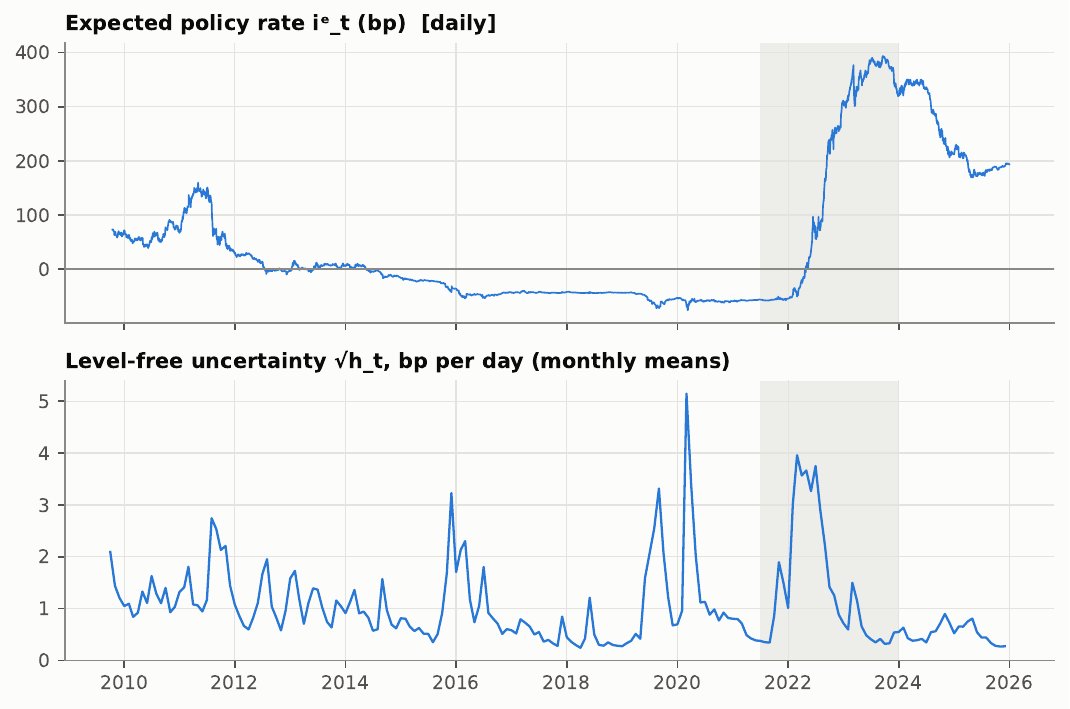}
\caption{\textbf{Level-free uncertainty: what remains of market volatility once the level is removed.} Model \eqref{eq:lgm-mult}. \emph{Top pair:} the expected gap $d^e_t$ (one-year inflation swap minus $2$, bp) and its level-free uncertainty $\mathcal U^e_t=\sqrt{h_t}$, bp per day, monthly means, roll days excluded. \emph{Bottom pair:} the expected policy rate $i^e_t$ (one-year \texteuro STR swap, bp) and its level-free uncertainty. Shaded: 2021:07--2023:12. Once the level factor $\Lambda_t$ is removed, the 2022 surge largely disappears from the expected gap's uncertainty (2022 mean $3.3$ bp a day, against $15.4$ for the total conditional standard deviation), whose peak is March 2020; the expected rate's level-free uncertainty is low after 2022 ($0.6$ against $4.8$ in 2023), the rise in rate volatility being attributed to the level of rates.}
\label{fig:lgm-levelfree}
\end{figure}

\begin{table}[!htbp]\centering
\caption{The arms specification, source by source}
\label{tab:armsbysource}
\begin{tabular}{lrrrrr}
\toprule
 & $n$ & $a$ & $b_-$ & $b_+$ & $R^2$ \\
\midrule
ECB SPF, average individual variance $W_t$ & $109$ & $0.383$ & $+0.004$ & $+0.855$ & $0.705$ \\
                              &         &         & $(0.05)$ & $(12.22)$ &        \\
ECB SPF, disagreement $D_t$   & $109$   & $0.038$ & $+0.113$ & $+1.028$ & $0.719$ \\
                              &         &         & $(2.36)$ & $(8.93)$ &        \\
ECB SPF, total mixture variance $T_t$ & $109$ & $0.420$ & $+0.117$ & $+1.882$ & $0.789$ \\
                              &         &         & $(1.01)$ & $(13.69)$ &        \\
Inflation swaps, 2y, 63d RV   & $5{,}421$ & $0.451$ & $-0.078$ & $+1.606$ & $0.496$ \\
                              &         &         & --- & --- &        \\
\quad non-overlapping windows & $87$ & $0.395$ & $-0.015$ & $+1.664$ & $0.461$ \\
                              &         &         & $(-0.1)$ & $(6.4)$ &        \\
Inflation options, adjusted   & $3{,}941$ & $0.380$ & $-0.089$ & $+0.680$ & $0.359$ \\
                              &         &         & $(-0.78)$ & $(7.85)$ &         \\
US SPF, core CPI, individual  & $1{,}783$ & $0.314$ & $+0.046$ & $+0.063$ & $0.005$ \\
                              &         &         & $(0.59)$ & $(2.67)$ &          \\
\bottomrule
\end{tabular}
\vspace{0.3em}
\parbox{\linewidth}{\footnotesize \textit{Notes:} $t$-statistics in parentheses.
Each row fits $V=a+b_-(-d)_++b_+(d)_+$ on that source's own distance measure; the survey
rows use the forecast gap, the market rows the level gap. Inference is each source's own:
the survey rows HAC(4); the daily swap row reports no $t$-statistics: because its sixty-three-day windows
overlap, daily HAC inference overstates precision (it returns $t\approx25$); the honest
inference uses the $87$ non-overlapping windows of the row beneath it, where the same arms
are $+1.664$ with $t=6.4$ (block-bootstrap $95\%$ interval $[1.24,\,2.26]$) and $-0.015$
with $t=-0.1$; the options row on the premium-adjusted daily series; the US row clustered by round. \emph{Slopes are
not comparable across rows}---the variances are in different units; only the shape is. The
three ECB rows are the objects of Section~\ref{sec:fitlaw}: the law is estimated on the average
individual variance; the total is what aggregate data deliver and is the object of the US
announcement regressions, the market rows and every macro application; disagreement is the
complement. The ratio $b_+/a$ that normalizes the corrected series where it is estimated is that of the average individual variance, $0.855/0.383=2.23$. The level of the survey arms depends on where the mass in the open top bin is
placed; the bracket is in \citet{vansteenberghe2026uncertain}. The
US SPF row is the individual-level fit after January 2012; the aggregate US fit is
symmetric because disagreement, which is symmetric there, dominates it
(Section~\ref{sec:law}). The options row is the premium-adjusted series; the same days
fitted on raw implied variance return inverted arms, the case for the adjustment
(Section~\ref{sec:law}). Script: \texttt{tab01\_arms\_by\_source.py}, which estimates the
survey rows; the swap and option rows rest on licensed data and are carried as printed.}
\end{table}

\begin{table}[!htbp]\centering
\caption{Uncertainty and the price of credit: raw versus purged}
\label{tab:credit}
\begin{tabular}{lcccc}
\toprule
 & (1) alone & (2) $+$ macro & (3) $+$ bank FE & (4) $+$ all FE \\
\midrule
raw standard deviation $\sqrt V$ & $0.0149^{**}$ & $0.0086^{*}$ & $0.0005$ & $0.0004$ \\
                                 & $(2.42)$ & $(1.87)$ & $(0.51)$ & $(0.34)$ \\
\quad deposit facility rate      & --- & $0.0189^{***}$ & $0.0138^{***}$ & $0.0137^{***}$ \\
                                 &     & $(3.27)$ & $(3.26)$ & $(3.26)$ \\
\addlinespace
Normalized Uncertainty $\mathrm{NU}$ & $0.0171^{***}$ & $0.0084^{***}$ & $0.0090^{***}$ & $0.0089^{***}$ \\
                                 & $(3.51)$ & $(2.76)$ & $(3.98)$ & $(4.04)$ \\
\quad deposit facility rate      & --- & $0.0195^{***}$ & $0.0110^{***}$ & $0.0109^{***}$ \\
                                 &     & $(3.24)$ & $(3.22)$ & $(3.23)$ \\
\bottomrule
\end{tabular}
\vspace{0.3em}
\parbox{\linewidth}{\footnotesize \textit{Notes:} the specification is
\eqref{eq:creditspec}: $71{,}069$ French overdraft facilities to non-financial
corporations---AnaCredit France, accessed at the Banque de France---new business,
September 2018 to March 2026; $60$ bank groups, $31$ quarters. Column~(1) is the
uncertainty measure alone; column~(2) adds the debtor default probability, the deposit
facility rate and log industrial production; column~(3) adds bank fixed effects;
column~(4) adds sector, size and location effects. The two uncertainty measures are
substituted one at a time---each block is its own set of four regressions, never a joint
one---and the deposit facility rate row beneath each measure is that regression's own
coefficient on the rate. $t$-statistics in parentheses, clustered by bank. The
uncertainty measure varies only by quarter, and the contrast survives every treatment of
that fact: by quarter, the fixed-effects columns read $t=4.88$ and $4.98$ for $\mathrm{NU}$ against $0.44$ and $0.31$ for the raw standard deviation (in the no-control columns (1)--(2), $2.60$ and $1.88$ against $2.34$ and $1.90$); two-way bank-quarter clustering gives $3.43$ against $0.30$ in column~(4); and a restricted wild-cluster (Rademacher) bootstrap on the $31$ quarters, $999$ draws, gives $p<0.001$ for $\mathrm{NU}$ against $p=0.783$ for the raw standard deviation in column~(4), and $p<0.001$ against $0.682$ in column~(3). The raw forecast variance behaves like its square
root throughout (column-(4) $t=-0.28$ by bank), and the additive-residual construction---the
same specification with the subtraction in place of the ratio---carries $+0.0062$ ($t=4.15$; wild-quarter $p=0.004$) in column~(4).
Appendix~\ref{app:ngu} runs the same specification on the growth analog of the
correction and on the orthogonalized components. The
dependent variable is the annualised agreed rate, entered as a decimal fraction with a
mean of $0.0889$ on this sample; the regressors are standardised but the dependent
variable is not, so the column-(4) coefficient of $0.0089$ is $0.89$ percentage points
per standard deviation of the purged measure, against a mean rate of $8.9\%$. Rebuilding the densities from the raw survey files and bracketing the survey's open top bin between a point mass at its edge and a uniform density to the era's realized maximum \citep{vansteenberghe2026uncertain} moves the column-(4) $\mathrm{NU}$ coefficient only within $[0.0070,\,0.0083]$, $t$ between $3.96$ and $4.04$.
$^{***}p<0.01$, $^{**}p<0.05$, $^{*}p<0.1$.}
\end{table}

\begin{table}[tp]
\centering
\caption{Cross-country growth regressions, raw and purged inflation volatility}
\label{tab:crosscountry}
\begin{tabular}{lccccc}
\toprule
 & (1) & (2) & (3) & (4) & (5) \\
\midrule
Mean inflation $\bar\pi$ & $-0.0170^{***}$ &  &  & $-0.0114$ & $-0.0099^{*}$ \\
 & $(-3.52)$ &  &  & $(-1.28)$ & $(-1.65)$ \\
Raw $\sigma_\pi$ &  & $-0.0194^{***}$ &  & $-0.0082$ &  \\
 &  & $(-2.60)$ &  & $(-0.83)$ &  \\
Purged $\sigma_\pi/\sqrt{1+|\bar\pi|}$ &  &  & $-0.2032^{***}$ &  & $-0.1425^{**}$ \\
 &  &  & $(-3.28)$ &  & $(-2.18)$ \\
\midrule
Observations & 241 & 241 & 241 & 241 & 241 \\
Countries & 94 & 94 & 94 & 94 & 94 \\
$R^2$ & 0.602 & 0.600 & 0.604 & 0.603 & 0.608 \\
Period fixed effects & Yes & Yes & Yes & Yes & Yes \\
Country-clustered SE & Yes & Yes & Yes & Yes & Yes \\
Wild-cluster bootstrap $p$ &  &  &  & $0.43$ & $0.032$ \\
\bottomrule
\end{tabular}
\vspace{0.2em}
\begin{minipage}{0.90\linewidth}\footnotesize
\textit{Notes:} Each column is one growth regression on \citeauthor{barro1995inflation}'s panel
over his own periods---1965--75, 1975--85 and 1985--90; the dependent variable is the country's
average annual growth of real GDP per capita over the period, in 1985 international prices. Every
column includes his full control set, period fixed effects, and standard errors clustered by
country; $t$-statistics are in parentheses. The columns differ only in which summary of the
inflation experience enters: \textbf{(1)} mean inflation $\bar\pi$ alone; \textbf{(2)} the raw
within-period standard deviation $\sigma_\pi$ alone; \textbf{(3)} the purged measure alone;
\textbf{(4)} mean inflation together with the raw standard deviation; \textbf{(5)} mean inflation
together with the purged measure. Comparing~(4) with~(5) is the point of the exercise: the raw
standard deviation is virtually zero conditional on the level, as \citeauthor{barro1995inflation}
reports, and masks it; the purged measure is significant at five percent and the level's own
coefficient sharpens beside it.
\emph{Raw} is $\sigma_\pi$; \emph{purged} divides it by $\sqrt{1+|\bar\pi|}$. That deflator uses no
inflation target---no country in these periods had announced one---and contains nothing estimated:
it is $\sqrt{1+|\bar\pi|}$ exactly, with $\bar\pi$ in percentage points, the units in which
the envelope of \eqref{eq:var_decomp} is fitted. The dependent variable is a fraction: in column~(5) a
one-standard-deviation rise in the purged measure is worth $0.28$ percentage points of
annual growth, against $0.15$ for the raw measure in column~(4). The wild-cluster row
reports restricted Rademacher bootstrap $p$-values ($999$ draws, clustered by country)
for the volatility term of columns~(4) and~(5).
The controls are \citeauthor{barro1995inflation}'s own: initial income, male and female school
attainment at the secondary and higher levels, the interaction of initial income with human
capital as he defines it, life expectancy, fertility, government consumption net of education and
defence, public education spending, the black-market premium, the terms-of-trade change, the
investment ratio, and the political-rights index with its square. The one control of his we cannot
match is the rule-of-law index, which is not in the \citeauthor{barro1995inflation}--Lee
distribution and which he himself uses among his instruments. Inflation moments are period means
and standard deviations of the annual log change in the consumer price index
\citep{muller2025global}, untrimmed, as in his own exercise. Script: \texttt{tab03\_cross\_country.py}.
$^{***}p<0.01$, $^{**}p<0.05$, $^{*}p<0.1$.
\end{minipage}
\end{table}

\begin{table}[!htbp]\centering\small
\caption{The one-sided law in three conditional variances}
\label{tab:lgm-law}
\resizebox{\linewidth}{!}{%
\begin{tabular}{lccc}
\toprule
 & Realized inflation & Expected inflation & Expected policy rate\\
 & monthly, additive \eqref{eq:lgm-add} & daily, \eqref{eq:lgm-mult} & daily, \eqref{eq:lgm-mult}\\
\midrule
Above the centre & $\hat\delta_+=7.97$ & $\hat\eta_+=1.04$ & $\hat\eta_+=1.36$\\
\quad LR (zero) & $13.0$ ($p=0.0002$) & $36.5$ ($p<0.001$) & $19.4$ ($p<0.001$)\\
Below the centre & $\hat\delta_-=0.11$ & $\hat\eta_-=-0.08$ & $\hat\eta_-=-1.25$\\
\quad LR (zero) & $0.14$ ($p=0.35$) & $0.2$ ($p=0.67$) & $1.0$ ($p=0.31$)\\
\addlinespace
Episode dummy in the variance & $7.89$ ($p=0.005$) & $0.75$ (LR $17.8$) & $1.52$ (LR $7.7$)\\
Early sample & $5.21$ ($p=0.097$), to 2020 & $1.05$ ($p=0.099$), to mid-2021 & $2.23$ (LR $32.4$), to mid-2022\\
Oil in mean and variance & $4.61$ ($p=0.0002$) & -- & --\\
Core; ex energy & $0.00$; $2.76$ ($p=0.02$) & -- & --\\
\addlinespace
Sample & 1999:02--2026:10 & 2009:10--2026:07 & 2009:10--2025:12\\
Observations & $333$ & $4{,}171$ & $4{,}230$\\
\bottomrule
\end{tabular}}
\vspace{0.3em}
\parbox{\linewidth}{\footnotesize \textit{Notes:} realized inflation is month-on-month HICP, seasonally adjusted and annualized, minus $2$; the level enters through the AR(6) prediction $\mu_t$, centre $0$, and the episode dummy covers 2021:07--2023:12. Expected inflation is the one-year inflation swap minus $2$ (centre $0$), roll days missing; the expected policy rate is the one-year \texteuro STR swap (centre $0$), with a 2022:07--2023:12 hiking-phase dummy. Student-$t$ errors; likelihood ratios after cross-polishing, $\bar\chi^2$ for $\delta_\pm$ and $\chi^2$ for $\eta_\pm$, which are unrestricted in sign. Oil in the mean is the Brent price in euros, months $t$ and $t-1$; in the variance, Brent's realized variance in month $t-1$ (window 1999:04--2026:09). Core excludes energy, food, alcohol and tobacco. Robustness, alternative forms and forced ARCH in Appendix~\ref{app:lgm}.}
\end{table}

\begin{table}[!htbp]\centering\small
\caption{The level effect on one-year-ahead inflation variance: survey, market and realized}
\label{tab:lgm-compare}
\resizebox{\linewidth}{!}{%
\begin{tabular}{llccc}
\toprule
Source & What is measured & At target & Slope per point of overshoot & Below\\
\midrule
ECB-SPF \eqref{eq:lawsurvey} & width of forecasters' densities & $0.38$ & $0.86$ & flat\\
Expected inflation, \eqref{eq:lgm-mult} & volatility of market revisions & $0.28$ & $0.9$ (first point; convex) & flat\\
\midrule
ECB-SPF realized shadow \eqref{eq:lawrealized} & squared AR(1) innovations & -- & $1.46$ & flat\\
Realized inflation, \eqref{eq:lgm-add} & conditional variance & $0.74$ & $1.7$ on impact, $4.6$ steady state & flat\\
\quad with oil controlled & & & $1.0$; $2.9$ & \\
\quad ex energy & & & about $0.6$ & \\
\bottomrule
\end{tabular}}
\vspace{0.3em}
\parbox{\linewidth}{\footnotesize \textit{Notes:} all slopes in pp$^2$ of one-year-ahead inflation variance per pp of expected overshoot. The market's daily conditional variances are cumulated over $252$ trading days: at target, the median $h$; over the first point, $0.28\times(2^{2\hat\eta_+}-1)$. Daily revisions are serially correlated (variance ratio about $3$ at one year near target; \citealp{lo1988stock}), so the market's one-year revision variance at target is nearer $0.8$, and the levels are not comparable. Realized inflation: $\hat\delta_+$ times the sum of squared AR(6) weights of the twelve-month average ($0.22$), divided by $1-\hat\beta_h$ in steady state; with oil and ex energy, the headline weights are used. The realized shadow is the quarterly envelope of Section~\ref{sec:realized} carried to the survey's horizon and units. Orders of magnitude, not a test: the three rows measure different objects.}
\end{table}

\begin{table}[!htbp]\centering\small
\caption{The predictive link: the market's expected gap and next month's realized inflation}
\label{tab:lgm-theta}
\resizebox{\linewidth}{!}{%
\begin{tabular}{lcc}
\toprule
 & Headline & Core\\
\midrule
Maximum likelihood, full model \eqref{eq:lgm-aug} & $1.07$ ($p=0.005$) & $0.24$ ($p=0.004$)\\
\quad level alone & $0.89$ ($p=0.038$) & --\\
\quad lagged oil, exchange-rate and energy controls & $0.73$ ($p=0.076$) & $0.33$ ($p=0.001$)\\
\quad 2021--23 dummies in mean and variance & $1.18$ ($p<0.001$) & $0.19$ ($p=0.10$)\\
\quad admissible forced variance memory & $p\le0.020$ & --\\
\addlinespace
IVX \citep{kostakis2015robust}, all variants & $1.43$--$1.90$ ($p\le0.004$) & --\\
IVX with controls, weighted & $1.65$ ($p=0.001$) & $0.37$--$0.41$ ($p\le0.035$)\\
Block bootstrap, weighted, $90\%$ interval & $0.96$ [$0.41$, $1.71$] & $0.34$ [$0.10$, $0.80$]\\
\addlinespace
Out of sample 2016--2026: MSPE ratio (Clark--West $p$) & $0.974$ ($0.11$) & $0.906$ ($0.005$)\\
\quad with oil controls & $0.993$ ($0.23$) & $0.905$ ($0.009$)\\
\addlinespace
Uncertainty in the mean, $\varpi_{\mathcal U}$ (headline) & \multicolumn{2}{c}{$2.71$ ($p=0.10$)}\\
Uncertainty in the variance, $\delta_{\mathcal U}$ (headline) & \multicolumn{2}{c}{$26.9$ ($p=0.048$); $0.0$ with oil controls, on core and ex energy}\\
\bottomrule
\end{tabular}}
\vspace{0.3em}
\parbox{\linewidth}{\footnotesize \textit{Notes:} $\theta_e$ in pp of realized annualized month-on-month inflation per pp of last month's expected gap (the one-year swap minus $2$, monthly mean); 2009:11--2026:01, $n=195$ ($194$ for IVX). Likelihood ratios after cross-polishing; IVX instruments $d^e_{t-1}$ with a mildly integrated filter of its own changes \citep{phillips2009econometric}, heteroskedasticity-robust, with AR lags and $\mathcal U^e$ as controls; ``weighted'' divides by the model's conditional standard deviation. Bootstrap: twelve-month moving blocks. Out of sample: recursive OLS from 2009:11, one-month-ahead forecasts from 2016:01, AR(6) with and without $d^e_{t-1}$, one-sided \citet{clark2007approximately} test. The core model's oil controls are those of the headline model.}
\end{table}

\begin{table}[!htbp]\centering\small
\caption{An external uncertainty shock? Average response over the first year}
\label{tab:lgm-oil}
\resizebox{\linewidth}{!}{%
\begin{tabular}{llccc}
\toprule
Block & Measure & Oil-specific $o_t$ & Global $f^G_t$ & Oil supply news\\
\midrule
1 Inflation level & HICP headline (\%) & $-0.06$ ($0.34$) & $+0.05$ ($0.60$) & $+0.37$ ($0.004$)\\
2 Inflation volatility & absolute innovation, headline (pp) & $+0.02$ ($0.046$) & $+0.01$ ($0.053$) & $-0.00$ ($0.85$)\\
3 Expected inflation & 1y1y swap (bp) & $-2.2$ ($0.22$) & $-4.9$ ($0.043$) & $+11.4$ ($0.001$)\\
4 Its uncertainty & option-implied s.d., one year (bp) & $+0.8$ ($0.65$) & $+4.1$ ($0.021$) & $+0.3$ ($0.95$)\\
5 Policy rate & deposit facility rate (bp) & $-0.1$ ($0.98$) & $-9.0$ ($0.067$) & $+11.3$ ($0.14$)\\
6 Its variance & $|\Delta$DFR$|$ within month (bp) & $+0.3$ ($0.71$) & $+2.1$ ($0.034$) & $-0.1$ ($0.95$)\\
7 Expected policy rate & one-year OIS (bp) & $-1.3$ ($0.80$) & $-4.0$ ($0.59$) & $+10.7$ ($0.11$)\\
8 Its uncertainty & GARCH s.d.\ of daily $\Delta$OIS (bp) & $+0.06$ ($0.45$) & $+0.25$ ($0.031$) & $+0.09$ ($0.64$)\\
\midrule
\multicolumn{2}{l}{Smallest Holm-adjusted $p$, 24 measures} & $0.43$ & $0.50$ & \\
\bottomrule
\end{tabular}}
\vspace{0.3em}
\parbox{\linewidth}{\footnotesize \textit{Notes:} monthly local projections, 2008:06--2026:07, $y_{t+h}-y_{t-1}=\beta_h\,\text{shock}_t+\text{controls}_t+e_{t+h}$; the entry is the average of $\beta_h$ over $h=0,\dots,12$ for a one-standard-deviation shock, Newey--West $p$-value in parentheses. Controls: two lags of the shock, of $\Delta y$, of $\Delta\log$ Brent and of $\Delta\log$ VIX; the oil supply-news shock of \citet{kanzig2021macroeconomic}; month dummies for HICP. $o_t$: monthly sum of daily OVX news net of oil-return arms and a global block (VIX, VXV, gold volatility, policy uncertainty); $f^G_t$: first principal component of the VIX, VXV and gold-volatility changes, with the same OVX impulse as $o_t$. One measure per block; the twenty-four measures, the broadest purge and the subsamples are in Appendix~\ref{app:lgm}. Option-implied s.d.\ from euro inflation caps and floors; rate-option implied volatility is not in our data.}
\end{table}

    \clearpage
    \appendix
    \section{Proofs and derivations}
\label{app:derivations}

This appendix documents Sections~\ref{sec:results}, \ref{sec:evidence} and~\ref{sec:living}. Appendices~\ref{app:optlambda}--\ref{app:clocks} and~\ref{app:meanside} derive the results of Section~\ref{sec:results}: the policy block, the three walls, the proof of Lemma~\ref{prop:oe}, the uncertainty law, and the absence of a mean-side mechanism; Appendix~\ref{app:benefit} reads the professionals' longer-term forecasts that Appendix~\ref{app:clocks} uses. Appendices~\ref{app:us}, \ref{app:ngu} and~\ref{app:abel_replication} support the measurement of Section~\ref{sec:evidence}. Appendices~\ref{app:budget}, \ref{app:dsge} and~\ref{app:spd_panel} support Section~\ref{sec:living}: the band and its budget, tolerance when the slope and the shocks are unobserved, and the perceived-rule band.

\subsection{The policy block: optimum, rule, and reduced form}
\label{app:optlambda}

\paragraph{The share.} Take the period loss $\mathcal L=d^2+\vartheta y^2$ and the Phillips curve static approximation $d=\kappa y+u$ for a given cost-push realisation $u$. Substituting and differentiating,
\[
\frac{\partial}{\partial y}\Big[(\kappa y+u)^2+\vartheta y^2\Big]=0
\quad\Longrightarrow\quad
y^{\mathrm{opt}}=-\frac{\kappa u}{\vartheta+\kappa^2},
\qquad
d^{\mathrm{opt}}=\frac{\vartheta\,u}{\vartheta+\kappa^2} .
\]
The optimal gap is therefore a fixed fraction of the shock, Equation \eqref{eq:lstar-model}, with the micro-founded
weight $\vartheta=\kappa/\ep$ of \citet{Gali2015}, with $\ep$ the elasticity of substitution; \eqref{eq:lstar-model} is decreasing in both $\kappa$ and $\ep$: a steeper Phillips curve or a less distorted economy both call for less accommodation. Equation
\eqref{eq:lstar-model} is the partial cost-push accommodation of
\citet{clarida1999science} and \citet{woodford2003interest}, the tolerated
deviation of \citet{svensson1997inflation}, and the finite optimal response to
supply-driven inflation of \citet{hofmann2026targeted}. Published slopes put the share
high. 

\paragraph{Persistence and the regime.} Let the cost-push follow an AR(1),
$u_t=\rho_u u_{t-1}+e_t$ with $\rho_u\in[0,1)$, and write $\chi\equiv\ep\kappa=\kappa^2/\vartheta$.
Under discretion the bank cannot commit to future policy, so it takes
$\E_t d_{t+1}$ as given and, with no endogenous state, solves each period
\[
\min_{d_t,\,y_t}\ \tfrac12 d_t^2+\tfrac12\vartheta y_t^2
\quad\text{s.t.}\quad d_t=\beta\E_t d_{t+1}+\kappa y_t+u_t .
\]
The first-order condition is the targeting rule $\kappa d_t+\vartheta y_t=0$, so
$y_t=-(\kappa/\vartheta)\,d_t$. Substituting it into the Phillips curve,
\[
d_t=\frac{\beta}{1+\chi}\,\E_t d_{t+1}+\frac{u_t}{1+\chi} .
\]
Since $\beta/(1+\chi)<1$, iterating forward gives the unique bounded solution, and
$\E_t u_{t+j}=\rho_u^{\,j}u_t$ sums it:
\[
d_t=\frac{1}{1+\chi}\sum_{j\ge0}\Bigl(\frac{\beta}{1+\chi}\Bigr)^{j}\E_t u_{t+j}
=\frac{u_t}{1+\chi}\sum_{j\ge0}\Bigl(\frac{\beta\rho_u}{1+\chi}\Bigr)^{j} .
\]
The response per unit of cost-push is therefore
\begin{equation}\label{eq:lstar_rho}
\lstar(\rho_u)=\frac{1}{1+\ep\kappa-\beta\rho_u},
\end{equation}
with output $y_t=-\ep\,\lstar(\rho_u)\,u_t$, since $\kappa/\vartheta=\ep$. The
denominator is positive because $\beta\rho_u<1$, and at $\rho_u=0$ the share is
\eqref{eq:lstar-model}. It \emph{rises} with persistence,
$\partial\lstar/\partial\rho_u=\beta\,\lstar(\rho_u)^2>0$, and exceeds one once
$\beta\rho_u>\ep\kappa$. Measured instead in the units of Assumption~\ref{ass:comp}, the gap on impact under full look-through, which is $u_t/(1-\beta\rho_u)$ when the bank does not lean, the same policy accommodates the share
\begin{equation}\label{eq:iota_rho}
\iota(\rho_u)=\frac{1-\beta\rho_u}{1-\beta\rho_u+\ep\kappa}\in(0,1),
\end{equation}
which \emph{falls} with persistence: the bank lets through more inflation per unit of cost-push, and a smaller share of what full look-through would deliver. Discretion amplifies
a persistent cost-push shock, the stabilization bias: leaning is futile against a
shock the bank cannot credibly commit to fight. Under commitment, in the timeless
sense, the gap on impact is $\delta u_t/(1-\beta\delta\rho_u)$, with $\delta\in(0,1)$ the stable
root of $\beta\delta^2-(1+\beta+\chi)\delta+1=0$
\citep{clarida1999science,woodford2003interest}. That root satisfies
$\delta(1+\chi)=1-\beta\delta(1-\delta)<1$. The impact share under commitment is
smaller than \eqref{eq:lstar_rho} if and only if $\delta(1+\chi-\beta\rho_u)<1-\beta\delta\rho_u$,
that is, $\delta(1+\chi)<1$, so the bank lets through a smaller fraction of the shock
at every persistence. The same condition signs the derivative of the difference in
$\rho_u$, so the shortfall widens as the shock persists. The
case for looking through \emph{less} when a shock is persistent is thus a commitment
argument. A Phillips curve that steepens at high inflation
\citep{benigno2023nonlinear} lowers the optimum in the surge. None of these forces
is estimated; the intensity $\lstar$ is their calibrated benchmark, a transitory
shock on a constant slope.

\paragraph{The episode identity.} Write the impulse of Assumption~\ref{ass:comp} through its structural disturbances $(u_t,v_t)$, with $\kappa'\equiv\kappa/(1+\sigma\varphi_y)$ the weight at which a demand disturbance enters the gap once the rule has responded to output. Under full tolerance nothing is offset and the gap is the impulse itself, $d^{(\lambda=1)}_t=u_t+\kappa' v_t\equiv \Pi_t$, which divides as in Assumption~\ref{ass:comp}: the supply component $u_t=s_t\Pi_t$ and the demand component $\kappa'v_t=(1-s_t)\Pi_t$. At the optimum the demand component is offset completely and the supply component is accommodated at $\lstar$, so $d^{\mathrm{opt}}_t=\lstar u_t$. Dividing,
\begin{equation}\label{eq:optlambda_episode}
\lstare=\frac{d^{\mathrm{opt}}_t}{d^{(\lambda=1)}_t}
=\frac{\lstar s_t\Pi_t}{\Pi_t}=s_t\lstar=\frac{s_t}{1+\ep\kappa},
\end{equation}
the structural pair $(\ep,\kappa)$ fixes how much of a supply shock to accommodate; the shock fixes how much of the impulse \emph{is} supply. The ratio is the full-information optimum; a bank acting on its estimates operates $\lstar\hat s_t$, with $\hat s_t$ its estimate of the share, and the distance between the two is what real-time ignorance of the state costs.
The division is exact and holds impulse by impulse, though its two inputs are not: $d^{\mathrm{opt}}_t$ comes from the static period loss and $d^{(\lambda=1)}_t$ from the reduced form with $\beta\E_t d_{t+1}$ dropped and the target residual $\xi_t$ set aside.

\paragraph{The rule, as the data read it.} The rule \eqref{eq:taylor} responds to the deviation from the effective target \eqref{eq:target}, $\hat\imath_t=\varphi_\pi(d_t-m_t-\xi_t)+\varphi_y y_t$. An observer who regresses for an episode $e$ the rate on the announced gap does not see $m_t$ and $\xi_t$ separately. Writing the target's shift as a share of the gap plus a residual, $m_t+\xi_t=\lcbe\,d_t+(\text{residual})$ with $\lcbe\equiv\Cov(m_t+\xi_t,d_t)/\Var(d_t)$ the projection coefficient and the residual written $\xi_t$ again below (it is the policy-rule residual in which the drift lands, together with everything else that shifts the intercept), the rule reads as the gap-scale \emph{operating form} $\pi^{*}_t=\bar\pi+\lcbe d_t+\xi_t$; with $\lcbe$ the bank's operated share held over an episode. The block \eqref{eq:nkpc}--\eqref{eq:taylor} then reads, with the rule alone rewritten:
\begin{equation}\label{eq:phieff}
\hat\imath_t = \varphi_\pi(1-\lcbe)\,d_t
\;-\;\varphi_\pi\,\xi_t
\;+\;\varphi_y\,y_t
\qquad\Longleftrightarrow\qquad
\varphi_\pi^{\mathrm{eff}}=\varphi_\pi(1-\lcbe).
\end{equation}
Look-through is observationally a weaker response to inflation. Two features of \eqref{eq:phieff} carry much of what the paper measures, and a third bounds what it may claim. The residual term is a \emph{policy-rule residual}---anything that shifts the rule's intercept, a moving natural rate, a drifting long-run target, smoothing, lands in $\xi_t$ and is observationally entangled with tolerance. And a second mechanism produces the same effective coefficient: a bank that distrusts the transmission of its instrument---multiplicative uncertainty of variance $\sw$ on the Phillips term---attenuates its response to $\varphi_\pi/(1+\sw)$ \citep{brainard1967uncertainty}, so caution of size $\sw$ behaves exactly like a look-through of size $\lB=\sw/(1+\sw)$. Softening by choice and softening by necessity are the same wedge on $\varphi_\pi$: they compose as $\varphi_\pi^{\mathrm{eff}}=\varphi_\pi(1-\lambda)(1-\lB)$, so any rule-based measurement recovers $\leff=\lambda+\lB-\lambda\lB$ and therefore an \emph{upper bound} on deliberate look-through. The third is the entanglement: the projection entangles accommodation and drift by construction---it loads on whichever persistent component the gap carries, so a rule residual measures tolerance-plus-drift, never tolerance alone (Wall~II, Appendix~\ref{app:walls}).

\paragraph{The reduced form.} One lemma carries the dial from the rule into the law of motion of the gap:

\begin{lemma}[The decomposition is the policy block's reduced form]
\label{lem:nkreduced}
Keep the Phillips curve \eqref{eq:nkpc} forward and impose the static-IS approximation: the system with the look-through rule \eqref{eq:phieff} collapses, with no backward element, to an exact relation in which the dial discounts the expected future. Transposed to the regression frequency---the form in which the gap's law of motion is fitted to data---the same coefficients move to the lag seat (below), and the gap follows
\begin{equation}\label{eq:rho}
d_{t+1}=\rho(\lambda)\,d_t+k_\xi\,\xi_{t+1}+k_u\,(u_{t+1}+\kappa'v_{t+1}),
\qquad
\rho(\lambda)=\frac{\beta}{1+A(1-\lambda)},
\qquad
A\equiv\frac{\kappa\sigma\varphi_\pi}{1+\sigma\varphi_y},
\end{equation}
with $k_u=1/(1+A(1-\lambda))$, $k_\xi=A\,k_u>0$ and $\kappa'\equiv\kappa/(1+\sigma\varphi_y)$: the drift and a composite cost-push-plus-demand impulse; $\rho'(\lambda)>0$: \emph{tolerating more of the gap slows its mean reversion}---mean reversion fastest under full reaction, and $\rho(1)=\beta$ at full tolerance.
\end{lemma}

\paragraph{The reduction.} Impose the static-IS approximation $y_t=-\sigma\hat\imath_t+v_t$ and the look-through rule \eqref{eq:phieff}. Solving the first two for output,
\[
y_t=\frac{-\sigma\varphi_\pi(1-\lambda)d_t+\sigma\varphi_\pi\xi_t+v_t}{1+\sigma\varphi_y} ,
\]
and substituting into the forward Phillips curve gives, with $A\equiv\kappa\sigma\varphi_\pi/(1+\sigma\varphi_y)$ and $\kappa'\equiv\kappa/(1+\sigma\varphi_y)$,
\[
d_t=\frac{\beta}{1+A(1-\lambda)}\,\E_t d_{t+1}
+\frac{u_t+\kappa' v_t+A\xi_t}{1+A(1-\lambda)} ,
\]
the exact forward relation behind \eqref{eq:rho}, with $k_\xi=A/(1+A(1-\lambda))>0$, and
\[
\rho'(\lambda)=\frac{\beta A}{\big(1+A(1-\lambda)\big)^2}>0 :
\]
tolerance slows mean reversion, with $\rho(1)=\beta$ and $0<\rho<1$ on $[0,1]$ whenever $\beta<1$. Setting $\lambda=1$ returns the full-look-through impulse $\Pi_t=u_t+\kappa'v_t$ of Assumption~\ref{ass:comp}, whose supply share $s_t$ enters \eqref{eq:optlambda_episode}. Fully forward, the gap is a jump variable whose autocorrelation is the shocks' own persistence at every dial, so writing \eqref{eq:rho} at the regression frequency is a modelling convention: the comparative static survives by composition, since the loadings of persistent components are the most $\lambda$-elastic. Two approximations are disclosed. The static IS above carries the nominal rate; keeping the expected inflation inside the real rate raises the numerator of $\rho$ by $\kappa\sigma/(1+\sigma\varphi_y)$ and can push $\rho(1)$ above one. Nothing turns on reconciling the two, because \eqref{eq:rho} is used only for the sign of $\rho'(\lambda)$ and its slope, and is never inverted against the data for $\kappa$ or $\varphi_\pi$.

\subsection{Three walls between the data and the optimum}
\label{app:walls}

\paragraph{The walls, and what is knowable.}
Lemma~\ref{prop:oe} is the model's own impossibility, and it does not
stand alone: between any dataset and the optimum's \emph{value} stand three
walls, one for each ingredient of the optimum, stated in full below. The curve: \emph{the Phillips slope is not
identified}. Identification-robust confidence sets for the euro-area slope
are unbounded; the external instruments that do achieve strong first stages
estimate a curve near flat, so better instruments rescue nothing; and the
paper's own internal route collapses in a way exhibited below as evidence
rather than hidden. Wall~I states the failure, and the
history demonstrates it: seventy years of corrections in which each
generation re-attributed what its predecessor had claimed to identify
\citep{vansteenberghe2026seventy}. The rule: the response coefficient
behind \eqref{eq:target} can be unidentified in principle, because the
Taylor principle operates through off-equilibrium threats that equilibrium
histories never display \citep{cochrane2011determinacy}, and the rule's
intercept entangles the dial with natural-rate and target drift---inside a
fully structural model, deliberate accommodation of supply shocks and
exogenous target drift fit half a century of data indistinguishably
\citep{ireland2007changes}. The shock: the episode optimum needs the supply
share of the impulse, and the shock does not carry its label---not in real
time, when the stance must be chosen, and not ex post, where the leading
decompositions of the same surge still differ by some twenty points
\citep{shapiro2022decomposing,bernanke2023caused}. The verdict on a
central bank can therefore be argued but not settled
(Corollary~\ref{cor:verdict}), ex post as well as in real time: the
configurations that would separate the dial from the rule differ only in
counterfactual menus that equilibrium histories never play, and the band of
reaction functions market participants actually hold, measured directly from
conditional scenario surveys, never closes at any date we observe (Appendix~\ref{app:spd_panel}). What a band for the optimum can honestly be is an
\emph{exercise}, and the paper runs it as one (Proposition~\ref{prop:budget},
Appendix~\ref{app:budget}): provisionally accepting the leading published
readings, a reasonable conditional band is $[0.36,0.68]$---for impulses whose
components share a sign, the case those readings support---with three-fifths of its
width---three-quarters of its variance---traceable to the inability to classify
the shock as supply or demand. The intensity entering the band is the
transitory-shock benchmark of \eqref{eq:lstar-model}: under discretion a
persistent shock pushes the optimum \emph{above} one (the stabilization
bias), under commitment it pushes it \emph{below}
(Appendix~\ref{app:optlambda}'s ledger of the two regimes)---so the
benchmark sits between forces that pull opposite ways, one more layer of the
conditionality the band's grade already carries. The number to carry is not the interval but its
grade: conditional on readings that are peer-reviewed and published, that
disagree, and that the paper has no way to certify---a debate no accumulation of
equilibrium data will close.

The rule the data actually see is \eqref{eq:phieff} on the observed gap plus an intercept polluted by the drift---the form the reduction of
Appendix~\ref{app:optlambda} runs on. Equilibrium data therefore deliver the
composite $\varphi_\pi(1-\lcbe)$, the persistence $\rho(\lambda)$ of
\eqref{eq:rho}, and the composite intercept $\varphi_\pi\xi_t$, while the
object of normative interest is $\lstare=s_t\lstar$ with
$\lstar=1/(1+\ep\kappa)$: every ingredient of the optimum sits at least one
layer behind what is estimable. Between the two stand three walls, one per
ingredient---the \emph{curve} behind the intensity $\lstar$ (Wall~I), the
\emph{rule} behind the dial (Wall~II), the \emph{shock} behind the share
$s_t$ (Wall~III). The first two are inherited from the great identification
programs of empirical macroeconomics, whose seventy-year history \citet{vansteenberghe2026seventy}
traces; the third is intrinsic to the
tolerance object itself, and it binds the central bank in real time as well
as the econometrician ex post. None is a small-sample problem; each is a
reason more data of the same kind cannot close the set.

\paragraph{Wall I --- the curve (inherited).}
The intensity $\lstar=1/(1+\ep\kappa)$ inherits every difficulty of the Phillips slope, and the slope is not identified: identification-robust sets for the euro-area $\kappa$ are unbounded, and the external instruments that achieve strong first stages in dynamic form leave the slope flat rather than sharpening it \citep{mavroeidis2014empirical,hazell2022slope}, while the endogeneity of policy flattens the estimable curve precisely when policy is good \citep{mcleay2020optimal}. The failure is constructive, and we exhibit it rather than hide it: attempting to identify $\lambda$ internally---from the cost-push shock's own persistence, the one route that avoids circularity---collapses at the euro-area slope, because the recovered shock is not a shock at all. At a flat $\kappa$ with anchored survey expectations, the Phillips-curve residual is inflation relabeled, $\mathrm{corr}(\hat u_t, d_t)=0.98$, and the route returns infeasible values at every $\kappa$ in the published range.
Replacing survey expectations with realized future inflation breaks the degeneracy ($0.98\to0.21$) but fails in turn---first to errors-in-variables, then, instrumented, to instruments that are inflation lags and hence collinear. Two doors, one wall: on euro-area data the flat curve, not the expectation proxy, is binding. What a steep curve would permit, the euro-area curve denies.

\paragraph{Wall II --- the rule (inherited).}
The dial is revealed only through the operated rule \eqref{eq:phieff}: a composite coefficient one layer behind a $\varphi_\pi$ that is itself contested, and an intercept that carries the drift. The coefficient first. The rule's parameters can be unidentified in the baseline model because off-equilibrium threats are untestable \citep{cochrane2011determinacy}; where they are identified, the verdict has swung with the estimator and the data vintage---forward-looking GMM found a pre-Volcker violation of the Taylor principle that real-time data undid \citep{clarida2000monetary,orphanides2001monetary}, and the case for and against OLS is still being relitigated \citep{carvalho2021taylor}. Whatever $\varphi_\pi$ is, the public does not hold it fixed: the perceived response of the policy rate to inflation, read directly off dealers' scenario answers, fell by two-fifths across the 2024 pivot---from $0.96$ to $0.58$ at matched horizon, Table~\ref{tab:spd_panel} of Appendix~\ref{app:spd_panel}---so the coefficient against which tolerance would have to be measured is itself a moving, perceived object \citep{bauer2024perceptions}.
The intercept is the wall's second face. A natural-rate drift, a target drift, and interest-rate smoothing all land in the composite $\varphi_\pi\xi_t$ of \eqref{eq:phieff} and are observationally entangled with the dial: what reads as tolerance can be a drifting natural rate, a drifting target, or gradualism, and the rule residual cannot say which. Nor is the entanglement an artifact of the reduced form: estimated inside a fully structural New Keynesian model, deliberate supply-accommodation of the target and exogenous target drift are statistically indistinguishable on half a century of United States data \citep{ireland2007changes}.\footnote{Ireland's target (his eq.~10) is a random walk whose innovations load on the supply shocks---the deliberate, ``opportunistic'' component---plus an exogenous innovation. His maximum-likelihood point estimates assign \emph{all} target movement to the deliberate response ($\sigma_\pi=0.0000$), yet the likelihood ratio against the all-exogenous model is $1.62$: ``considerable uncertainty remains about the true source of movements in the Federal Reserve's inflation target, so that one cannot statistically reject the exogenous target model that depicts those movements as purely random'' \citep[p.~1864]{ireland2007changes}. One realized path, two attributions, both maintained---the structure of Corollary~\ref{cor:verdict}, at the unit root, in 2007.}

\paragraph{Wall III --- the shock (intrinsic).}
Even with $\kappa$ and $\varphi_\pi$ in hand, the episode optimum needs the supply share $s_t$, and the shock does not carry its label---not in real time, and not ex post. In real time the wall is the central bank's own: the state that fixes the optimum resolves only after the moment in which the stance must be chosen, so the bank acts against a composition it cannot yet read. And the label never arrives: the leading decompositions of the same surge, published after the fact with full hindsight, still differ by twenty points \citep{shapiro2022decomposing,bernanke2023caused}. Proposition~\ref{prop:budget} (Appendix~\ref{app:budget}) quantifies this wall---it alone carries between two-thirds and four-fifths of the irreducible uncertainty in the optimum, more than the celebrated flat curve. Structural shock identification does not escape this wall; it formalizes it. Sign-restricted identification delivers admissible \emph{sets} of decompositions \citep{uhlig2005what}, the prior on the rotation does work the data cannot undo \citep{baumeister2015sign}, and a historical decomposition inherits the maintained slope and rule---Walls~I and~II re-imported. Nor does the signal-extraction tradition relocate the problem. When a rule must feed on a mismeasured argument, the optimal restricted response shifts weight \emph{off} the noisy variable and \emph{onto} the well-measured one \citep{staiger1997precise,peersman1999taylor}: in \citeauthor{peersman1999taylor}'s estimates, output-gap noise \emph{raises} the optimal inflation weight (from $1.53$ to $1.65$) while the gap response falls, optimally to zero once the gap error's standard deviation passes $1.15$ percent. Measurement error thus pushes the inflation coefficient the wrong way to explain softening---on the inflation leg the well-measured variable is inflation itself---so a mismeasurement account of look-through must locate the noise \emph{inside} the impulse's own composition $s_t$: this wall, not an alternative to it.

\paragraph{The consequence: the verdict is not identified.}
The three walls compound into a single denial. The optimum
$\lstare=s_t\lstar$ needs the curve for the intensity (Wall~I) and the shock
for the share (Wall~III); the stance it would grade is read off the rule
(Wall~II); each ingredient is unidentified---the curve's robust sets
unbounded, the share free on the unit interval, the rule read only through
its composite---and none of these gaps is of the kind more equilibrium data
close. The question the public---and
the bank itself---most wants answered, whether a given accommodation was
warranted, is therefore not merely unanswered but unanswerable on field
data.

\begin{corollary}[The verdict is not identified]\label{cor:verdict}
Fix any equilibrium episode path and any stance reading $\lcb<1$ obtained from a rule benchmark. There exist admissible configurations---$\varphi_\pi$ anywhere on the ridge of \eqref{eq:phieff}, $\kappa$ anywhere the unbounded robust sets allow, $s_t\in[0,1]$, transmission uncertainty $\sw\ge 0$---under which the realized stance is episode-optimal, and others under which it is not, all inducing the same law for every field observable. Any verdict on the realized stance---episode-optimal, too tolerant, too tight---is identified only conditionally on maintained readings of $(s_t,\lstar,\varphi_\pi,\lB)$.
\end{corollary}

\noindent\emph{Sketch.} Optimality is one equation, $\lcb=s_t\lstar+\lB(1-s_t\lstar)$---the same multiplicative composition that produces $\leff$---in objects of which equilibrium observables pin only the product $\varphi_\pi(1-\lambda)$ (Wall~II). With $s_t$ free on $[0,1]$ (Wall~III), $\lstar$ free on the image of the unbounded $\kappa$-sets, which reaches one at the flat limit (Wall~I), and $\lB$ free on $[0,1)$ (the Brainard reading of Appendix~\ref{app:optlambda}), both verdicts are attainable at any $\lcb<1$: setting $\lB=0$ and $s_t\lstar=\lcb$ makes the stance optimal; setting $s_t=0$ and $\lB=0$ makes the same stance pure unwarranted tolerance. Nor does the distance from a full-information benchmark grade the verdict on its own: at $s_t=0$ and $\lB=\lcb$ the stance is warranted caution---optimal while sitting the full $\lcb$ above the benchmark $s_t\lstar=0$---so even that distance separates nothing until the readings are maintained. Observational equivalence along the ridge (Wall~II) completes the argument.

The corollary fixes the paper's grammar of claims, and it says which data
are missing. Three grades run through everything that follows: the gap, its
second moments, and the occupancy of the region above target are
\emph{measured}; the accommodation share is \emph{bracketed only by published readings};
the episode verdict is \emph{undecidable}. And the missing data are not more
of the same: configurations along the ridge are separated only by
observables conditioned on states of announced magnitude and
composition---\emph{counterfactual menus}, in the language of
\citet{healy2026which}---which field histories never contain; what instruments
could construct them, in the laboratory and rarely in the field, is the
measurement program of Section~\ref{sec:measuring_tolerance}.

\paragraph{The walls, against the canonical monitoring solution.}
The profession has had a blueprint for policing exactly this object since \citet{svensson1997inflation}: make the bank's inflation forecast the intermediate target; have the bank reveal its model, information, and judgements; and let the public detect any \emph{implicit} target---his $z_t$, our $\lambda d_t+\xi_t$---by comparing the observed instrument with the one the revealed model prescribes, criticism then enforcing the official target. The mechanism's stated premise is that the public knows the model and its coefficients, and its benchmark inherits them: his optimal rule responds to inflation with $b_1=1/(\alpha_1\beta_2)$, $\alpha_1$ the Phillips slope, so at the flat, unidentified slope of Wall~I the benchmark is unboundedly sensitive to the slope error---and the comparison is in any case a rule comparison, Wall~II. Even executed with true coefficients, the extraction recovers the \emph{composite} deviation, never its split into tolerance and drift: the intercept entanglement of \eqref{eq:phieff}, Wall~II's second face, already present in his notation. He saw the objection---``the lack of knowledge, and resulting disagreement, about the appropriate macroeconomic model \ldots\ may be so substantial as to make both the implementation and monitoring of the target rule proposed here too difficult'' \citep[p.~1138]{svensson1997inflation}---and answered it with the incentive the framework gives the bank to learn its economy; \citet{vansteenberghe2026seventy} is the record of what that learning ran into. The walls therefore bind twice: they stop the econometrician's estimator, and they stop the public's monitor. What survives of the monitoring program, and what completes it, is taken up in Section~\ref{sec:living}.

\subsection{Proof of Lemma~\ref{prop:oe}}
\label{app:oe}

Write $L$ for the label, $L=1$ under the tolerance reading and $L=0$ under the
drift reading, so that $m_t=L\mu_t$ and $\xi_t=(1-L)\mu_t$ at every date and
$\Pr(L=1)=q$. By Assumption~\ref{ass:innov}(b), $L$ is independent of the process
$(\mu_t,\varepsilon_t)$. A signal whose law depends on $(m,\xi)$ only through their
sum is a function of $\mu$ and of noise independent of $L$. The re-basing decisions
of others are such a signal: they depend on the public information and on the
weight $q$, not on $L$ (Section~\ref{sec:model-law}), and so does the innovation
their scatter adds to the error (Assumption~\ref{ass:contagion}). Hence $L$ is
independent of the joint process of $\mu$, the gap $d_t=\mu_t+\varepsilon_t$ and
the signals, in particular of $(\Iset_t,\mu_{t+1},\{d_s\}_{s>t})$.

\emph{(i)} The law of $(\Iset_t,\{d_s\}_{s>t})$ given $L$ does not depend on $L$.
Any integrable function of the path of inflation therefore has the same conditional
expectation under each reading, and under every mixture of them; the predictive
variance of $d_{t+1}$ is such a function. \emph{(ii)} $\Pr(L=1\mid\Iset_t)=\Pr(L=1)=q$
almost surely, at every date. \emph{(iii)} Since $L$ is independent of
$(\mu_{t+1},\Iset_t)$, $\E_t[\xi_{t+1}]=\E_t[(1-L)\,\mu_{t+1}]=(1-q)\,\E_t[\mu_{t+1}]$,
and $\E_t[\mu_{t+1}]=\E_t[d_{t+1}]=\de$ because the error is unpredictable.
\hfill$\square$

\paragraph{Scope, stated as bluntly as the result.}
Three boundaries are worth the same care as the proof. \emph{First}, the
lemma concerns the inflation-gap history. The public also observes the
policy rate and output, and in principle these could separate the readings:
under tolerance a positive cost-push lowers the output gap, under drift a moved
target raises it. In practice the joint record does not settle the question,
because reading the intention off output requires three things the public does
not have. The output gap is not observed, and real-time estimates of it are
revised by about as much as the gap itself \citep{orphanides2002unreliability};
the slope that turns an output path into an inflation path is weakly identified
\citep{mavroeidis2014empirical}---for the euro area, without bound (Wall~I,
Appendix~\ref{app:walls}); and the composition of each impulse is never
revealed, so a small fall in output fits a tolerated supply shock as well as a
drift met by a negative demand shock, and decompositions of the 2021--23 surge
still differ by some twenty points
\citep{shapiro2022decomposing,bernanke2023caused} (Wall~III). A persistent
drift the public expects, moreover, works through $\beta\E_t d_{t+1}$ more than
through output, and the rule's intercept mixes target drift with drift in the
natural rate: inside a fully structural model estimated on inflation, rates and
output together, accommodation of supply shocks and exogenous target drift fit
half a century of US data equally well \citep{ireland2007changes}. The step from
the inflation history to everything the public observes therefore rests on the
state of the evidence, not on this proof. For the inflation history alone, the
assumption that the rule delivers the effective target is sufficient, not necessary. At a finite
$\varphi_\pi$ the cost-push also enters the gap directly, beside the intention built from it, so a drift with the law of the intention no longer gives the gap the same law, already with serially uncorrelated components. With the IS curve kept static, $A$ as in Lemma~\ref{lem:nkreduced}, $v_t$ serially uncorrelated, precise
estimates and AR(1) components, the minimum-state-variable solution is
$d_t=a_u u_t+a_\xi\xi_t+a_v v_t$, with $a_u=(1+A\lstar)/(1+A-\beta\rho_u)$,
$a_\xi=A/(1+A-\beta\rho_\xi)$ and $a_v=\kappa'/(1+A)$: two persistent
components and a white term. The lemma then applies to the loaded
components $a_uu_t$ and $a_\xi\xi_t$, with Assumption~\ref{ass:innov}(b) stated for their laws. The coefficients are the
full-information benchmark, not a solved equilibrium for a public that filters;
a persistent demand disturbance would add a third component. \emph{Second}, the result is only as strong as
Assumption~\ref{ass:innov}(b), and it is worth saying what would break it. With
two AR(1) components of distinct positive persistences and nonzero variances,
the autocovariances generically
reveal the unordered pair of persistences and their variances; knowing which
persistence belongs to the drift---an ordering, such as drifts outlasting supply
shocks---would then label the components and with them the intention, and so
would known, distinct distributional signatures of the two innovations.
Identification of the intention therefore requires an a priori asymmetry
between the two sources, and nothing else does; Assumption~\ref{ass:innov}
imposes none, and the persistence wedge of Section~\ref{sec:announce} is read
as the data selecting a configuration, not as an a priori ordering that would
identify it. \emph{Third}, nothing in the argument uses a distributional form,
stationarity, independence of the two components, a common persistence or a sample
size. The two pure readings are the extreme points of a larger set. A public that also
entertains structures in which the bank both tolerates and drifts, closed under
exchanging the laws of the two components, meets the label switching of mixture
models \citep{celeux2000computational,stephens2000dealing,fruhwirth2006finite}: the
exchange leaves the law of the history unchanged path by path, so the labeling is never
updated, and a prior invariant to the exchange attributes half of the expected
unannounced part of the target to each component, the case $q=\tfrac12$. Nothing in
Section~\ref{sec:model-law} uses that generality.

\subsection{Re-basing, contagion and the uncertainty law}
\label{app:clocks}

\paragraph{The cost of trust.} An agent who keeps the anchored plan expects to lose
$c_iz_t$ per period relative to the re-based one, with $z_t$ the exposure
\eqref{eq:exposure}; one who re-bases pays the flow cost $k_i$ of running her own
reading of the drift. She re-bases if and only if $z_t\ge\tau_i=k_i/c_i$. The first
term of the exposure is what re-basing recovers. Let it recover the drift in full and
a share $\chi\in[0,1]$ of the tolerated part, $\xi_{t+1}+\chi\,m_{t+1}$. In the pure
readings this is $\chi\mu_{t+1}$ with weight $q$ and $\mu_{t+1}$ with weight $1-q$, so
its expectation is $[1-q(1-\chi)]\,\de$. Section~\ref{sec:model-law} takes $\chi=0$:
in the block of Section~\ref{sec:model-oe} a tolerated point comes with an output gap
of $-\ep$ per point, a real loss that re-basing passes on and does not recover, and a
drifted point with $+1/\kappa$ per point. At $\chi=1$ the payoff is a function of
inflation alone and $q$ drops out, as Lemma~\ref{prop:oe}(i) requires; for every
$\chi<1$ the law keeps its form, with $1-q(1-\chi)$ in place of $1-q$. For an
even-handed public, $q=\tfrac12$, and the first term of the exposure is
$\tfrac12\de$: on the scale of the expected gap alone, the threshold is $2k_i/c_i$.
Knowing the composition of the impulse would not change that fraction. The
accommodated component is $\lstar\hat u_t$, so a public that read the supply share
would still need the intensity $\lstar$, which the Phillips slope leaves
unidentified (Wall~I, Appendix~\ref{app:walls}). Clarity about the composition acts
on the identified margin---how large and how persistent an overshoot the public
expects---not on the price of trust per point of it; Appendix~\ref{app:benefit}
reads this across the euro area's four inflation hikes.

\paragraph{Why the stake must be first order, and signed.} A plan set to be right
about inflation under a quadratic loss also carries a second-order cost, the square of
its expected miss on the drift, $x=(1-q)\,\de$. With both costs, and leaving the
coordination term aside, an agent facing $x>0$ re-bases if and only if
$c_i x+x^{2}>k_i$, that is, if and only if
\[
x\;>\;\tfrac12\Bigl(\sqrt{c_i^{2}+4k_i}\;-\;c_i\Bigr),
\]
which is $k_i/c_i$ to first order in $k_i/c_i^{2}$: the threshold of
\eqref{eq:rebase-choice}, and a positive density of $k_i/c_i$ at zero gives
$g(0)>0$. With the quadratic cost alone ($c_i=0$) the threshold is $\sqrt{k_i}$, on
the absolute expected drift. It is two-sided, and for any distribution of attention
costs with a bounded density $f_k$ near zero, $G(\tau)=F_k(\tau^{2})\approx
f_k(0)\,\tau^{2}$, so $g(0)=0$: the law would rise with the square of the overshoot
near the target, a dead band rather than a kink. This is the near-rationality of
\citet{akerlof1985near}: second-order losses do not justify first-order action. The
kink at the announced number therefore requires the first-order stake of
Section~\ref{sec:model-law} (condition K1). Its one-sidedness is the sign of that
stake (condition K4). With the stake $(1-q)|\de|$ the share would be even in the gap
and the law would have two equal arms. With a mass $w$ of claimants and $1-w$ of
debtors who also quote terms, and no coordination across the two groups,
$S=w\,G\bigl((1-q)\dpl\bigr)+(1-w)\,G\bigl((1-q)(-\de)_+\bigr)$, whose arms stand in
the ratio $(1-w)/w$: the lower arm of Section~\ref{sec:evidence} measures that ratio.
Without an announced number, let the reference points $r_i$ of nominal plans have a
distribution $H$; the share is then $\int S(\de-r)\,dH(r)$, which is differentiable
wherever $H$ has no atom, so the kink requires a number that plans share.

\paragraph{Proof of Proposition~\ref{prop:law}.} Extend $G$ by zero to negative
arguments. Agent $i$ re-bases if and only if $z_t\ge\tau_i$, and over a continuum of
agents whose thresholds are drawn independently of the aggregate path the share with
$\tau_i\le z_t$ is $G(z_t)$ (the law of large numbers), which is zero whenever
$z_t\le0$ because $\tau_i\ge0$ and $G(0)=0$.
\emph{(i)} The map $T(S)=G\bigl((1-q)\,\de+\zeta S\bigr)$ sends $[0,1]$ into itself,
and the extended $G$ is Lipschitz with constant $\sup g$, so $T$ is a contraction of
modulus $\zeta\sup g<1$ (Assumption~\ref{ass:rebase}): Banach's theorem gives a unique
fixed point $S(\de)$, continuous in $\de$ because $T$ is jointly continuous and
uniformly contracting, and nondecreasing because $T$ is nondecreasing in $(S,\de)$.
\emph{(ii)} If $\de\le0$, $S=0$ gives $z=(1-q)\,\de\le0$ and $G(z)=0$, so zero is the
fixed point. If $\de>0$, $\de=\dpl$, so the two maps in \eqref{eq:share-fixed-point}
coincide at every gap.
\emph{(iii)} As $\de\downarrow0$, $S(\de)\to0$ by continuity, and $G$ has right
derivative $g(0)$ at zero, so $S=g(0)\bigl[(1-q)\,\de+\zeta S\bigr]+o(\de+S)$. Solving,
$S\,(1-R_U)=g(0)(1-q)\,\de+o(\de)$, which is \eqref{eq:share-local}.
\emph{(iv)} Substitute $S(\de)$ into Assumption~\ref{ass:contagion}. The forecaster's
own threshold $\tau_i$, and hence whether she herself re-bases, appears nowhere in the
expression.
\emph{(v)} Average over forecasters: the intercepts $a_i$ average to $a$, and
$V(\de)=a+\gamma_0\hat\sigma^2_{0,t}+\gamma\upsilon\,S(\de)$. By (ii) it is constant
below the target, and by (iii) it is \eqref{eq:law} above it, with
$b_+=\gamma\upsilon(1-q)\,g(0)/(1-R_U)$. At the target the one-sided derivatives are
$0$ and $b_+$: the law is continuous and kinked at the announced number.
\hfill$\square$

\paragraph{Proof of Corollary~\ref{cor:bill}.} $\mu_{t+1}=\lstar\hat u_{t+1}+\xi_{t+1}$,
so $\E[\mu_{t+1}\mid\Iset_t]=\lstar\hat u^e_t+\xi^e_t$ under the true structure. The
public's forecast $\de$ is that expectation, because it filters the sum rationally and
the law of the sum is the same under every labeling it entertains. The chain rule on
$V(\lstar\hat u^e_t+\xi^e_t)$ and Proposition~\ref{prop:law}(v) give the first
expression \eqref{eq:bill}. If $\xi\equiv0$, $\mu=\lstar\hat u$ and
$\de=\lstar\hat u^e_t$ exactly; since $\lstar>0$, $\dpl=\lstar(\hat u^e_t)_+$. The derivative is taken at given
expected components because the public filters $\hat u$ from a gap that carries the error: the precision of
that filter, and so $\hat u^e_t$ itself, depends on $\lstar$ through the signal-to-noise ratio
$(\lstar)^{2}\Var(\hat u_t)/\Var(\varepsilon_t)$. With no error the dependence vanishes and the bill is exactly linear in $\lstar$.
\hfill$\square$

\begin{remark}[What the law identifies]\label{rem:ident}
Within a regime the law identifies the product $(1-q)\,g(0)/(1-R_U)$, not $q$: halving distrust and
doubling the density of thresholds leaves every share unchanged. The testable content of the
confidence channel is that the slope is a constant of the regime; comparative statics in $q$
(Section~\ref{sec:announce}) compare regimes. $\upsilon$ and $\zeta$ are not separately identified either, and the
fitted slopes of Section~\ref{sec:evidence} are ceilings on $b_+$ \eqref{eq:struct-bounds}, not estimates of it.
\end{remark}

\begin{remark}[The drift's variance is Ball's lottery]\label{rem:ball}
Observational equivalence also leaves a second moment. In the pure readings,
$\Var_t(\xi_{t+1})=(1-q)\,\Sigma+q(1-q)\,(\de)^2$, with $\Sigma\equiv\Var_t(\mu_{t+1})$: the drift's own
risk, weighted by the probability of a drift, plus the variance across readings, which is the lottery
over policymakers of \citet{ball1992does} and the regime uncertainty of \citet{evans1993inflation}, with
the difference that it never resolves. It is not inflation risk (Lemma~\ref{prop:oe}(i)). Priced in the
stake at a rate $\omega$, it would add a constant share of doubters at the target and a curvature
$\omega q\,b_+$ below it, symmetric in the sign of the gap and second order: no slope and no kink. The law
is the case $\omega=0$, and the flat lower arm of Section~\ref{sec:evidence} is what bounds $\omega$.
\end{remark}

\begin{remark}[The scar]\label{rem:scar}
Agents learn from the raw surprises and do not purge the doubters' contagion from them. The learned
baseline therefore converges to the risk the economy has delivered, contagion included:
$\hat\sigma^2_{0,t}-\sigma_0^2\to\upsilon\,\E[S_t]$. A regime that overshoots often keeps wider densities
even if every episode closes, the long-run price of tolerating under an intention that cannot be proved.
A single episode moves the baseline only in proportion to its length over the sample's, which is why,
conditional on the current gap, past overshoots leave no visible trace (Section~\ref{sec:realized}).
The surprises agents learn from have conditional variance $\sigma^2_{0}+\upsilon\,S_{s-1}$, so under
decreasing-gain learning, a stationary gap and bounded conditional fourth moments of the surprises,
$\hat\sigma^2_{0,t}-\sigma^2_0-\upsilon\,\tfrac1t\sum_{s\le t}S_{s-1}\to0$ almost surely (the strong law of
large numbers for martingale differences).
\end{remark}

\paragraph{What the law identifies (Remark~\ref{rem:ident}).} Let
$G'(\tau)\equiv G\bigl(\tau(1-q)/(1-q')\bigr)$ and $\zeta'\equiv\zeta(1-q')/(1-q)$. Then
$G'\bigl((1-q')\dpl+\zeta'S\bigr)=G\bigl((1-q)\dpl+\zeta S\bigr)$ for every gap and
every share: the two parameterizations have the same doubting share at every gap, so
the law does not separate confidence from the scale of the thresholds.

\paragraph{The drift's variance (Remark~\ref{rem:ball}).} In the pure readings, with
$\Sigma\equiv\Var_t(\mu_{t+1})$, $\E_t[\xi_{t+1}^2]=(1-q)\bigl(\Sigma+(\de)^2\bigr)$, so
\[
\Var_t(\xi_{t+1})=(1-q)\bigl(\Sigma+(\de)^2\bigr)-(1-q)^2(\de)^2=(1-q)\,\Sigma+q(1-q)\,(\de)^2 .
\]
With the exposure $(1-q)\,\de+\omega\Var_t(\xi_{t+1})+\zeta S_t$, the share at the
target solves $\bar S=G\bigl(\omega(1-q)\Sigma+\zeta\bar S\bigr)$, the right slope is
$(1-q)\bar g/(1-\zeta\bar g)$ with $\bar g$ the density at that exposure, and the
curvature below the target is $\omega q$ times the slope. As $\omega\to0$ these
converge to the statements of Proposition~\ref{prop:law}. Neither moment of the drift
enters the law of inflation, which is the same under both readings
(Lemma~\ref{prop:oe}(i)).

\paragraph{Why the contagion is linear in the share.} Let price and wage setting
be granular, with an effective number $N$ of large setters
\citep{gabaix2011granular}, and let each doubter's re-based decision carry its
own error, of variance $\sigma^2_{e}$, around the effective target she reads.
With independent errors the doubters add to inflation a variance
$S_t\,\sigma^2_{e}/N$, linear in the share: Assumption~\ref{ass:contagion} with
$\upsilon=\sigma^2_{e}/N$. A \emph{common} error---every doubter misreading the
effective target the same way---would add $S_t^{2}\sigma^2_{e}$ instead,
quadratic in the share; the law would then rise with the square of the overshoot
near the target and have no kink there. The kink in the data selects the
granular reading---for the shape, not the size: at any plausible $N$, with
$\sigma^2_e$ of the order of the dispersion of forecasters' means, $\sigma^2_e/N$
falls far short of the slope the law requires, so $\upsilon$ is read as a
reduced-form loading whose linear form the kink selects
(Section~\ref{sec:model-law}).

\paragraph{Complementarity and uniqueness.} The gain from re-basing,
$\mathcal R_i(S,\de)=c_iz_t-k_i$, is nondecreasing in the doubting share,
$\partial\mathcal R_i/\partial S=c_i\zeta\ge0$, and in the expected overshoot above the
target, with both dominance regions: re-basing dominates whenever
$(1-q)\dpl\ge\tau_i$, whatever the others do, and the anchored plan whenever
$(1-q)\dpl+\zeta<\tau_i$. In words, $\zeta\sup g<1$ asks that the thresholds
be dispersed enough relative to the coordination motive: the extra exposure $\zeta$
that a full population of doubters would impose, times the highest density of
thresholds, stays below one, so that the re-basing a marginal rise in the exposure
triggers never raises it by as much again. It is a condition on primitives, not an
implication of the law; since $b_+$ rises with $g(0)$, a sharp kink at the target
works against it.

\paragraph{Heterogeneity, not global games.} The condition is that of
\citet{herrendorf2000ruling}: heterogeneity in private characteristics---here the
stakes and attention costs behind $\tau_i$---rules out multiplicity when it is
large relative to the complementarity. It is not the equilibrium selection of
global games \citep{carlsson1993global,morris1998unique,morris2003global}, in which
uniqueness comes from noisy private signals about a common fundamental and is
undone when public information is precise enough
\citep{angeletos2006crises,hellwig2006self}, or when the public signal can be
manipulated \citep{edmond2013information}. Here the fundamental of the game, the
expected overshoot $\de$, is public and learned, and the contraction holds under
any information structure, common knowledge included. For plans revised only at random times, such as the longer-term plans
of Appendix~\ref{app:benefit}, revision frictions and shocks to the fundamental
deliver uniqueness under the conditions of \citet{frankel2000resolving} and
\citet{burdzy2001fast}, even without the dispersion condition.

Under the condition there is no run, and no signal can select another
equilibrium, since there is none (Section~\ref{sec:model-law}). Large shocks do not
change this. In \citet{morris2019crises} a large shock to a fat-tailed common
component of payoffs, which players cannot separate from the idiosyncratic part,
leaves each unsure whether her payoff is above or below the others' (uniform rank
beliefs), and play shifts to the risk-dominant action where several equilibria
coexist. Here the common component of the payoff, the overshoot, is public, and an
agent's rank is her place in the distribution of thresholds, drawn independently
of the aggregate path (Assumption~\ref{ass:rebase}), which no shock moves. What the
announcement buys is that difference: a withdrawal of the benefit of the doubt
proportional to the overshoot, which reverses as the overshoot closes, rather than
a crisis.

\subsection{The benefit of the doubt on the professionals: a reading}
\label{app:benefit}

Professional forecasters' longer-term points---four to five years ahead---give a
direct reading of the benefit of the doubt they extend to the central bank.
Figure~\ref{fig:benefit} plots the share of respondents whose longer-term point
sits at or above $2.2$ per cent---two tenths above the target, so that the points
that moved to exactly $2.0$ with the July 2021 strategy review do not
count---against inflation's distance above the target. The share is low in calm
years, rises with the 2021--23 overshoot, peaks after inflation itself, and
re-anchors as the overshoot closes. It is descriptive: the longer-term point is
not the model's object, which is the one-year plan, and the paper uses the share,
written $S^{LT}_t$, only to ask whose withdrawal widens a forecaster's density
(Appendix~\ref{app:us}).

Its lag behind the overshoot has a simple reading. Longer-term plans are revisited
rarely: if in each period only a share $\lambda_{LT}$ of forecasters reconsiders
its longer-term point against the threshold rule of
Assumption~\ref{ass:rebase}, and the others keep the point they had,
\[
S^{LT}_t\;=\;(1-\lambda_{LT})\,S^{LT}_{t-1}\;+\;\lambda_{LT}\,G_{LT}\bigl(\dpl\bigr),
\]
the sticky information of \citet{mankiw2002sticky} and the epidemiology of
expectations of \citet{carroll2003macroeconomic}. The share then builds up while
the overshoot lasts, peaks after it, and re-anchors gradually as it closes. The
one-year point is restated in every round, $\lambda=1$, so the law
\eqref{eq:law} carries no such lag: the stickiness belongs to the longer-term
reading, not to the one-year density.

\paragraph{When the label is visible: four hikes.} The euro area has had four
inflation hikes since 1999. Read side by side, they show what the law responds to
and what it does not. Each was led by energy, and each pushed headline inflation
well above the target, but they differ in how plainly the overshoot carried its
label---whether it could be seen, print after print, to sit in energy and food
rather than in the core of the basket.

In \emph{2000--01} oil prices tripled from their 1999 trough and the euro
depreciated, and in 2001 animal-disease scares raised food prices; headline
inflation reached $2.9$ per cent while core inflation peaked at $2.2$. In
\emph{2007--08} oil climbed towards $\$150$ a barrel and food commodities surged;
headline inflation reached $3.9$ per cent, core never exceeded $1.9$. In
\emph{2011--12} oil rose again after the Libyan supply disruption, and several
consolidating governments raised indirect taxes; headline inflation reached $2.9$
per cent, core $1.6$. In all three nearly all of the overshoot sat outside the core, and
the professionals read it as such: in 2000--01 and 2011--12 their one-year
forecast stayed at or below the target in every round, and in 2007--08 it
exceeded the target by at most half a point. The law stayed nearly dormant. The
average individual variance, net of its learned baseline, moved by $-0.06$, $+0.08$
and $+0.05$ on average over the three episodes, at most $0.3$, and the share of
longer-term points at or above $2.2$ per cent never exceeded $0.16$.
\emph{2021--23} was different in kind. Energy led again, with gas prices after the
invasion of Ukraine, but the overshoot broadened: core inflation reached $5.5$
per cent and about half of the overshoot carried no energy-and-food label.
Forecasters expected inflation more than a point above target over the year on
average across the episode, three points at the peak; the average individual
variance rose by $0.77$ on average and by $1.9$ at the peak, and the longer-term
share climbed from a calm level of $0.04$ to $0.35$.

The intuition is the model's distinction between two margins. A visible supply
label is information on the identified margin: it tells the public that the
overshoot is of a kind that passes, so the expected gap stays small and few
thresholds are crossed. Headline inflation two points above target in 2008 widened
forecasters' densities by at most $0.3$ because almost no one
expected it to last. The label is no information on the unidentified margin: once
an overshoot is expected, the same share $1-q$ of it is still drift
(Lemma~\ref{prop:oe}(iii)), whatever its composition, and each point of it
should withdraw the same benefit of the doubt. That holds too. Per point of
expected overshoot, net of the learned baseline, the variance rose by $0.50$
($0.09$) in 2007--08 and by $0.65$ ($0.03$) in 2021--23, equality not rejected
($p=0.13$), and the longer-term share by $0.15$ and $0.13$. What the public can
learn moves the gap; the price per point of the gap does not move with it. The
reading is thin---four episodes, a 2007--08 response read on four rounds with an
expected overshoot, in a hike that ran into the financial crisis---and the label
is the one the public saw, not an identified supply share: sign-restricted
decompositions rank the same episodes differently, which is Wall~III again
(Appendix~\ref{app:walls}). Headline and core are quarterly averages of annual
HICP rates; the expected overshoot is the ECB Survey of Professional Forecasters'
one-year consensus less two, positive part; the slopes come from one regression
of the average individual variance on the learned variance of
Assumption~\ref{ass:contagion} and episode-specific terms in $\dpl$, HAC(4) standard
errors.\footnote{Script \texttt{results\_episodes.py} of the replication
package.}

\begin{figure}[!htbp]\centering
\includegraphics[width=0.95\linewidth]{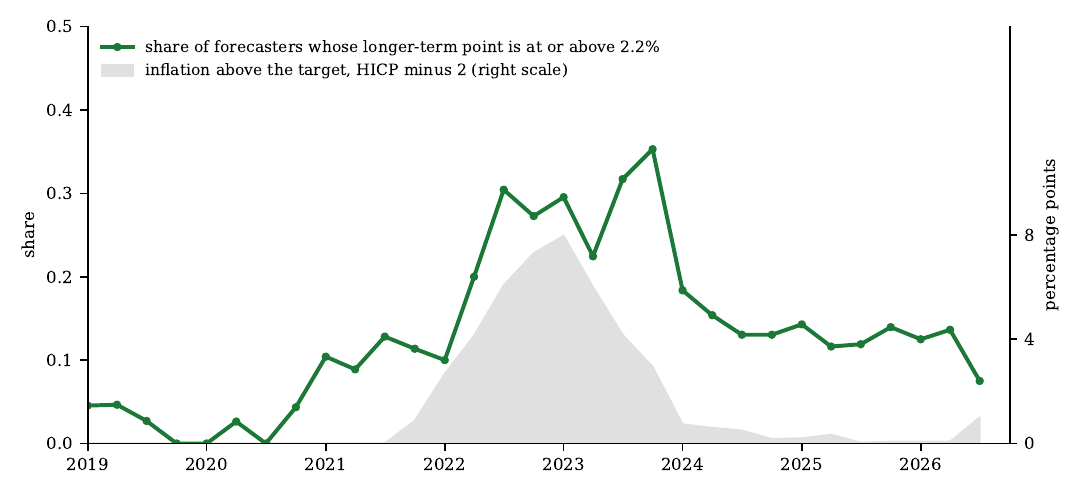}
\caption{\textbf{The benefit of the doubt, read on the professionals.} Share of
ECB Survey of Professional Forecasters respondents whose longer-term inflation
point forecast---four to five years ahead---is at or above $2.2$ per cent
(solid; $39$ to $52$ respondents per round over 2021--26), against inflation's
distance above the target (grey, right scale: headline HICP inflation minus two,
the mean of the last three prints known at each round's base month). The
threshold sits two tenths above the announced number so that the many points
that moved from $1.6$--$1.9$ to exactly $2.0$ with the July 2021 strategy review
do not count. Longer-term points are read from the survey's round files. Script:
\texttt{fig03\_benefit\_of\_doubt.py}.}
\label{fig:benefit}\end{figure}

\subsection{Where the observational-equivalence effect could be seen: the United States}
\label{app:us}

The United States is the one sample in which a numerical objective is announced inside the data, in January 2012. In the Survey of Professional Forecasters' core-CPI densities, asked on an unchanged grid since 2007, the variance of the consensus density does not rise with the expected overshoot before 2012 (upper arm $-0.63$, on seven rounds above $2\%$, all in 2007--08), nor between 2012 and 2020 ($-0.04$, $t=-0.5$), when the expected gap never exceeded $0.4$ points; it does in 2021--2026 ($0.20$, $t=3.2$), and the kink profiled on the post-2012 sample sits at $1.995$. Five-year breakevens give the same picture: no sixty-three-day window between 2012 and 2020 starts above the announced number's CPI image, and the upper arm is that of 2021--2026. This is what the OEE predicts once a number has been announced and is exceeded, and it is also what a level effect predicts in the first large overshoot. The sample cannot date the law to the announcement: a break placed at 2014, 2016, 2018 or 2021 fits as well as one at 2012. We read it as showing that the effect could be there, and no more.

Decomposed as in Section~\ref{sec:law}, the US law lives in a different component from the euro-area one. After the announcement US \emph{disagreement} rises with distance roughly symmetrically---on core CPI $b_-=0.480$ against $b_+=0.220$, steeper below the target---while US individual uncertainty, each forecaster's density on her own gap (Table~\ref{tab:armsbysource}), is flat below and rising above, as in the euro area. Aggregate US variance is dominated by the disagreement component and therefore looks symmetric. The euro area is the reverse: the individual component dominates and the aggregate inherits its shape.

The euro-area survey adds one pattern of the same standing. Within a round, forecasters whose own longer-term point forecast has left the target do not report wider one-year densities than the others (difference $-0.05$, $95$ per cent interval $[-0.24,\,0.14]$; Appendix~\ref{app:benefit}), while across rounds above the target the average density widens with the share who have. That is consistent with each forecaster pricing the re-basing she expects from others, as Proposition~\ref{prop:law}(iv) has it, and equally with any common driver of that share.

\subsection{No mean-side mechanism: Lemmas \ref{lem:ce} and \ref{lem:invariance}}
\label{app:meanside}

The law \eqref{eq:law} makes a conditional \emph{second} moment depend on the \emph{first} moment's state. The two lemmas used in Section~\ref{sec:model-law} record that nothing on the mean side of the block could have produced that dependence; they are stated here with their proofs.

\begin{lemma}[Certainty equivalence]\label{lem:ce}
In the linear-quadratic New-Keynesian block
\eqref{eq:nkpc}--\eqref{eq:taylor}, the conditional variance of future
inflation is independent of the current inflation gap, at every parameter
value and under every rule linear in the state.
\end{lemma}

\begin{lemma}[Mean-side invariance]\label{lem:invariance}
Write $m_t(h)=\E_t\pi_{t+h}$, measurable with respect to $\Iset_t$, and
compare policies that move only $m_t(h)$, holding fixed the conditional law
of the forecast error $\pi_{t+h}-m_t(h)$ given $\Iset_t$---policies that neither re-load
$\pi_{t+h}$ on future structural innovations nor add a random,
state-triggered component to it. Then $\Var_t(\pi_{t+h})$ is invariant to
the policy, however nonlinear.
\end{lemma}

\paragraph{Proof of Lemma \ref{lem:ce}.} Under any rule linear in the state, the
equilibrium of \eqref{eq:nkpc}--\eqref{eq:taylor} makes the inflation gap a
linear function of the exogenous shocks with coefficients that are constants
of the parameters and the rule. Take the cost-push AR(1),
$u_{t+1}=\rho_u u_t+e_{t+1}$ with $e$ i.i.d., mean zero, variance
$\sigma_e^2$. Under discretion the first-order condition $y_t=-\ep d_t$ and
the guess $d_t=\xi u_t$ give, in the Phillips curve,
$\xi(1+\ep\kappa-\beta\rho_u)=1$; any other linear rule changes only the
constant $\xi$. Iterating the shock $h$ steps,
\[
u_{t+h}=\rho_u^{h}u_t+\sum_{j=1}^{h}\rho_u^{\,h-j}e_{t+j},
\]
with the first term known at $t$, the $h$-step forecast error of the gap is
$\xi\sum_{j=1}^{h}\rho_u^{\,h-j}e_{t+j}$ and, the innovations being
independent,
\[
\Var_t\bigl(d_{t+h}\bigr)
=\xi^{2}\sigma_e^{2}\sum_{j=0}^{h-1}\rho_u^{2j}
=\xi^{2}\sigma_e^{2}\,\frac{1-\rho_u^{2h}}{1-\rho_u^{2}}\,.
\]
No term contains $d_t$: the current gap moves the conditional mean
$\xi\rho_u^{h}u_t$ and nothing else. With several shocks the error is the
corresponding fixed-weight sum across their innovations and the conclusion
is unchanged.

\paragraph{Proof of Lemma \ref{lem:invariance}.} With $m_t(h)$ measurable with
respect to $\Iset_t$,
\[
\Var_t\bigl(\pi_{t+h}\bigr)
=\Var_t\bigl(m_t(h)+[\pi_{t+h}-m_t(h)]\bigr)
=\Var_t\bigl(\pi_{t+h}-m_t(h)\bigr),
\]
because an $\Iset_t$-measurable quantity---however nonlinear in the
state---drops out of a conditional variance. A policy that moves only
$m_t(h)$ therefore moves nothing on the left, and by hypothesis it holds the
conditional law of the forecast error fixed. Contrapositively, a
conditional variance that depends on the gap requires the mechanism to reach
the innovation side---change the equilibrium loading of $\pi_{t+h}$ on
innovations dated after $t$, or add a component drawn after $t$---with the
reach itself state-dependent.

\subsection{Inflation-specific or general uncertainty: the credit discrimination}
\label{app:ngu}

The credit regression of Section~\ref{sec:nu} prices uncertainty about \emph{inflation}. A natural objection is that any uncertainty would do---that the survey's inflation uncertainty is priced only because it proxies for general, economy-wide uncertainty. This appendix runs the discrimination. The instrument is the growth analog of the correction: Normalized Growth Uncertainty ($\mathrm{NGU}$), the companion paper's measure \citep{vansteenberghe2026uncertain}---each forecaster's one-year-ahead GDP-growth density standard deviation $\sigma_i$, divided by $\sqrt{1+|\mu_i-g^{\mathrm{pot}}_t|}$ with $\mu_i$ the forecaster's density mean and $g^{\mathrm{pot}}_t$ the Hodrick--Prescott trend growth of euro-area GDP, averaged within each survey round. The construction mirrors $\mathrm{NU}$ with one deliberate difference: the denominator is symmetric, because growth has a benchmark but no announced target, so there is no announced point for a kink to sit at (the restriction stated in Section~\ref{sec:nu}). The series runs 1999Q1--2026Q3; over the full sample the two measures correlate at $0.66$.

Each measure is substituted for $U_q$ in \eqref{eq:creditspec} one at a time, never
jointly, exactly as in Table~\ref{tab:credit}; the orthogonalized components are the
residuals of one measure on the other taken forecaster by forecaster---one regression per
forecaster on the matched inflation--growth panel ($4{,}339$ forecaster-rounds; the $86$
forecasters with at least ten matched rounds carry their own slope, the remaining $21$ the
within-forecaster slope with their own intercept)---averaged within each round and
z-scored like every regressor. Table~\ref{tab:ngu} reports the full specification,
column (4) of the nested design.

\begin{table}[htbp]\centering
\caption{Inflation-specific or general uncertainty in the credit specification}
\label{tab:ngu}
\begin{tabular}{lc}
\toprule
substituted measure & (4) full specification \\
\midrule
$\mathrm{NU}$ (inflation) & $0.0089^{***}$ \\
 & $(4.04)$ \\
$\mathrm{NGU}$ (growth) & $0.0016$ \\
 & $(1.06)$ \\
$\mathrm{NU}\perp\mathrm{NGU}$ & $0.0058^{***}$ \\
 & $(4.05)$ \\
$\mathrm{NGU}\perp\mathrm{NU}$ & $-0.0025^{**}$ \\
 & $(-2.51)$ \\
\bottomrule
\end{tabular}
\vspace{0.3em}
\parbox{\linewidth}{\footnotesize \textit{Notes:} sample, dependent variable, controls, fixed effects and clustering
(bank) exactly as Table~\ref{tab:credit}; the column is the full specification, column
(4) of the nested design. Substituted one at a time, never jointly. By quarter the same cells read $4.98$, $0.76$, $3.53$ and $-1.85$; a restricted wild-cluster (Rademacher) bootstrap by quarter gives $p=0.004$ for $\mathrm{NU}\perp\mathrm{NGU}$ and $p=0.087$ for $\mathrm{NGU}\perp\mathrm{NU}$. The orthogonalized series are built by the replication repository's library (its \texttt{orthogonalize} routine); the regression rests on confidential loan-level data and is documented there, not reproduced.
$^{***}p<0.01$, $^{**}p<0.05$, $^{*}p<0.1$.}
\end{table}

\subsection{The level--uncertainty null of \texorpdfstring{\citet{abel2016measurement}}{Abel et al. (2016)}, replicated and extended}
\label{app:abel_replication}

The closest survey-level precedent for the survey law of Section~\ref{sec:fitlaw} reached the opposite
conclusion. \citet{abel2016measurement}, in their Section~3.2, ask whether aggregate inflation
uncertainty co-moves with the aggregate point prediction, and estimate
\begin{align}
  \bar{\sigma}_{q,t} &= \alpha + \beta\,\bar{f}_{q,t} + \varepsilon_{q,t},
  \label{eq:abel8}\\
  \tilde{\phi}^{\textsc{med}}_{q,t} &= \alpha + \beta\,\tilde{f}^{\textsc{med}}_{q,t}
    + \varepsilon_{q,t},
  \label{eq:abel9}
\end{align}
where $\bar\sigma_{q,t}$ is the square root of the cross-sectional \emph{mean of individual
predictive variances}, $\tilde\phi^{\textsc{med}}_{q,t}$ the cross-sectional \emph{median of
individual interquartile ranges}, and the regressor is matched in aggregator to the dependent
variable---mean with mean, median with median. Their Table~IV, for one-year-ahead HICP inflation
over 1999:Q1--2013:Q4, reports $\hat\beta=-0.097$ $(0.076)$ with $R^{2}=0.048$ on the
variance-based measure and $\hat\beta=-0.071$ $(0.081)$ with $R^{2}=0.008$ on the IQR-based one:
no relationship, on either measure. We also run the joint specification that adds forecast
dispersion,
\begin{align}
  \bar{\sigma}_{q,t} &= \alpha + \beta_{1}\,s_{q,t} + \beta_{2}\,\bar{f}_{q,t}
    + \varepsilon_{q,t},
  \label{eq:abel8joint}\\
  \tilde{\phi}^{\textsc{med}}_{q,t} &= \alpha + \beta_{1}\,\tilde{f}^{\textsc{iqr}}_{q,t}
    + \beta_{2}\,\tilde{f}^{\textsc{med}}_{q,t} + \varepsilon_{q,t},
  \label{eq:abel9joint}
\end{align}
with $s_{q,t}$ the cross-sectional standard deviation of point forecasts and
$\tilde f^{\textsc{iqr}}_{q,t}$ their interquartile range.

\paragraph{Implementation.} Individual predictive variances follow their equation~(1), which
assigns each histogram bin a uniform density:
\[
  {}_{i}\sigma^{2}_{q,t}
    = \sum_{n} {}_{i}p_{n}\,\frac{u_{n}^{3}-l_{n}^{3}}{3(u_{n}-l_{n})}
    - \left[\sum_{n} {}_{i}p_{n}\,\frac{u_{n}^{2}-l_{n}^{2}}{2(u_{n}-l_{n})}\right]^{2},
\]
for respondent $i$'s probability ${}_{i}p_{n}$ on the bin $[l_{n},u_{n}]$; individual
interquartile ranges come from linear interpolation inside the bin that first crosses each
quartile. Every regression uses Newey--West standard errors with four lags, their bandwidth.
Two implementation choices are not theirs and are reported both ways below. The first is the
closure of the open-ended tail bins: they widen each by twice the adjacent closed interval,
whereas the pipeline behind the rest of this paper places each bin's mass at the bin's
midpoint and closes the open bins on the realised range of euro-area inflation: the top bin's
mass at $7.81$, the midpoint of $[5,10.62]$ with $10.62$ the range's maximum, and the bottom
bin's at $-1$, inflation never having fallen below $-0.62\%$. The
second is respondent selection: they drop any respondent not returning both a point and a
density forecast and any density whose probabilities do not sum to unity, where our pipeline
renormalises.

\paragraph{Dating.} Rounds are labelled by the survey that produced them. The panel behind this
paper dates each observation by its \emph{target}, so the one-year-ahead round of $YYYY$Q1
carries the target date $YYYY$-12 and those of Q2, Q3 and Q4 carry $(YYYY{+}1)$-03, -06 and -09;
shifting back one year recovers the round. Abel's sample therefore ends at the shifted date
2013-09, and a window ending at 2013-12 would carry the 2014Q1 round as well. Ours ends where
theirs does, at $T=60$ rounds---their own count.

\begin{table}[tbp]
\centering
\small
\begin{tabular}{lccccc}
\toprule
 & & \multicolumn{2}{c}{Variance-based \eqref{eq:abel8}}
   & \multicolumn{2}{c}{IQR-based \eqref{eq:abel9}}\\
\cmidrule(lr){3-4}\cmidrule(lr){5-6}
Sample and specification & $T$ & $\hat\beta$ & $R^{2}$ & $\hat\beta$ & $R^{2}$\\
\midrule
\addlinespace[2pt]
\multicolumn{6}{l}{\textit{Panel A --- as published, and as replicated}}\\[2pt]
\quad\citet{abel2016measurement}, Table~IV & $60$ & $-0.097$ & $0.048$ & $-0.071$ & $0.008$\\
                                           &      & $(0.076)$ &        & $(0.081)$ & \\
\quad This paper, their sample             & $60$ & $+0.038$ & $0.010$ & $-0.071$ & $0.025$\\
                                           &      & $(0.066)$ &        & $(0.079)$ & \\
\quad This paper, extended to 2026Q2       & $110$ & $+0.299^{***}$ & $0.511$ & $+0.226^{**}$ & $0.302$\\
                                           &      & $(0.063)$ &        & $(0.096)$ & \\
\addlinespace[4pt]
\multicolumn{6}{l}{\textit{Panel B --- the same slope under three departures}}\\[2pt]
\quad Their tail closure, their sample     & $60$ & $-0.067$ & $0.040$ & $-0.071$ & $0.025$\\
                                           &      & $(0.058)$ &        & $(0.079)$ & \\
\quad Their tail closure, extended         & $110$ & $+0.102^{***}$ & $0.161$ & $+0.084^{*}$ & $0.088$\\
                                           &      & $(0.037)$ &        & $(0.044)$ & \\
\quad Pre-2024Q4 histogram grid only       & $103$ & $+0.290^{***}$ & $0.504$ & $+0.221^{**}$ & $0.289$\\
                                           &      & $(0.072)$ &        & $(0.103)$ & \\
\quad Their respondent filter, extended    & $110$ & $+0.297^{***}$ & $0.504$ & $+0.216^{**}$ & $0.280$\\
                                           &      & $(0.064)$ &        & $(0.096)$ & \\
\bottomrule
\end{tabular}
\caption{Uncertainty and the aggregate point prediction: \citet{abel2016measurement} Table~IV
(one-year-ahead HICP inflation), replicated on their sample and extended. Dependent variables:
$\bar\sigma_{q,t}$, the square root of the cross-sectional mean of individual predictive
variances; $\tilde\phi^{\textsc{med}}_{q,t}$, the cross-sectional median of individual
interquartile ranges. The regressor is the cross-sectional mean (variance-based) or median
(IQR-based) of point forecasts. Newey--West standard errors with four lags in parentheses.
Their sample is 1999Q1--2013Q4; the extension runs to the round conducted in 2026Q2, this
paper's cutoff. Panel~B's first two rows close the open-ended bins as they do rather than on the
realised range of inflation; the third drops every round on the histogram grid the ECB-SPF
adopted in 2024Q4; the fourth imposes their respondent selection in place of renormalisation.
$^{***}p<0.01$, $^{**}p<0.05$, $^{*}p<0.10$.}
\label{tab:abel_replication}
\end{table}

\begin{figure}[tbp]
\centering
\includegraphics[width=\textwidth]{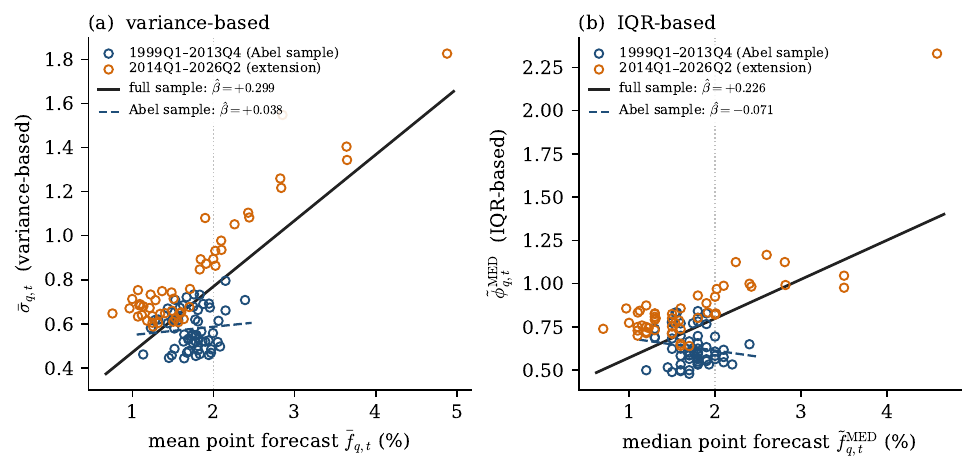}
\caption{Inflation uncertainty against the aggregate point forecast: \citet{abel2016measurement}
replicated, and the sample extended. \textbf{(a)} the variance-based measure
$\bar\sigma_{q,t}$ against the mean point forecast, equation~\eqref{eq:abel8}; \textbf{(b)} the
IQR-based measure $\tilde\phi^{\textsc{med}}_{q,t}$ against the median point forecast,
equation~\eqref{eq:abel9}. Blue: the $60$ rounds of their sample, 1999Q1--2013Q4. Orange: the
$50$ rounds of the extension, 2014Q1--2026Q2. Solid line: ordinary least squares on the full
sample; dashed line: on their sample alone. The dotted vertical marks the $2\%$ target. Their
cloud is flat because it sits almost entirely on the flat arm of \eqref{eq:lawsurvey}: the
consensus forecast left the target upward in $6$ of their $60$ rounds and never by more than
$0.39$ points. Script: \texttt{figA\_abel2016\_replication.py}, which also writes
Table~\ref{tab:abel_replication}.}
\label{fig:abel_scatter}
\end{figure}

\paragraph{What replicates.} On their own sample the null replicates. The variance-based slope is
$+0.038$ $(0.066)$, $p=0.57$, and the IQR-based slope is $-0.071$ $(0.079)$, $p=0.37$---the
latter their published estimate to three decimals. The joint
specification~\eqref{eq:abel8joint} reproduces their reading of which moment does the work:
dispersion enters at $+0.492$ $(0.205)$ while the point forecast stays flat at $+0.105$
$(0.067)$, so over their sample it is disagreement, not the level of expected inflation, that
co-moves with measured uncertainty.

The one gap is the sign of the variance-based point estimate, $+0.038$ against their $-0.097$,
both insignificant. The retired version of this appendix attributed that gap to the treatment of
the open-ended bins and did not test the attribution. Tested, it holds, and the mechanism is
visible: closing the tails as they do turns the slope negative, $-0.067$ $(0.058)$, while the
IQR-based slope is \emph{identical} under the two conventions to four decimals. The interquartile
range is an interior object---for these densities the quartiles never reach an outer bin---so the
variance is the only moment the tail closure can move, which is exactly why the IQR leg
reproduced their number and the variance leg did not.

\paragraph{What does not.} Extending the sample to the round conducted in 2026Q2 reverses the
conclusion. The variance-based slope rises to $+0.299$ $(0.063)$, $t=4.74$, with $R^{2}$ moving
from $0.01$ to $0.51$; the IQR-based slope to $+0.226$ $(0.096)$, $R^{2}=0.30$. In the joint
specification both regressors matter for the variance-based measure---dispersion at $+0.564$
$(0.111)$ and the point forecast at $+0.144$ $(0.045)$, $R^{2}=0.70$---whereas on the IQR side
dispersion alone survives.

Two checks stand between that reversal and a measurement artefact. The ECB-SPF changed its
histogram grid in the 2024Q4 round, a change of the kind
\citet{becker2026measuring} show can move measured uncertainty on its own; dropping every round
on the new grid leaves $T=103$ and $\hat\beta=+0.290$ $(0.072)$, $t=4.03$. And under their own
tail closure the reversal is smaller but survives, $+0.102$ $(0.037)$, $t=2.79$: the wide top bin
of our default closure amplifies measured variance precisely in the surge quarters, so the
\emph{magnitude} of the reversal is convention-sensitive where its \emph{existence} is not.
Imposing their respondent selection changes nothing, $+0.297$ against $+0.299$.

\paragraph{Why the two findings are the same finding.} The reversal is not a correction of
\citet{abel2016measurement}. Read through \eqref{eq:lawsurvey}, their null \emph{is} the flat
arm. Over their sample the consensus forecast exceeded the target in $6$ of $60$ rounds and never
by more than $0.39$ points, so the above-target gap they could have priced sums to $0.9$
percentage points across their rounds against $10.9$ across ours---eight percent of the
variation, in the only arm where the law has a slope. Fitting the arms form of
\eqref{eq:lawsurvey} on their window returns $b_{+}=+0.61$, positive and of the right order,
estimated on a regressor whose standard deviation is $0.057$ against $0.374$ on the full sample.
A regression of uncertainty on the \emph{level} of the forecast, run where the level almost never
leaves the neighbourhood of the target, is a regression of uncertainty on nothing. Their
conclusion was right for their sample, and it is the same conclusion this paper's law makes
below target; what the surge added was the arm on which the two readings part.

\subsection{The band and its budget}
\label{app:budget}

The main text quotes one number from the band exercise---between three-fifths and three-quarters of the irreducible uncertainty in the episode optimum is the classification of the shock---and this subsection states and proves the proposition that produces it. The band itself is adopted from the published readings, not identified.

\paragraph{How much of the optimum is unknowable, and which input is responsible.}
Let the supply share be taken from the published readings as $s_t\in[\underline s,\overline s]$ with midpoint $\check{s}$, and the intensity as $[\underline\iota,\overline\iota]$ with midpoint $\check{\iota}$ (the check marks a midpoint and the bars the ends of an interval; $\iota$ is the look-through intensity, $\lstar$ for a transitory impulse and $\iota(\rho_u)$ of \eqref{eq:iota_rho} for a persistent one).

\begin{proposition}[The conditional band and the uncertainty budget of the episode optimum]\label{prop:budget}
The episode optimum \eqref{eq:optlambda_episode} is bracketed only to the band $\lstare\in[\underline s\,\underline\iota,\ \overline s\,\overline\iota]$. Sweeping each input alone across its interval, the supply share generates a fraction $\check{\iota}(\overline s-\underline s)\big/\!\big[\check{\iota}(\overline s-\underline s)+\check{s}(\overline\iota-\underline\iota)\big]$ of the resulting variation in the optimum. Under independent uniform priors on the two intervals,
\begin{equation}\label{eq:budget_var}
\Var\!\big(\lstare\big)=\check{\iota}^{2}\,\Var(s_t)+\check{s}^{2}\,\Var(\lstar)+\Var(s_t)\,\Var(\lstar),
\end{equation}
of which the supply share contributes $\check{\iota}^{2}\Var(s_t)\big/\Var(\lstare)$. All three statements are immediate from the product form $\lstare=s_t\lstar$.
\end{proposition}

At the surge readings $[\underline s,\overline s]=[0.50,0.75]$ and the calibrated band $[\underline\iota,\overline\iota]=[0.78,0.95]$ of \eqref{eq:lstar-model}, the band is $\lstare\in[0.39,0.71]$. Prior-free, the supply share moves the optimum by $0.22$ against the intensity's $0.11$---$67\%$ of the one-at-a-time variation; under independent uniform priors it carries $80\%$ of the variance against $20\%$ for the intensity and below $1\%$ for their interaction. Either convention delivers the reading quoted in the main text: between two-thirds and four-fifths of the irreducible uncertainty in optimal tolerance is the inability to classify the shock.

\paragraph{Proof.} The episode optimum is the product $\lstare=s_t\,\iota$ with $\iota\equiv\lstar$. Take the two factors independent, with $\E s=\check{s}$ and $\E\iota=\check{\iota}$ under the uniform priors. For independent factors,
\[
\Var(s\iota)=\E[s^2]\E[\iota^2]-(\E s)^2(\E \iota)^2
=\check{\iota}^{\,2}\Var(s)+\check{s}^{\,2}\Var(\iota)+\Var(s)\Var(\iota) ,
\]
which is \eqref{eq:budget_var}. The three terms are the classification channel, the calibration channel, and their interaction; the supply share's contribution to the total is $\check{\iota}^{2}\Var(s)/\Var(\lstare)$. Independence is the only substantive assumption. It holds exactly under Assumption~\ref{ass:comp}, whose primitive is the pair $(s_t,\Pi_t)$: the share is then a property of the impulse and does not move with $\kappa$. It is an approximation when the structural disturbances $(u_t,v_t)$ are taken as primitive instead: $\lstar=1/(1+\ep\kappa)$ and the share $s_t=u_t/\Pi_t$, with $\Pi_t=u_t+\tfrac{\kappa}{1+\sigma\varphi_y}v_t$, are then both decreasing in the same $\kappa$, which is precisely the parameter the paper treats as unknown, so the induced correlation is positive. It cuts the safe way. Positive dependence raises the variance---to first order by $2\check{s}\check{\iota}\Cov(s,\iota)>0$, and on the draw below by more---so \eqref{eq:budget_var} is a \emph{lower} bound on the budget and the band it supports is conservative; on a draw with $\kappa\sim U[0.05,0.30]$, $\ep=6$ and equal impulse weights ($u=v$), the three-term formula understates $\Var(s\iota)$ by roughly three-tenths at
the demand block's $\sigma=1$, $\varphi_y=0.5$ (by a third with the demand
weight $\kappa'=\kappa$). Conditional independence of the estimated model's components---target, correction, innovation---is a different statement and does not license this step, which is why we take independence here as an explicit and signed approximation rather than as an inherited assumption. One further caveat bears on the inputs rather than on their independence. The share $s_t$ of Assumption~\ref{ass:comp} is the supply share of the impulse on impact, while the published decompositions measure the supply share of \emph{realized} inflation. At a finite $\varphi_\pi$ the realized gap loads the supply component by $a_u=(1+A\lstar)/(1+A-\beta\rho_u)$ and the demand component by $a_v=\kappa'/(1+A)$ (Appendix~\ref{app:oe}), so for components of the same sign the realized share exceeds $s_t$ whenever the bank accommodates or the cost-push persists. Read off the decompositions, the share is therefore overstated, and the band overstates the optimum with it; an overstatement proportional across the readings leaves the classification channel's share of the budget unchanged.

\subsection{Tolerance with an unobserved slope and unobserved shocks: no divine coincidence}
\label{app:dsge}

This subsection asks whether a bank that observes neither the Phillips slope nor the composition of its shocks should tolerate less than one that does, as it would have to for its own ignorance to shrink the OEE as a by-product (Section~\ref{sec:acting}). The answer is no, for two textbook reasons given at the end. It uses the textbook model of \citet[chs.~3--5]{Gali2015} at his baseline calibration.

\paragraph{The model.} The block is \eqref{eq:nkpc}--\eqref{eq:target}:
\begin{align*}
d_t&=\beta\E_t d_{t+1}+\kappa y_t+u_t, & y_t&=\E_t y_{t+1}-\sigma(\hat\imath_t-\E_t d_{t+1})+v_t,\\
\hat\imath_t&=\varphi_\pi(d_t-m_t)+\varphi_y y_t, & m_t&=\lambda\,\Pi_t ,
\end{align*}
with $\Pi_t=u_t+\kappa' v_t$ the impulse and $\kappa'=\kappa/(1+\sigma\varphi_y)$. The disturbances $u_t$ and $v_t$ are independent AR(1) processes with common persistence $\rho$. The bank observes $\Pi_t$ but not its composition; the supply share is $s\equiv\Var(u_t)/\Var(\Pi_t)$, with $\Var(\Pi_t)$ normalized to one. The welfare loss is $\Var(d_t)+(\kappa/\ep)\Var(y_t)$ \citep[ch.~4]{Gali2015}. The calibration is Galí's: $\beta=0.99$, $\sigma=1$, $\ep=9$, $\varphi_\pi=1.5$, $\varphi_y=0.125$, and $\rho=0.5$, the persistence of his demand shock. The two unobserved inputs take their published ranges: $\kappa\in[0.009,0.029]$, the euro-area range at the quarterly decision period on the output gap, and $s\in[0.50,0.75]$. The bank chooses $\lambda$ to minimize the expected loss under uniform priors on the two ranges; the minimax-regret choice is the robust alternative. The model is solved by the method of \citet{blanchard1980solution}; the solution matches the closed-form minimum-state-variable solution, and a simulation of $400{,}000$ quarters reproduces the analytic loss.

\paragraph{Tolerance is pass-through.} At Galí's $\varphi_\pi=1.5$ and the published slope, the rule lets $2.4$ times the impulse into inflation at $\lambda=0$, and no tolerance coefficient is optimal: the bank would need to lean harder than any $\lambda$ allows. The coefficient is entangled with $\varphi_\pi$ (Wall~II): the expected-loss minimizing coefficient is $-2.4$ at $\varphi_\pi=10$ and $0.41$ at $\varphi_\pi=100$, while the share of the impulse that reaches inflation stays at $0.70$ and $0.72$. That share is the invariant object. In the high-gain limit of Assumption~\ref{ass:innov} it equals $\lambda$, and Table~\ref{tab:dsge} reports the results there.

\begin{table}[!htbp]\centering\small
\caption{Optimal tolerance with the slope and the shocks unobserved: no divine coincidence}\label{tab:dsge}
\resizebox{\linewidth}{!}{%
\begin{tabular}{lcccccc}
\toprule
 & \multicolumn{4}{c}{optimal tolerance} & \multicolumn{2}{c}{price of tolerating less (\%)} \\
\cmidrule(lr){2-5}\cmidrule(lr){6-7}
$\rho$ & composition seen & share and slope known & nothing observed & minimax regret & half & none \\
\midrule
$0$    & $0.54$ & $0.53$ & $0.54$ & $0.58$ & $29$ & $116$ \\
$0.25$ & $0.64$ & $0.64$ & $0.66$ & $0.71$ & $24$ & $96$ \\
$0.5$  & $0.75$ & $0.75$ & $0.78$ & $0.87$ & $16$ & $64$ \\
$0.8$  & $0.66$ & $0.65$ & $0.67$ & $0.83$ & $4$  & $16$ \\
\bottomrule
\end{tabular}}
\parbox{\linewidth}{\footnotesize\smallskip \textit{Notes:} high-gain limit of the rule; quarterly persistence $\rho$ of both disturbances; $\kappa\in[0.009,0.029]$ and $s\in[0.50,0.75]$ with uniform priors; Galí's calibration otherwise. Optimal tolerance as the share of the impulse that reaches inflation. ``Composition seen'': average tolerance of a bank that observes each impulse's composition, offsets its demand part and tolerates its supply part optimally. ``Share and slope known'': average of the optima over the two ranges. ``Nothing observed'': the tolerance that minimizes the expected loss over the ranges; ``minimax regret'': the robust alternative. Price of tolerating less: excess expected loss, relative to the nothing-observed optimum, of tolerating half of it or nothing. Script: \texttt{fig07\_tolerance\_dsge.py}.}
\end{table}

\paragraph{Reading.} Knowing less does not mean tolerating less. Table~\ref{tab:dsge} sets the three banks side by side. At a quarterly persistence of one half, a bank that sees the composition of each impulse tolerates, on average, $0.75$ of it; a bank that knows only the average share and the slope, $0.75$; a bank that knows neither, $0.78$. Two textbook facts account for the ranking. Not seeing the composition changes nothing on average: with a quadratic loss the bank acts on its estimate of the share as it would on the share itself, which is certainty equivalence \citep{svensson2003indicator}. Not knowing the slope raises tolerance, it does not lower it. The slope enters the loss through its inverse, so the loss averaged over the slope's range is minimized at the harmonic mean of the slope, which lies below its mean, and a flatter curve calls for letting more of the impulse through (Figure~\ref{fig:dsge}(a)): unsure of what the instrument does, the bank uses it less \citep{brainard1967uncertainty}. A bank that guards against the worst slope and share goes further, to $0.87$. The ranking is the same at every persistence. What dominates is the range of the answer: at a persistence of one half, the tolerance a fully informed bank would choose runs from $0.49$ to $1.13$ over the published ranges, so not knowing the inputs matters far more for the answer than for the choice. A bank that wanted to shrink the OEE would therefore have to tolerate less than its own ignorance prescribes, at the price in the last two columns of the table, which falls steeply with persistence. At zero persistence the model reproduces the static band of Proposition~\ref{prop:budget}. The exercise assumes commitment to a constant tolerance, rational expectations and a common persistence of the two disturbances; it leaves out learning and de-anchoring, and it contains no OEE, whose cost cannot be measured (Section~\ref{sec:reverse}).

\subsection{The perceived-rule band: NY Fed scenario matrices}
\label{app:spd_panel}

The instrument is the New York Fed's Survey of Primary Dealers and Survey of Market Participants, which since 2023 have asked respondents for the policy rate they expect under stated macroeconomic scenarios. A scenario matrix is a set of conditional rate paths, and the dispersion across respondents within a cell is a direct read on how far apart market participants are about the reaction function itself---not about the outlook.

\paragraph{The literature.} The perceived reaction function has been measured
in three ways, none of which conditions on the state. The first estimates it
from realized histories: \citet{hamilton2011estimating} recover the rule market
participants use from the joint response of rate and macroeconomic forecasts
to news, and \citet{bauer2024perceptions} estimate it from a panel of
professional forecasts, finding a perceived dependence of the funds rate on
economic conditions that varies substantially over the policy cycle and is
revised on monetary surprises---imperfect information about the rule itself.
The second reads it off forecast errors: \citet{cieslak2018short} and
\citet{schmeling2022monetary} document large, persistent errors in the
expected policy path, concentrated in downturns, when the public
underestimates how much the Fed eases---a rule learned slowly, from rare
states. The third asks respondents directly: \citet{carvalho2014people} infer
a Taylor-type rule from the co-movement of households' inflation, unemployment
and interest-rate answers, \citet{andre2022subjective} elicit the perceived
effects of a rate rise from hypothetical vignettes and find experts and
households disagreeing even on its sign, and the New York Fed's dealer survey
has itself been read for the perceived rule---\citet{femia2013effects} trace
in it how the 2011--12 guidance moved dealers' expectation of the unemployment
rate at which the Committee would first raise rates. Dispersion among
professionals about the policy path under common information is documented
for forward guidance by \citet{andrade2019forward} and priced, as uncertainty
about the path, in interest-rate options by \citet{bauer2022market}. What the
scenario matrix adds is the missing conditioning. Every respondent answers the
same counterfactual states at the same date, so the cross-respondent dispersion
of a cell is dispersion about the reaction function and not about the
outlook---the separating menu of \citet{healy2026which} that field histories
never supply---and the slope across cells is a finite difference each
respondent reports, not a coefficient estimated from the histories on which,
by \eqref{eq:phieff}, the rule is not identified.

\paragraph{The panel.} Fifteen survey-panel rows across six survey dates, with dealer, buy-side and combined readings. The nine-cell mean interquartile ranges run from $26.3$ to $95.8$ basis points. None of them approaches the $12.5$ basis-point grid on which respondents answer, which is the check that the dispersion is substantive rather than an artefact of the response scale.

\paragraph{Proximity, and what it is worth.} The buy side is at least as dispersed as the dealers in \textbf{four of six} instances, on the nine-cell mean and on the centre cell alike, and the two measures fail on overlapping but not identical surveys---July 2025 being the one instance that fails on both. On the nine-cell mean the exceptions are April 2024 ($33.3$ against $43.2$) and July 2025 ($29.9$ against $30.6$); on the centre cell they are July 2025 ($0$ against $25$) and April 2026 ($31$ against $44$). Two of the six comparisons are knife-edge---July 2025 turns on $0.7$bp against a $12.5$bp grid step, and April 2024's centre cell on a $25$-against-$25$ tie---so the count itself should not be leaned on. The state-pinned dispersion exceeds the matched-horizon unconditional dispersion in seven of the ten dealer and buy-side rows, and in eight of all thirteen rows with a nonzero comparator; the restriction that produces it is the ridge of \eqref{eq:phieff}: histories reach the rule only through $\varphi_\pi(1-\lambda)$, so common information cannot narrow the coefficient behind it.

\paragraph{What the panel does not support.} It does not support a trend \emph{in the nine-cell mean}: six survey dates over three years cannot date a narrowing of the band, and the apparent decline from 2023 to 2026 is within the dispersion of the cells themselves. The centre-cell series is a separate object and moves the other way---it rises from zero to $44$bp across the six dates---and that reading is weakly monotone, with three consecutive ties, on six points. It does not support a level comparison against the euro area, which has no equivalent instrument. And it does not identify the reaction function, which is the point---the band is evidence that the function is not agreed upon, not a measurement of it.

\paragraph{How the perceived responses are computed.} Each matrix reports,
for nine hypothetical year-end states---core PCE inflation and the
unemployment rate each at the current SEP median, fifty basis points below it
and fifty above it---the cross-respondent quartiles and count of the target
rate expected at that year-end. Write $m(\pi_r,u_c)$ for the median cell in
inflation row $r$ and unemployment column $c$. The perceived responses are
\[
\frac{\partial i}{\partial\pi}=\frac13\sum_{c}\big[m(\pi_{+50},u_c)-m(\pi_{-50},u_c)\big],
\qquad
\frac{\partial i}{\partial u}=\frac13\sum_{r}\big[m(\pi_r,u_{+50})-m(\pi_r,u_{-50})\big]:
\]
the high-inflation row minus the low-inflation row, averaged over the three
unemployment columns, and its transpose. The two scenarios in each difference
are one percentage point apart, so each number is the median respondent's
change in the expected target rate per point of the conditioning variable, in
the survey's own units---a finite difference across counterfactual states, not
a regression coefficient, and one that needs no model of the rule. Three
properties follow. The intercept common to the nine cells---the unconditional
level, the respondent's view of $r^{*}$, every state the question does not
pin---cancels in the difference, so the slope reading is free of the level
confound that the rule-residual route carries (Wall~II). The cells are
published medians, so the slope is the median respondent's rather than the
mean of individual slopes, and without the micro data it carries no standard
error; what is said of it is ordinal. And the horizon enters as partial
adjustment: a target rate three months away can absorb fewer meetings than
one fifteen months away, so a shorter horizon mechanically damps the slope,
which is why only matched-horizon comparisons are read---the two eight-month
instances, April--May 2024 and April 2026. Answers sit on a $12.5$ basis-point
grid and the published percentiles are interpolated, so a difference of a few
basis points at a single instance is noise. Script: \texttt{tabA\_nyfed\_scenario\_panel.py}.

\paragraph{The perceived responses.} On the two eight-month instances the perceived inflation response falls from $+0.96$ to $+0.58$ per point while the unemployment response stays within its range ($-0.54$ to $-0.67$); through $\varphi_\pi^{\mathrm{eff}}=\varphi_\pi(1-\lambda)$ that is the direction of a tolerance-regime switch, not its size---the baseline $\varphi_\pi$ is unidentified, and no reading of the tolerated share is imported here.

\begin{table}[htbp]\centering
\small
\resizebox{\linewidth}{!}{\begin{tabular}{lcccccccc}
\toprule
Survey & $n$ & Horizon & \multicolumn{3}{c}{Cross-dealer IQR (bp)} & Ratio & \multicolumn{2}{c}{Perceived response}\\
\cmidrule(lr){4-6}\cmidrule(lr){8-9}
 & & (months) & 9-cell mean & center cell & uncond.\ path & cond./uncond. & $\partial i/\partial\pi$ & $\partial i/\partial u$\\
\midrule
Mar.\ 2023     & 25 & 9  & 33.3 & 0  & 25.0 & 1.33 & $+0.83$ & $-0.58$\\
Sep.\ 2023     & 22 & 15 & 50.7 & 19 & 88.0 & 0.58 & $+1.04$ & $-0.79$\\
Apr./May 2024  & 24 & 8  & 43.2 & 25 & 32.0 & 1.35 & $+0.96$ & $-0.54$\\
Sep.\ 2024     & 23 & 3  & 34.7 & 25 & 0.0  & ---  & $+0.50$ & $-0.58$\\
Jul.\ 2025     & 24 & 5  & 30.6 & 25 & 25.0 & 1.22 & $+0.58$ & $-0.58$\\
Apr.\ 2026     & 26 & 8  & 26.3 & 44 & 50.0 & 0.53 & $+0.58$ & $-0.67$\\
\bottomrule
\end{tabular}
}
\caption{The perceived-rule band: dealer answers to the NY Fed scenario-matrix question.
Each survey asks for the year-end target rate under nine hypothetical states (2024
files: values recovered from the published PDFs, the machine-readable files being
zeroed at source; March 2023 predates the workbook format and was likewise read from
the published PDFs; later instances machine-read and cross-validated---the April 2026
matrix matches an independent PDF reading in all 27 cells). ``9-cell mean''
and ``center cell'' are cross-dealer interquartile ranges of the state-pinned answers;
``uncond.\ path'' is the IQR of the unconditional modal-path question at the same
horizon; the ratio uses these two (Sep.\ 2024 omitted: unconditional IQR zero, a full
hold consensus against 25--50bp conditional cells). Perceived responses are median cell
differences per 100bp of the scenario variable; horizons differ across instances, so
only matched-horizon slope comparisons are meaningful. Responses sit on a 12.5bp
grid; $n$ is dealers answering the matrix. Buy-side rows are included in the panel;
methods above. Script: \texttt{tabA\_nyfed\_scenario\_panel.py}.}
\label{tab:spd_panel}
\end{table}

    \section{The reverse link: specifications and robustness}
\label{app:lgm}

This appendix documents Section~\ref{sec:reverse}. Throughout, the GARCH constant and persistence are written $\omega_h$ and $\beta_h$, the level factor of the multiplicative form $\Lambda_t$, and $\nu$ is not used: errors are Student-$t$ with estimated degrees of freedom. Every likelihood-ratio test is computed after iterated cross-polishing---each model warm-started from all others until none improves---and every nesting check passes.

\begin{table}[!htbp]\centering\footnotesize\setlength{\tabcolsep}{3pt}
\caption{Main specifications and where the variants are}
\label{tab:lgm-map}
\begin{tabular}{p{2.6cm}p{5.8cm}p{7.0cm}}
\toprule
Model & Main specification & Variants\\
\midrule
Realized inflation & additive \eqref{eq:lgm-add}, AR(6), level trigger $\mu_t$, 1999--2026 & multiplicative form; forced ARCH; oil, core, ex energy, episode (\ref{app:lgm:m1})\\
Expected inflation & multiplicative \eqref{eq:lgm-mult}, centre $2\%$, roll days missing & additive and literal \citet{brenner1996another} forms; pre-2021; episode dummy (\ref{app:lgm:m23})\\
Expected rate & multiplicative \eqref{eq:lgm-mult}, centre $0$ & additive and literal forms; pre-2022; hiking dummy (\ref{app:lgm:m23})\\
Two-way & \eqref{eq:lgm-mult}, other series in changes, $\mathcal U=\sqrt h$ & levels and Toda--Yamamoto; realized-volatility $\mathcal U$; split at July 2022 (\ref{app:lgm:twoway})\\
Augmented & additive, no floor, $\theta_e d^e+\varpi_{\mathcal U}\mathcal U^e$, $\delta_{\mathcal U}(\mathcal U^e)^2$ & floor $\omega_h\ge1$; forced memory; IVX; oil, core, episode; realized $\mathcal U^e$; out of sample (\ref{app:lgm:aug})\\
\bottomrule
\end{tabular}
\end{table}

\subsection{Realized inflation}
\label{app:lgm:m1}

\paragraph{Choices.} The AR(6) is the smallest order whose standardized residuals are white (Ljung--Box at $12$ and $24$ lags, $p>0.10$). The level enters the variance through the predicted level $\mu_t$, known at $t-1$: the contemporaneous $d_t$ would make the variance depend on the shock it scales. The sample is 1999:02--2026:10, $n=333$. The additive form is the literature's for monthly inflation \citep{brunner1993higher,grier1998inflation,grier2000effects}, and the data select it: in the multiplicative form, $\sigma_t^2=h_t\Lambda_t$ with $\Lambda_t=(1+(\mu_t)_+)^{2\eta_+}(1+(-\mu_t)_+)^{2\eta_-}$, the fit is worse with the same number of parameters (log-likelihood $-722.7$ against $-717.7$, optimum confirmed from $48$ starting points) and the level effect weaker ($\hat\eta_+=0.70$, $p=0.086$).

\paragraph{Estimates.} $\hat\delta_+=7.97$ (likelihood ratio $13.0$, $p=0.0002$), $\hat\delta_-=0.11$ ($0.14$, $p=0.35$); jointly $13.0$ ($p=0.0005$); the symmetric restriction $\delta_+=\delta_-$ is rejected ($7.4$, $p=0.007$, without the in-mean term). The in-mean coefficient is $\hat\varpi=0.45$ ($p=0.069$); $\hat\omega_h=1.26$, $\hat\alpha=0.00$, $\hat\beta_h=0.63$, $5.4$ degrees of freedom. A predicted overshoot of one point raises the monthly conditional variance by $8$ pp$^2$ on impact (in annualized month-on-month units), six times $\hat\omega_h$, and by $21$ once variance persistence has built up; mapped into the variance of inflation over the next twelve months with the AR(6) weights, $1.7$ pp$^2$ on impact and $4.6$ in steady state, against a baseline of $0.74$. On 2010--2026 alone the effect keeps its sign and one-sidedness but weakens ($\hat\delta_+=6.4$; $p=0.031$ for $\delta_+=0$, $0.074$ jointly). The size is not comparable to the year-on-year ARMA--GARCH of Section~\ref{sec:acting}, which is estimated on a smoother gap.

\paragraph{Forced ARCH.} In the free fit $\hat\alpha=0$: the variance is a smoothed function of the predicted level. Profiling the likelihood in $\alpha$, with every other parameter free, gives a $95\%$ set $\alpha\in[0,0.08]$; $\alpha\ge0.10$ is rejected ($p=0.035$ at $0.10$). Within the admissible set the ladder of tests is re-run under $\alpha\ge\alpha_0$, and $\delta_+$ stays strongly significant (likelihood ratio $10.9$ at $\alpha_0=0.05$, $9.5$ at $0.08$) while $\delta_-$ stays at zero. Plain GARCH's $\hat\alpha=0.13$ is absorbed by $\delta_+$: what looks like ARCH is the overshoot effect. From $\alpha\approx0.13$ the optimum switches to symmetric level terms with a negative in-mean term, values the data reject. Residual ARCH remains in the free model (LM(12) $p=0.013$), and forcing $\alpha$ does not remove it; lagged Brent volatility does (below).

\begin{table}[!htbp]\centering\small
\caption{Realized inflation: $\delta_+$ under energy and episode controls}
\label{tab:lgm-m1energy}
\begin{tabular}{lcccc}
\toprule
Specification & $n$ & $\hat\delta_+$ & LR ($p$) & $\hat\delta_-$\\
\midrule
Headline, 1999:04--2026:09 (oil window) & 330 & 8.52 & 14.6 (0.0001) & 0.12\\
\quad + realized oil in the mean ($\Delta\log$ Brent in euros, $t$ and $t-1$) & 330 & 6.05 & 10.9 (0.0005) & 0.32\\
\quad + lagged Brent realized variance in the variance & 330 & 5.86 & 15.5 ($<$0.001) & 0\\
\quad + both & 330 & 4.61 & 13.0 (0.0002) & 0.15\\
Headline, 2021:07--2023:12 variance dummy & 331 & 7.89 & 6.8 (0.005) & 0.19\\
Headline, 1999:04--2020:12 & 261 & 5.21 & 1.7 (0.097) & 0.04\\
Headline, 2010:01--2026:10 & 202 & 6.35 & 3.5 (0.031) & 0.17\\
\midrule
Ex energy & 331 & 2.76 & 4.0 (0.022) & 0\\
Core (excluding energy, food, alcohol, tobacco) & 331 & 0.00 & 0.0 (0.50) & 0.02\\
Core, with the 2021--23 variance dummy & 331 & 0.00 & 0.0 (0.50) & 0.04\\
Core, 2010:01--2026:10 & 202 & 4.32 & 3.6 (0.029) & 1.67\\
Core, 2010+, with the 2021--23 variance dummy & 202 & 0.00 & 0.0 (1.0) & 0.03\\
\bottomrule
\end{tabular}
\vspace{0.3em}
\parbox{\linewidth}{\footnotesize \textit{Notes:} HICP seasonally adjusted (ECB). Brent in euros: monthly mean, $\Delta\log$ in \%; realized variance: sum of squared daily returns in month $t-1$. $\bar\chi^2$ tests. Core 2010--2020 is degenerate (the predicted core gap is almost never positive) and is not reported. The headline baseline on the energy window differs slightly from the main estimate because the energy series start later.}
\end{table}

\paragraph{Energy.} Overshoots are mostly energy-driven (2008, 2021--23), so the one-sided law could be composition: realized headline variance is high when the volatile energy component is large. Oil in the mean and the variance halves $\delta_+$ and lagged Brent volatility removes the residual ARCH (Ljung--Box on $z^2$, $p=0.15$ against $0.02$); the remaining effect is significant and one-sided. Core inflation has no level effect in the full sample, and its effect from 2010 is the 2021--23 episode itself. Headline $\delta_+$ survives the episode dummy: within the episode, months with a higher predicted overshoot are more volatile. The realized side of the law is therefore stated for headline inflation, with energy about half of it.

\begin{figure}[!htbp]\centering
\includegraphics[width=0.86\linewidth]{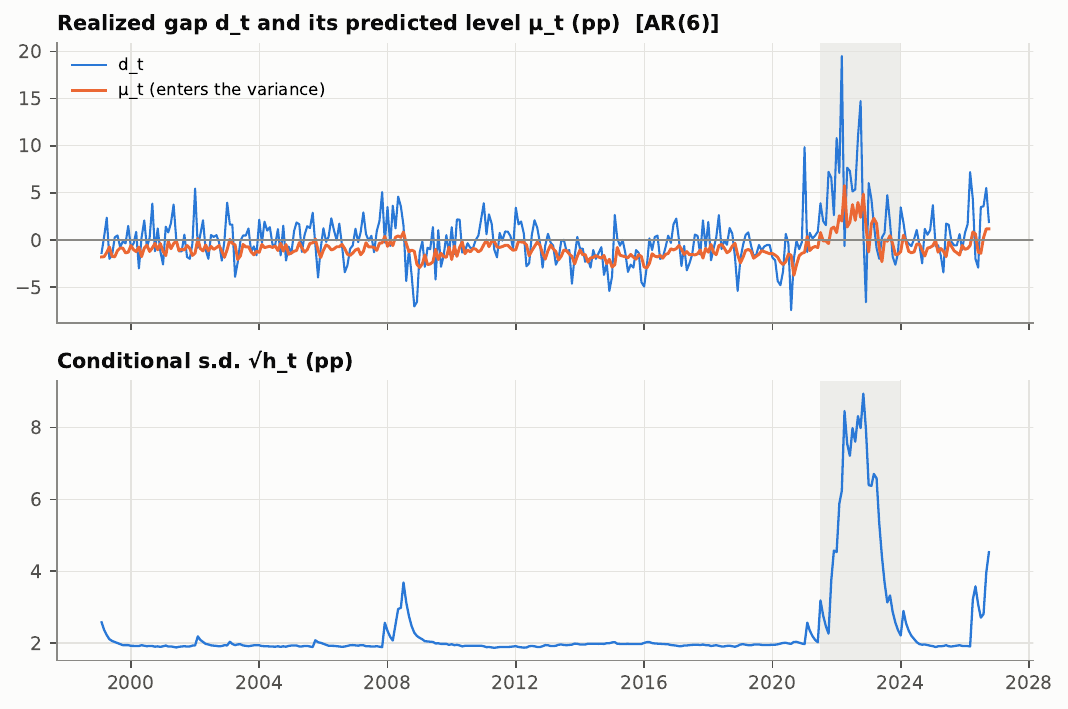}
\caption{\textbf{Realized inflation and its conditional standard deviation.} Model \eqref{eq:lgm-add}. Top: the realized gap $d_t$ (annualized month-on-month HICP minus $2$) and its predicted level $\mu_t$. Bottom: the conditional standard deviation $\sqrt{h_t}$, pp. Volatility rises when the predicted level is above target (2008, 2021--23, 2026).}
\label{fig:lgm-m1}
\end{figure}

\subsection{Expected inflation and the expected policy rate}
\label{app:lgm:m23}

\paragraph{The form.} The literal level-GARCH of \citet{brenner1996another} and \citet{koedijk1997dynamics} sets $\sigma_t^2=h_t\,r_{t-1}^{2\eta}$ for a positive short rate. Our series are negative on part of the sample (the one-year OIS on $53\%$ of days, minimum $-0.76\%$; the one-year swap on $3\%$), so the level function is split at a centre, as in \eqref{eq:lgm-mult}; the centring at $2\%$ for the expected gap keeps the target asymmetry testable, which a power of the raw swap level could not. The mean is in Dickey--Fuller form, with the same likelihood as the AR(1) in levels; $\phi-1$ is not tested, since its distribution is non-standard under a unit root, while tests on $\eta_\pm$ are standard even if the level is integrated \citep{han2014asymptotic}. The \citet{brenner1996another} literal form, $\sigma_t^2=(\omega_h+\alpha\varepsilon_{t-1}^2+\beta_h\sigma_{t-1}^2)\Lambda_t$, and an additive form are estimated as checks.

\begin{table}[!htbp]\centering\small
\caption{Expected inflation and expected policy rate: estimates of \eqref{eq:lgm-mult}}
\label{tab:lgm-m23}
\resizebox{\linewidth}{!}{%
\begin{tabular}{lcc}
\toprule
 & Expected gap $d^e$ & Expected rate $i^e$\\
\midrule
$\hat\eta_+$; LR $\eta_+=0$ & $1.04$; $36.5$ ($<0.001$) & $1.36$; $19.4$ ($<0.001$)\\
$\hat\eta_-$; LR $\eta_-=0$ & $-0.08$; $0.2$ ($0.67$) & $-1.25$; $1.0$ ($0.31$)\\
LR $\eta_+=\eta_-$ & $26.5$ ($<0.001$) & $10.1$ ($0.002$)\\
$\hat\varpi$; LR $\varpi=0$ & $0.085$; $7.9$ ($0.005$) & $0.046$; $3.8$ ($0.051$)\\
$\hat\omega_h$, $\hat\alpha$, $\hat\beta_h$; degrees of freedom & $0.45$, $0.18$, $0.82$; $3.5$ & $0.004$, $0.10$, $0.90$; $3.7$\\
Ljung--Box $Q(12)$ on $z^2$, $p$ & $0.78$ & $0.73$\\
Log-likelihood: \eqref{eq:lgm-mult}; additive; literal & $-11{,}237.4$; $-11{,}240.6$; $-11{,}250.4$ & $-6{,}136.1$; $-6{,}134.5$; $-6{,}132.1$\\
\bottomrule
\end{tabular}}
\vspace{0.3em}
\parbox{\linewidth}{\footnotesize \textit{Notes:} daily, bp; the level in pp. Expected gap: 2009-10-15 to 2026-07-17, $4{,}171$ days plus $201$ roll days treated as missing. Expected rate: 2009-10-15 to 2025-12-31, $4{,}230$ days. The three forms have the same number of parameters.}
\end{table}

\paragraph{Reading.} For the expected gap, $\hat\eta_+\approx1$: one point of expected overshoot doubles the conditional standard deviation of daily revisions ($2^{1.04}=2.06$; variance $\times4.2$), and the effect is more than proportional ($\times9.8$ in variance at two points). The expected gap is positive on only $21\%$ of days, mostly in 2021--23, so the daily $n$ overstates the information; with a 2021:07--2023:12 dummy in $h_t$, $\hat\eta_+=0.75$ (likelihood ratio $17.8$); on the sample ending June 2021 ($d^e>0$ on $4\%$ of days), $1.05$, borderline ($p=0.099$). For the expected rate, the additive and literal forms fit slightly better (by $1.6$ and $4.0$ log-likelihood points), and the level effect is significant in all three (literal form: likelihood ratio $30.9$, with a smaller per-period $\eta$ because the level also feeds the persistent term); the multiplicative form is kept because it is that of the expected gap and of the two-way tests. On the sample ending June 2022, identified from the positive rates of 2009--11, $\hat\eta_+=2.23$; with a 2022:07--2023:12 hiking-phase dummy, $1.52$ ($p=0.006$). $\eta_-$ is poorly identified, since $1+(-i^e)_+$ ranges only from $1$ to $1.76$. With an expected rate of $1\%$ daily revisions are $2.6$ times larger than at zero; near and below zero the cumulated annual variance of revisions is about $0.016$ pp$^2$, a standard deviation of $13$ basis points a year.

\subsection{Units of the survey comparison}
\label{app:lgm:units}

Each model's variance is taken over the survey's one-year horizon. The market's daily variances are cumulated over $252$ trading days; the realized model's monthly variances are mapped into the variance of the next twelve months' average with its AR(6) weights $w_k$ ($\sum_k w_k^2=0.22$), times $1/(1-\hat\beta_h)$ in steady state. The cumulation of daily revisions assumes them serially uncorrelated, and they are not: the variance ratio \citep{lo1988stock} of daily changes in $d^e$ at one year is $2.9$ in 2013--19, near target, and $7.0$ in 2021--26 (full sample $6.2$, robust $z=10.5$). The market's one-year revision variance at target is therefore nearer $0.8$ than $0.28$ pp$^2$, about twice the survey's intercept. The market's response is convex rather than linear. Realized headline inflation gives $1.7$ on impact and $4.6$ once variance persistence has built up, two to five times the survey's slope; with oil controlled, $1.2$ to $3.4$ times; ex energy, about $0.6$, at or below it. The range of realized slopes across specifications runs from about $0.6$ (ex energy) to $4.6$ (headline, steady state); with lagged oil controls the augmented model gives about $3.1$ on impact and $3.8$ in steady state; with oil in Model~\eqref{eq:lgm-add}, $\hat\delta_+=4.6$ and $\hat\beta_h=0.64$ give $1.0$ and $2.9$; ex energy, $\hat\delta_+=2.8$ and $\hat\beta_h=0.04$ give $0.6$.

\subsection{Two-way Granger tests between the expected gap and the expected rate}
\label{app:lgm:twoway}

Direction A (rate to expected inflation) augments the expected-gap model with the rate; direction B (expected inflation to rate) augments the rate model with the expected gap. With $y$ the dependent series, $L$ the other series and $\mathcal U$ the level-free uncertainty of the other series from its univariate model,
\begin{equation}\label{eq:lgm-twoway}
\Delta y_t=c+(\phi-1)y_{t-1}+\theta_\Delta\,\Delta L_{t-1}+\varpi_{\mathcal U}\,\mathcal U_{t-1}+\varpi\,\sigma_t+\varepsilon_t,\qquad
h_t=\omega_h+\alpha\frac{\varepsilon_{t-1}^2}{\Lambda_{t-1}}+\beta_h h_{t-1}+\delta_{\mathcal U}\,\mathcal U_{t-1}^2,
\end{equation}
with $\sigma_t^2=h_t\Lambda_t$ and $\Lambda_t$ the own level factor; $4{,}227$ common days to 2025-12-31, previous-day values. The other series enters in changes so that every mean regressor is stationary and the tests are standard; $\delta_{\mathcal U}$ enters the level-free part, so its effect on $\sigma_t^2$ is scaled by $\Lambda_t$.

\begin{table}[!htbp]\centering\small
\caption{Two-way Granger tests: coefficient [likelihood-ratio $p$]}
\label{tab:lgm-twoway}
\begin{tabular}{lcc}
\toprule
Channel & A: rate $\to$ expected inflation & B: expected inflation $\to$ rate\\
\midrule
Change $\to$ change ($\theta_\Delta$) & $+0.030$ [0.32] & $+0.003$ [0.13]\\
Uncertainty $\to$ level ($\varpi_{\mathcal U}$) & $-0.09$ [0.20] & $-0.004$ [0.40]\\
Uncertainty$^2$ $\to$ variance ($\delta_{\mathcal U}$) & $+0.021$ [0.19] & $+0.0011$ [$<$0.001]\\
\quad LR for $\delta_{\mathcal U}=0$ & 0.7 & 25.3\\
Own $\eta_\pm=0$ (LR, $\chi^2_2$) & 27.5 & 44.7\\
\midrule
\multicolumn{3}{l}{\emph{Other series in levels, $L_{t-1}$ (pp)}}\\
Level $\to$ level, lag 1 & $-0.30$ [$<$0.001] & $+0.041$ [$<$0.001]\\
\quad \citet{toda1995statistical} & [0.51] & [0.18]\\
Uncertainty $\to$ level & $-0.15$ [0.041] & $-0.003$ [0.61]\\
LR for $\delta_{\mathcal U}=0$ & 0.8 & 26.4\\
\bottomrule
\end{tabular}
\vspace{0.3em}
\parbox{\linewidth}{\footnotesize \textit{Notes:} dependent in bp per day; $\theta_\Delta$ in bp per bp of the other series' previous-day change; levels in pp. The roll-day change of $d^e$ is missing as a dependent variable and set to zero as a regressor. Each channel alone gives the same picture.}
\end{table}

\paragraph{Reading.} In changes neither series predicts the other. In levels the lag-one effects carry the signs economics predicts---a higher expected rate lowers expected inflation ($-0.30$ bp per point, the credibility channel), higher expected inflation raises the expected rate ($+0.04$ bp per point, the reaction function)---but they vanish under \citet{toda1995statistical}: they reflect the joint persistence of two near-unit-root series. The significant effect of rate uncertainty on the expected-inflation level in levels ($p=0.041$) also loses significance with a dummy for negative expected rates ($p=0.24$). Rate uncertainty does not raise the variance of expected inflation. Expected-inflation uncertainty predicts a higher rate variance, and not only after 2022: with $\delta_{\mathcal U}$ split at July 2022 both parts are significant (before: likelihood ratio $28.7$; after: $11.9$; equality rejected, $p=0.001$), and on the sample ending June 2022 alone the likelihood ratio is $6.4$ ($p=0.006$). The level-free coefficient is larger before ($0.0015$ against $0.0002$), but after mid-2022 the rate's level factor is about $100$ against $2$, so the average contribution to the rate's conditional variance is larger after ($0.38$ against $0.07$ bp$^2$). Replacing $\sqrt{h^e}$ by the model-free standard deviation of the last $21$ daily changes of $d^e$ gives a likelihood ratio of $26.5$ (against $25.4$ for $\sqrt{h^e}$ in the same re-run).

\subsection{The augmented model}
\label{app:lgm:aug}

\paragraph{Construction.} $d^e$ is the monthly mean of the daily expected gap. $\mathcal U^e$ is the level-free uncertainty of the expected gap, the monthly mean of $\sqrt{h_t}$ over non-roll days (first and last three trading days excluded), rescaled to pp per month; its correlation with $d^e$ is $-0.10$, and it is not persistent (AR(1) $0.48$). Monthly $d^e$ is nearly a unit-root process (AR(1) $0.95$; augmented Dickey--Fuller $p=0.16$).

\paragraph{The constant.} The main specification puts no floor on $\omega_h$; $\omega_h\ge1$ is the robustness check. Profiling the full model over floors from $1.5$ to $0$, the likelihood rises smoothly to a plateau, $\hat\omega_h\to0$ and $\sqrt{h_t}$ stays above $0.74$ pp, because $\delta_{\mathcal U}(\mathcal U^e)^2>0$ and the level arms carry the baseline variance. In the full fit $\hat\omega_h=\hat\alpha=\hat\beta_h=0$: the conditional variance is a deterministic function of the predicted level and of $(\mathcal U^e)^2$, a location--scale model, which is why $\delta_{\mathcal U}$ and $\omega_h$---both slow-moving positive shifters---are not separately identified. The floor matters only for $\delta_{\mathcal U}$; the models without $\delta_{\mathcal U}$ or without $\theta_e$ have interior constants at any floor.

\begin{table}[!htbp]\centering\small
\caption{Augmented model \eqref{eq:lgm-aug}: each component alone and together (no floor; $\omega_h\ge1$ in brackets)}
\label{tab:lgm-aug}
\resizebox{\linewidth}{!}{%
\begin{tabular}{lcccc}
\toprule
 & Level only & Uncertainty in $\mu$ only & Uncertainty in $h$ only & Full\\
\midrule
$\hat\theta_e$ & $0.89$ ($0.038$) & -- & -- & $1.07$ ($0.005$) [$1.03$; $0.011$]\\
$\hat\varpi_{\mathcal U}$ & -- & $1.01$ ($0.64$) & -- & $2.71$ ($0.10$) [$2.63$; $0.13$]\\
$\hat\delta_{\mathcal U}$ & -- & -- & $21.0$ ($0.048$) & $26.9$ ($0.048$) [$12.9$; $0.11$]\\
Joint test & -- & -- & -- & $p=0.010$ [$0.018$]\\
$\hat\varpi$ (own, in mean) & & & & $0.31$ ($0.006$)\\
$\hat\delta_+$, $\hat\delta_-$ & $19.0$, -- & $7.4$, -- & $13.2$, -- & $17.1$, $1.92$\\
$\hat\alpha$, $\hat\beta_h$ & $0$, $0.05$ & $0$, $0.55$ & $0$, $0.47$ & $0$, $0$\\
AIC (base model $860.5$) & $858.2$ & $862.3$ & $859.7$ & $856.0$ [$857.2$]\\
\bottomrule
\end{tabular}}
\vspace{0.3em}
\parbox{\linewidth}{\footnotesize \textit{Notes:} 2009:11--2026:01, $n=195$; likelihood-ratio $p$-values, $\bar\chi^2$ for $\delta_{\mathcal U}$. $\theta_e$ in pp per pp; $\delta_{\mathcal U}$ in pp$^2$ per pp$^2$. Ljung--Box on $z^2$ in the full model, $p=0.89$. The likelihood is flat in $(\omega_h,\delta_{\mathcal U})$: $\delta_{\mathcal U}$ moves from $26.9$ to $12.9$ with the floor for $0.6$ log-likelihood points. With the monthly realized standard deviation of daily $d^e$ changes in place of the GARCH $\mathcal U^e$: headline $\hat\theta_e=0.94$ ($p=0.019$), $\hat\delta_{\mathcal U}=17.1$ ($p=0.19$); core $0.18$ ($p=0.068$), $\delta_{\mathcal U}=0$.}
\end{table}

\paragraph{Forced variance memory.} The full model's variance is memoryless, so $\alpha$ or $\beta_h$ is fixed on a grid with every other parameter free, and the full model, the model without $\theta_e$ and the model without $\delta_{\mathcal U}$ are re-estimated at each point. The admissible set (likelihood ratio against the free fit at most $3.84$) is $\alpha\in[0,0.15]$ and $\beta_h\in[0,0.2]$ without the floor, $\alpha\in[0,0.10]$ and $\beta_h\in[0,0.3]$ with it. $\theta_e$ stays significant at every admissible point ($p\le0.020$; $0.055$ at the edge $\beta_h=0.3$ with the floor), and close to one under typical GARCH memory that the data reject ($\alpha=0.10$, $\beta_h=0.80$: $1.25$, $p=0.006$). Its one weak region is $\beta_h=0.3$--$0.6$ with $\alpha=0$, also rejected, where pure variance persistence absorbs the overshoot months that $\theta_e$ otherwise explains in the mean ($0.44$--$0.84$, $p=0.055$--$0.16$). $\delta_{\mathcal U}$ does not survive: borderline under forced ARCH up to $\alpha=0.15$ ($p=0.050$--$0.061$) and at $\beta_h=0.1$, not from $\beta_h=0.2$ on ($p\ge0.10$), and with GARCH memory its estimate falls towards zero; with the floor it is never significant once memory is forced.

\begin{table}[!htbp]\centering\small
\caption{$\theta_e$ under persistent-regressor inference}
\label{tab:lgm-ivx}
\begin{tabular}{lcccc}
\toprule
Estimator & $\hat\theta_e$ & s.e. & asymptotic $p$ & bootstrap $p$\\
\midrule
\multicolumn{5}{l}{\emph{Headline}}\\
OLS, heteroskedasticity-robust (reference) & 1.30 & 0.63 & 0.039 & 0.055\\
Weighted by $1/\sqrt{h_t}$ (reference) & 0.96 & 0.34 & 0.004 & 0.009\\
IVX, $d^e$ alone & 1.43 & 0.30 & $<0.001$ & 0.007\\
IVX, $d^e$ and $\mathcal U^e$ & 1.46 & 0.30 & $<0.001$ & 0.007\\
IVX with controls & 1.89 & 0.66 & 0.004 & 0.005\\
IVX with controls, $\mathcal U^e$ also IVX & 1.90 & 0.66 & 0.004 & 0.005\\
IVX with controls, weighted & 1.65 & 0.50 & 0.001 & 0.003\\
\multicolumn{5}{l}{\emph{Core}}\\
OLS & 0.69 & 0.19 & $<0.001$ & 0.001\\
Weighted & 0.34 & 0.15 & 0.029 & 0.087\\
IVX with controls; weighted & 0.70; 0.37 & 0.19; 0.17 & $<0.001$; 0.035 & 0.001; 0.10\\
IVX with controls and oil; weighted & 0.75; 0.41 & 0.20; 0.17 & $<0.001$; 0.016 & 0.001; 0.073\\
\bottomrule
\end{tabular}
\vspace{0.3em}
\parbox{\linewidth}{\footnotesize \textit{Notes:} linear mean equation, $n=194$. IVX instruments $d^e_{t-1}$ with $\tilde z_t=\sum_{j\ge0}\rho_z^j\Delta d^e_{t-j}$, $\rho_z=1-1/n^{0.95}$ \citep{phillips2009econometric,kostakis2015robust}, with the finite-sample intercept correction and a heteroskedasticity-robust middle matrix (the original homoskedastic form overstates the statistic here, since the error variance moves with the level). Controls: the six AR lags of $d_t$ and $\mathcal U^e_{t-1}$. ``Weighted'' divides every variable by the conditional standard deviation of the full model, mimicking the likelihood's weighting. Bootstrap: $9{,}999$ draws under $\theta_e=0$, regressors regenerated from their fitted AR(1) with one wild Rademacher draw per month on the joint residuals. The innovation correlation between the $d_t$ residual and the $d^e$ innovation is $0.02$, so the \citet{stambaugh1999predictive} distortion is negligible.}
\end{table}

\paragraph{Energy, core and the episode.} The one-year swap is priced partly off the energy futures curve, and euro inflation swaps reference HICP excluding tobacco with a three-month indexation lag, so a headline $\theta_e$ near one could be anticipated energy pass-through. With last month's Brent price in euros, the euro--dollar rate and energy inflation as controls, $\hat\theta_e=0.73$ ($p=0.076$) and $\hat\delta_{\mathcal U}=0.01$; adding the realized oil move of month $t$, which removes energy noise from $d_t$, restores it ($0.93$). With 2021--23 dummies in the mean and the variance, $1.18$; before 2021, $1.20$ but not significant ($p=0.31$, $n=134$). On core, $0.24$ ($p=0.004$), $0.33$ with the oil controls, $0.19$ ($p=0.10$) with the episode dummies; ex energy, $0.44$ ($p=0.001$). $\delta_{\mathcal U}$ is zero in every specification with oil controls and on core and ex-energy inflation. Unweighted core estimates (about $0.7$) give full weight to the large 2022--23 months.

\paragraph{Out of sample.} Recursive OLS from 2009:11, one-month-ahead forecasts from 2016:01, AR(6) with and without $d^e_{t-1}$. MSPE ratios (one-sided \citealp{clark2007approximately} $p$): headline $0.974$ ($0.11$) over 2016--2026, $0.913$ ($0.021$) in 2016--2020, $0.984$ ($0.18$) in 2021--2026; with oil controls $0.993$ ($0.23$); core $0.906$ ($0.005$), $0.992$ ($0.11$), $0.866$ ($0.003$); with oil controls $0.905$ ($0.009$), $0.971$ ($0.047$), $0.874$ ($0.009$). Headline inflation is dominated by energy noise that the swap does not anticipate one month ahead.

\paragraph{Sizes.} With the AR(6) dynamics ($\sum\hat\phi_i=0.14$), the long-run effect of a permanent one-point shift in $d^e$ is $\hat\theta_e/(1-\sum\hat\phi_i)=1.24$ points; a one-standard-deviation move in $d^e$ ($1.26$ points) moves predicted realized inflation by $1.35$ points on impact. On average $\hat\delta_{\mathcal U}(\mathcal U^e)^2$ contributes $0.94$ pp$^2$ [$0.45$ with the floor] to the conditional variance of realized inflation, $10\%$ [$5\%$] of its average and $28\%$ [$12\%$] on average month by month; a one-standard-deviation rise in $\mathcal U^e$ adds about $4\%$ [$2\%$] to the conditional standard deviation. With $\hat\omega_h=0$ the term takes over the constant, so these shares are upper readings, and they are predictive magnitudes, not the effect of an exogenous rise in uncertainty.

\subsection{Oil-specific uncertainty: construction, checks and the full set of outcomes}
\label{app:lgm:oil}

\paragraph{The daily news.} With $v_t=\Delta\log\mathrm{OVX}_t$, the daily change in CBOE crude-oil implied volatility (thirty-day WTI options, from May 2007), $o_t$ is the residual of
\[
v_t=c+\sum_{k=0}^{2}\sum_{j\in\{\mathrm{WTI},\mathrm{Brent}\}}\bigl(\varsigma^+_{jk}\,r^+_{j,t-k}+\varsigma^-_{jk}\,r^-_{j,t-k}\bigr)+\sum_{k=0}^{1}\Gamma_k'\,\Delta \mathbf{G}_{t-k}+e_t,
\]
where $r^\pm_j$ are the positive and negative parts of the daily oil return and $\mathbf G_t$ collects log VIX, log VXV, log GVZ (gold implied volatility) and the log seven-day mean of daily US economic policy uncertainty, asset classes in which the euro area is a price-taker. The residual is prewhitened (AR(5)), euro-area release days are set to missing---on HICP flash days OVX moved with the HICP surprise (joint $p=0.07$)---and the series is standardized; the sample starts in June 2008 with GVZ. The return arms explain $21\%$ of OVX news and the global block another $12$ points.

\paragraph{Variants.} The broadest purge (2011:03--2025:12) adds emerging-market equity volatility, the ECB's composite indicator of systemic stress and daily geopolitical risk \citep{caldara2022measuring}, and further purges the monthly shock of the European and global policy-uncertainty indices; purging geopolitical risk is conservative, since it may remove genuine supply uncertainty. Shocks built from returns only, and from returns and the VIX, correlate $0.92$ and $0.99$ with $o_t$. The global factor $f^G_t$ is the first principal component of the daily changes in log VIX, VXV and GVZ ($74\%$ of their variance); $\mathrm{corr}(o_t,f^G_t)=-0.01$ by construction, while raw OVX news correlates $0.36$ with $f^G_t$. For the variance equations the state $\tilde O_t$ is the residual of log OVX on the log global block, signed $5$-, $20$- and $60$-day oil returns and the oil price's gap to its $250$-day mean, lagged one day. Alternatives to option-implied volatility are realized GARCH \citep{hansen2012realized} and stochastic volatility, which give volatility its own innovation; we keep the option measure, which is model-free and daily.

\paragraph{The projections.} The monthly shock is the sum of $o_t$ over the month. Each of the twenty-four measures is projected on it with lags of the shock, of the outcome, of Brent and of the VIX, and with the oil supply-news shock of \citet{kanzig2021macroeconomic} as a control; the statistic is the average response over the first year, with Holm's correction across the measures \citep{holm1979simple}. The global benchmark $f^G_t$ is scaled so that it moves the OVX as much as $o_t$ does.

\paragraph{Exogeneity checks.} On Governing Council days $o_t$ does not load on the monetary-event-window surprises (overnight swaps at one month and two years, the ten-year Bund; joint $p=0.56$), and its variance is the same as on other days (Levene $p=0.89$). Lags one to five of the daily changes in $d^e$, the one-year OIS and the option-implied standard deviation do not predict it ($p=0.67$ and $0.76$). An earlier version of the shock is robust to removing the months of US monetary shocks \citep{jarocinski2020deconstructing}. The OVX closes at 16:15 New York time, after the euro-area close, so the daily responses start at $h=1$. A one-standard-deviation monthly $o_t$ raises log OVX by $6\%$ on impact and $12.6\%$ after a year; Brent rises by $1.6\%$ on impact and falls by $7.3\%$ after a year ($p=0.005$; $-2.7\%$ on average over the year), the precautionary-demand pattern.

\begin{table}[!htbp]\centering\footnotesize\setlength{\tabcolsep}{3.5pt}
\caption{Average response over the first year to a one-standard-deviation shock: all measures}
\label{tab:lgm-oil-full}
\begin{tabular}{llcccc}
\toprule
& Measure & Oil-specific $o_t$ & MDE$_{80}$ & Global $f^G_t$ & Oil supply news\\
\midrule
1 Inflation level & HICP headline (\%) & $-0.06$ (0.34) & 0.19 & $+0.05$ (0.60) & $+0.37$ (0.004)\\
 & HICP core (\%) & $-0.07$ (0.34) & 0.20 & $+0.07$ (0.21) & $+0.05$ (0.45)\\
 & HICP energy (\%) & $+0.02$ (0.96) & 0.95 & $-0.40$ (0.39) & $+2.76$ (0.000)\\
2 Inflation volatility & $|\hat e|$ headline (pp; mean 0.19) & $+0.02$ (0.046) & 0.03 & $+0.01$ (0.053) & $-0.00$ (0.85)\\
 & $|\hat e|$ core (pp; mean 0.16) & $-0.01$ (0.28) & 0.02 & $0.00$ (0.44) & $-0.00$ (0.82)\\
3 Expected inflation & one-year swap (bp) & $-4.4$ (0.35) & 13.2 & $-3.0$ (0.62) & $+27.9$ (0.008)\\
 & 1y1y swap & $-2.2$ (0.22) & 5.0 & $-4.9$ (0.043) & $+11.4$ (0.001)\\
 & 5y5y swap & $-2.8$ (0.018) & 3.3 & $-2.1$ (0.025) & $+5.6$ (0.002)\\
4 Its uncertainty & option s.d., one year & $+0.8$ (0.65) & 4.7 & $+4.1$ (0.021) & $+0.3$ (0.95)\\
 & option s.d., two years & $+0.5$ (0.71) & 3.7 & $+2.6$ (0.031) & $-0.0$ (0.99)\\
 & GARCH s.d.\ of daily $\Delta$ swap & $+0.4$ (0.18) & 0.8 & $+0.4$ (0.24) & $+0.6$ (0.32)\\
5 Policy rate & DFR & $-0.1$ (0.98) & 13.6 & $-9.0$ (0.067) & $+11.3$ (0.14)\\
 & three-month Euribor & $-0.1$ (0.98) & 17.6 & $-6.7$ (0.29) & $+12.1$ (0.22)\\
6 Its variance & $|\Delta\mathrm{DFR}|$ within month & $+0.3$ (0.71) & 2.5 & $+2.1$ (0.034) & $-0.1$ (0.95)\\
 & $|\Delta$Euribor$|$ & $+0.1$ (0.88) & 2.2 & $+2.4$ (0.043) & $-1.8$ (0.41)\\
7 Expected policy rate & one-year OIS & $-1.3$ (0.80) & 14.2 & $-4.0$ (0.59) & $+10.7$ (0.11)\\
 & 1y1y OIS & $-1.3$ (0.72) & 10.4 & $+1.6$ (0.80) & $+12.2$ (0.032)\\
8 Its uncertainty & realized s.d.\ of daily $\Delta$OIS & $+0.05$ (0.57) & 0.25 & $+0.24$ (0.028) & $+0.19$ (0.25)\\
 & GARCH s.d.\ of daily $\Delta$OIS & $+0.06$ (0.45) & 0.23 & $+0.25$ (0.031) & $+0.09$ (0.64)\\
\bottomrule
\end{tabular}
\vspace{0.3em}
\parbox{\linewidth}{\footnotesize \textit{Notes:} as Table~\ref{tab:lgm-oil}; bp unless stated. MDE$_{80}=2.8\times$ s.e., the effect detectable with $80\%$ power at $5\%$. $|\hat e|$ is the absolute residual of an AR(12) with month dummies on month-on-month HICP inflation; option s.d.\ is the risk-neutral standard deviation of one- and two-year inflation from caps and floors; GARCH s.d.\ is from an AR(1)--GARCH(1,1)-$t$ on daily changes. For $o_t$ no Holm-adjusted $p$ falls below $0.43$, for $f^G_t$ none below $0.50$. Results on the broadest purge and the subsamples agree.}
\end{table}

\paragraph{By outcome.} The confidence intervals for the inflation level exclude effects of the size of the supply-news shock's. Headline innovations become slightly larger ($+0.02$ pp, about $10\%$ of their mean; Holm $1.0$), oil-price volatility passing through to energy prices; core does not move. The 5y5y swap falls by $2.8$ bp, but the fall disappears without 2020 ($p=0.39$) and halves net of the realized Brent path ($p=0.12$); in daily projections the one-year swap rises by $0.7$ bp after one day and $1.3$ after five ($p<0.01$), robust to the oil path and outside 2020--23, and reverts within ten to twenty days---a transient move in short inflation compensation, not a revision of expectations, since the 5y5y, the survey mean and option-implied uncertainty do not follow. Option-implied inflation uncertainty does not respond, monthly or daily at any horizon to sixty days; in the daily model \eqref{eq:lgm-mult} with $\tilde O_{t-1}$ as variance shifter, one standard deviation of oil uncertainty raises the variance of the expected gap by $0.3\%$ ($p=0.32$). In the survey, forecasters' density widths rise with $\tilde O$ only through 2021--23 (panel with forecaster effects: $+0.021$, $p=0.09$; outside the episode $-0.001$, $p=0.95$); the one trace is the skew, which tilts upward with oil-specific uncertainty ($p=0.013$) and downward with global uncertainty ($p=0.005$). Policy rates, their variance and the expected rate do not respond; in the daily model on the one-year OIS the variance effect is $0.1\%$ ($p=0.83$), and the monthly $95\%$ interval for $o_t$ on its GARCH standard deviation, $[-0.10,+0.22]$ bp, excludes the global factor's $+0.25$.

\paragraph{Two further designs.} As an exogenous moderator of the response of $d^e$ to HICP surprises, oil uncertainty fails its pre-test: the pre-release change in oil uncertainty predicts upside headline surprises. Dated two weeks earlier, it suggests a fragile upward asymmetry on headline news only ($p=0.007$; $0.04$--$0.11$ in real time). A proxy-SVAR and local-projection IV with two OPEC-announcement-day instruments, for the first and the second moment, is weak: first-stage $F=5.2$ and $2.0$, conditional $F$ of \citet{sanderson2016weak} $0.55$.

\paragraph{What the global factor would need.} Separating its uncertainty effect from its first-moment effects would require a second instrument that controls for those first moments, in the logic of \citet{piffer2018identifying}. Implied volatility on rate options and Brent risk reversals, the latter for the skew, would sharpen the tests; neither is in our data.

\end{document}